\documentclass[10pt]{article}
\usepackage{dcolumn}
\usepackage{bm}
\usepackage{hyperref}
\usepackage{verbatim}       

\usepackage{color}
\usepackage{graphicx}
\usepackage{array}

\usepackage{amsthm}
\usepackage{amssymb}

\usepackage{amstext}
\usepackage{psfrag}
\usepackage{amsmath}
\usepackage{amsfonts}
\usepackage{units}
\usepackage{mathrsfs}

\newcommand{\p}{\partial}

\newcommand{\dd}{{\rm d}}

\newtheorem{theorem}{Theorem}[section]
\newtheorem{corollary}[theorem]{Corollary}
\newtheorem{lemma}[theorem]{Lemma}
\newtheorem{proposition}[theorem]{Proposition}
\theoremstyle{definition}
\newtheorem{definition}[theorem]{Definition}
\theoremstyle{remark}
\newtheorem{remark}[theorem]{Remark}
\newtheorem{example}[theorem]{Example}

\newtheorem*{claim*}{Claim}  

\begin{document}

\title{\Large A Cauchy Busemann function approach to the geometry\\ of cosmological spacetimes with applications\\ to Bartnik's conjecture}

\author{E. Minguzzi\footnote{Dipartimento di Matematica, Universit\`a degli Studi di Pisa,  Via
B. Pontecorvo 5,  I-56127 Pisa, Italy. E-mail:
ettore.minguzzi@unipi.it, ORCID:0000-0002-8293-3802}}

\date{}
\maketitle

\begin{abstract}
\noindent We study the geometry of future timelike complete spacetimes admitting a compact Cauchy hypersurface. Future timelike completeness is compatible with an initial singularity and disregards the incompleteness expected inside black holes, so it is a reasonable large scale idealization for a Universe without final singularity. Given a compact Cauchy hypersurface $\Sigma$, we introduce the future Cauchy Busemann function $b^+$. Its finiteness domain, the future cosmological development, and its boundary, the future cosmological horizon, do not depend on $\Sigma$. For spacetimes with an initial singularity, namely with regular cosmological time, we obtain a canonical Cauchy Busemann function and a canonical calibrating semi-flow. Using a product trick we show that $b^+$ is the graph function of a Cauchy horizon in $M\times \mathbb{R}$: $b^+$ is locally Lipschitz and locally semi-concave on the future development, the future horizon is a $C^0$ null hypersurface, and through every point of the development passes a calibrated timelike ray, an observer who ages maximally into the far future. We give a foliation-free formula for $b^+$ as a minimal asymptotic clock defect. Under the timelike convergence condition (TCC) the level sets of $b^+$ are mean convex in the support sense. We apply these results to Bartnik's splitting conjecture for timelike complete spacetimes satisfying the TCC. If the future and past cosmological developments intersect, which happens whenever $b^+$ is $C^1$ on its finiteness domain, then the spacetime contains a timelike line and hence splits. Therefore any counterexample contains a non-empty closed causally convex set, the cosmological wall, separating the two developments; some component of the future horizon contains a lightlike line, and if the wall has empty interior it coincides with both horizons and is ruled by lightlike lines.
\end{abstract}

\newpage
\tableofcontents

\section{Introduction}

\subsection{Setting and motivation}

This work studies the geometry of spacetimes $(M,g)$ which admit a compact Cauchy hypersurface and which are future timelike complete, namely such that every future inextendible timelike geodesic is future complete.

In the title, by a {\em cosmological spacetime} we simply mean a globally hyperbolic spacetime admitting a compact Cauchy hypersurface. This meaning is broader than the original one, due to Bartnik \cite{bartnik88}, according to which a cosmological spacetime is a globally hyperbolic spacetime with compact Cauchy hypersurfaces which satisfies the timelike convergence condition (TCC), $Ric(v,v)\ge 0$ for every timelike vector $v$. In this work the TCC is not a standing assumption; it is used only in some results, in which it is explicitly assumed.

Future timelike completeness is a reasonable assumption for the study of the large scale geometry of the Universe. It is compatible with an initial singularity, which is expected in expanding models: by Hawking's singularity theorem \cite{hawking67}, a spacetime which satisfies the TCC and admits a compact Cauchy hypersurface of everywhere positive future mean curvature is past timelike geodesically incomplete. Future timelike completeness excludes a final singularity, such as the big crunch of recollapsing models, and it disregards the future timelike incompleteness expected inside black holes, where the observers which cross the event horizon reach the singularity in a finite proper time. Neglecting this incompleteness is an approximation which is acceptable for a study at the cosmological scale, where black holes are small scale features, much as, for the local study of a black hole, it is customary and acceptable to describe it through an asymptotically flat spacetime, rather than through a metric which matches a realistic cosmological model at large distances. The class of future timelike complete spacetimes with compact Cauchy hypersurface includes, for instance, the warped products $-\dd t^2+a(t)^2h$ with compact fiber and suitable scale factor $a$, a class which contains the non-contracting Friedmann--Lema\^itre--Robertson--Walker models with compact spatial sections (Example \ref{exwarp}). The study of these spacetimes under future timelike completeness alone, which has already been considered in connection with CMC Cauchy hypersurfaces \cite{galloway18c}, seems interesting in its own right, and the objects introduced here provide a framework for it.

\subsection{The Cauchy Busemann function}

Let $\Sigma$ be a compact Cauchy hypersurface. Future timelike completeness guarantees that the sets $\Sigma_k$ of points at Lorentzian distance $k>0$ from $\Sigma$ are compact Cauchy hypersurfaces (Theorem \ref{kvvb}), and we introduce the future Cauchy Busemann function $b^+=\lim_k [k-d(\cdot,\Sigma_k)]$ and, under past timelike completeness, its time dual $b^-$. The key observation, which comes from the product trick \cite[Sec.\ 5.2]{minguzzi18b} \cite{minguzzi17}, is that in the product spacetime $\tilde M=M\times \mathbb{R}$, $\tilde g=g+\dd z^2$, the past domain of dependence of a suitable set is the strict hypograph of $b^+$, so that the graph of $b^+$ is a Cauchy horizon. The regularity theory of horizons and the dynamics of their generators thus become available for $b^+$. The function $b^+$ may take the value $+\infty$, and its finiteness domain $\textrm{dom}\, b^+$, the {\em future cosmological development}, is an open past set, whose boundary $\p\, \textrm{dom}\, b^+$ we call the {\em future cosmological horizon}; dually, $\textrm{dom}\, b^-$, the {\em past cosmological development}, is an open future set, bounded by the past cosmological horizon. The $b^+$-calibrated curves, along which $b^+$ increases by the Lorentzian length, are the observers who age maximally into the far future, and whose clocks asymptotically agree with the cosmic time $d_\Sigma$ (Remark \ref{remasym}). For the warped products mentioned above, $b^+$ is smooth on the whole spacetime and its calibrated curves are the comoving observers (Example \ref{exwarp}).

\subsection{Results under future timelike completeness}

We list our main results which require only future timelike completeness, and, where indicated, the TCC; the precise assumptions are stated in each result (see also Section \ref{secconv}).
\begin{itemize}
\item[(a)] {\em Intrinsic objects.} The future cosmological development and the future cosmological horizon do not depend on $\Sigma$: the Cauchy Busemann functions of two compact Cauchy hypersurfaces have the same finiteness domain, on which they differ by a bounded function (Theorem \ref{xxxt}). The development meets every Cauchy hypersurface, contains the chronological past of every future $\Sigma$-ray, and is the union of the chronological pasts of the $b^+$-calibrated curves, which are TIPs (Propositions \ref{cmol} and \ref{coruni}). The horizon is a $C^0$ future null hypersurface (Proposition \ref{propbdry}) which, much as the event horizon of a black hole, can be crossed by future directed causal curves only in the outward direction. Moreover, without any completeness assumption, every TIP contains the chronological past of a future $\Sigma$-ray (Theorem \ref{cbdc}).
\item[(b)] {\em Spacetimes with initial singularity.} If the cosmological time of Andersson, Galloway and Howard \cite{andersson98} is regular, the normalized Lorentzian distances to its level sets converge to a canonical Cauchy Busemann function $b^{\textnormal{\textsc{m}}}$, which does not depend on any choice and is characterized as the minimal asymptotic delay of the clocks of the observers with respect to the cosmological time. Its calibrated rays define a canonical calibrating semi-flow on the set of differentiability points of $b^{\textnormal{\textsc{m}}}$, which has full measure in the future development $\textrm{dom}\, b^{\textnormal{\textsc{m}}}$, possibly a proper subset of $M${, and whose complement in $\textrm{dom}\, b^{\textnormal{\textsc{m}}}$ meets every Cauchy hypersurface in an $\mathcal{H}^n$-null set} (Section \ref{seccosm}, Proposition \ref{propcosm}).
\item[(c)] {\em Regularity and dynamics.} The function $b^+$ is continuous, and locally Lipschitz and locally semi-concave on $\textrm{dom}\, b^+$ (Theorem \ref{xxhg}). Its calibrated curves are future complete timelike rays, at least one of them starts from every point of $\textrm{dom}\, b^+$, and $b^+$ is differentiable precisely at the points from which exactly one of them starts (Proposition \ref{propcrease}), the set of non-differentiability points being $\mathcal{H}^{n+1}$-null {and meeting every achronal set, for instance every Cauchy hypersurface, in an $\mathcal{H}^n$-null set} (Theorem \ref{thmfine}). The level sets of $b^+$ are partial Cauchy hypersurfaces (Proposition \ref{cllp}), and under timelike completeness those of non-negative level are Cauchy horospheres in the sense of Galloway and Vega (Corollary \ref{corhoro}). The function $b^+$ admits a foliation-free formula as a minimal asymptotic clock defect (Theorem \ref{thmintr}), and the distances to the level sets $\Sigma_k$ are locally eventually equi-Lipschitz on $\textrm{dom}\, b^+$ (Theorem \ref{thmequi}), but not near its boundary (Corollary \ref{coreqlip}).
\item[(d)] {\em Mean convexity and rigidity under TCC.} Under the TCC the level sets of $b^+$ are mean convex in the support sense, namely the semi-flow of the calibrated rays has non-negative expansion in the support sense (Theorem \ref{thmmean} and Corollary \ref{cormeanconv}). This gives a localized rigidity version of Hawking's singularity theorem: if a smooth compact Cauchy hypersurface $\Sigma$ has non-positive future mean curvature on $\Sigma\cap \textrm{dom}\, b^+$, then its future $J^+(\Sigma)$ is isometric to a static product $([0,+\infty)\times \Sigma,-\dd t^2+h)$ (Theorem \ref{thmhawk}). Thus, unless its future is static, every smooth compact Cauchy hypersurface has positive future mean curvature somewhere in the future development.
\end{itemize}
In a companion paper \cite{minguzzi26g} we show that the product trick also turns the area theorem for horizons into a volume theorem for the semi-flow of the calibrated rays, which under the TCC gives the monotonicity of the volume measured by the calibrated observers, and that, under the TCC, the future cosmological horizon satisfies an area theorem.

\subsection{Applications to Bartnik's conjecture}

The singularity theorems of general relativity show that, under suitable energy and causality conditions, spacetimes are typically geodesically incomplete. For spacetimes with compact Cauchy hypersurfaces the relevant result is Hawking's theorem \cite{hawking67}: a spacetime which satisfies the TCC and admits a compact Cauchy hypersurface of everywhere negative future mean curvature is future timelike geodesically incomplete.

Every such theorem raises the question of what happens in the borderline case, namely of which spacetimes manage to be geodesically complete while still satisfying its energy condition. Yau proposed to study this rigidity problem systematically \cite{yau82}, and in several instances the answer takes the same form: completeness forces the spacetime to be a static product, the model case in which the focusing predicted by the energy condition does not occur at all: the static observers form a congruence of timelike geodesics with vanishing expansion and shear, along which $Ric(\p_t,\p_t)=0$, so that no conjugate points ever develop. The best known instance is the Lorentzian splitting theorem of Eschenburg, Galloway and Newman \cite{eschenburg88,galloway89,newman90}. It assumes, besides TCC, the existence of a complete timelike line, that is, of an inextendible timelike geodesic, defined for all values of its proper time, each of whose segments maximizes the Lorentzian distance between its endpoints, and concludes that the spacetime splits isometrically as a product $(\mathbb{R}\times \Sigma, -\dd t^2+h)$.

A timelike line is, however, a strong and non-local assumption, and in a spacetime with compact Cauchy hypersurfaces there is no reason for one to exist. Bartnik's conjecture asks whether, in that setting, the compactness of a Cauchy hypersurface can replace it: completeness alone should then force the same product structure.

\begin{quote}
{\bf Bartnik conjecture} \cite{bartnik88}: Every  spacetime $(M,g)$ admitting a compact Cauchy hypersurface and satisfying the timelike convergence condition either is timelike geodesically incomplete or  splits into a product $\mathbb{R}_t \times \Sigma_x$, $g=-\dd t^2+ h$ where $(\Sigma,h)$ is a compact Riemannian space.
\end{quote}

Bartnik himself proved the conjecture under the additional assumption that the spacetime admits a compact constant mean curvature (CMC) Cauchy hypersurface \cite{bartnik88}, and he related the conjecture to the problem of existence of such hypersurfaces; the existence and non-existence of CMC Cauchy hypersurfaces in this setting are reviewed in \cite{galloway19}, and vacuum spacetimes with compact Cauchy hypersurfaces and without CMC Cauchy hypersurfaces are now known \cite{ling25}. Several further partial results have been obtained by adding assumptions which, directly or indirectly, allow one to construct a timelike line: the absence of observer horizons \cite{galloway84b,galloway19}, or of null rays \cite{lesourd18b}; the existence of a future and a past ray with intersecting chronological past and future \cite{galloway97}; the max--min condition of Galloway and Vega \cite{galloway14b}, and their condition on the relative position of past and future Cauchy horospheres \cite{galloway17}; non-positive timelike sectional curvature \cite{ehrlich90} (under this curvature condition and future timelike completeness alone, compact CMC Cauchy hypersurfaces exist \cite{galloway18c}), also in the synthetic setting of Lorentzian length spaces \cite{flores24}; conformal symmetries \cite{costa20b}; a spacelike past or future causal boundary \cite{galloway17}; and the min-min condition on the minimal elements of the causal boundary \cite{heveling25}. Lorentzian Busemann functions and their level sets, the horospheres, have been a central tool in these investigations \cite{galloway96,sharifzadeh09,amini18,galloway14b}. Some care is needed with approximation arguments. The local equi-Lipschitz property of the Lorentzian distances from the points of a timelike line is an essential ingredient of the proofs of the Lorentzian splitting theorem, but, as shown by Galloway, McCann and Ohanyan \cite{galloway26}, the analogous property fails in general for the Lorentzian distances to the level sets of a Cauchy temporal function.

A basic obstruction to many strategies is the example of Ehrlich and Galloway \cite{ehrlich90}: a causally geodesically complete globally hyperbolic spacetime with compact Cauchy hypersurface and no timelike line. Due to it, it is futile to try to prove the existence of a timelike line without using the timelike convergence condition, to only then invoke the splitting theorem. Of course, their example does not satisfy TCC.

The partial results recalled above add hypotheses which, together with TCC, lead to a timelike line. Here we adopt a complementary point of view: without additional hypotheses, we identify the global structure that a counterexample must have. In this way, if Bartnik's conjecture is false, the search for a counterexample will be aided by the constrained causal structure that we determine. Moreover, any counterexample will necessarily display the structures predicted here, which, by (a) above, do not depend on auxiliary choices, such as that of a Cauchy hypersurface, so they are intrinsic to the spacetime and carry physical significance. If, on the other hand, the conjecture is true, knowing the causal structure would most likely help the development of a proof.

Our first main result in this direction is that if the two developments intersect, the spacetime contains a timelike line, hence it splits under TCC. Bartnik's conjecture is thus equivalent to the statement that, for timelike complete spacetimes with compact Cauchy hypersurfaces satisfying the TCC, the two developments always intersect. If the conjecture is false, the developments are disjoint and, $M$ being connected, the closed set
\[
W:=M\backslash (\textrm{dom}\, b^+\cup \textrm{dom}\, b^-)
\]
is non-empty. We call it the {\em cosmological wall}. The last part of this work is devoted to its structure. The picture that emerges, for thin walls, is that of Figure \ref{figstruct}, which in its simplest realization in dimension 1+1 reproduces the causal structure of the Ehrlich--Galloway example.

\begin{figure}[ht]
\begin{center}
 \includegraphics[width=12cm]{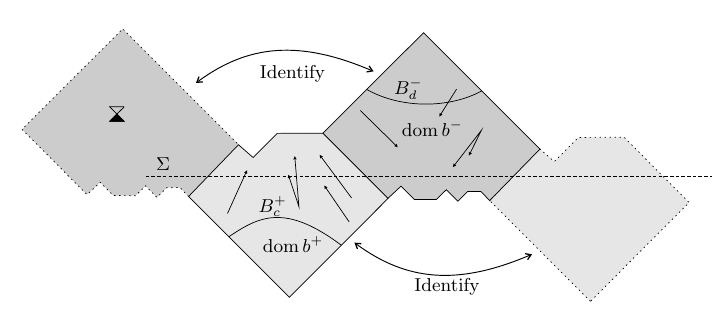}
\end{center} \caption{A schematic configuration compatible with the structure of a counterexample to Bartnik's conjecture established in this work, in the case of a thin wall, which coincides with the two cosmological horizons.  The existence of calibrating rays is suggested by arrows (though recall that we parametrize past rays so that they are future directed). The calibrating rays and the horospheres may depend on $\Sigma$, but the other elements displayed do not. The equispaced level sets $B^+_c$ of $b^+$ accumulate on $\p\, \textrm{dom}\, b^+$, since $b^+$ diverges to $+\infty$ approaching this set. Still, causal curves might avoid level sets of $b^+$ in the past direction. Note that this figure respects the causal structure constraint of Theorem \ref{cbdc}: every TIP contains a future $\Sigma$-ray necessarily entirely contained in $\textrm{dom}\, b^+$.} \label{figstruct}
\end{figure}

Our results on Bartnik's conjecture concern timelike complete spacetimes, and the precise assumptions are stated in each result.
\begin{itemize}
\item[(e)] {\em Splitting criteria.} If $\textrm{dom}\, b^+\cap \textrm{dom}\, b^-\ne \emptyset$ the function $b^+-b^-$ attains a minimum and a timelike line passes through each minimum point (Theorem \ref{notr}). This happens iff $2k-d(\Sigma_{-k},\Sigma_k)$ is bounded (Theorem \ref{thmscal}), and it happens if $b^+$ is $C^1$ on $\textrm{dom}\, b^+$ (Theorem \ref{thmcone}), namely if from every point of $\textrm{dom}\, b^+$ starts a unique calibrated ray (Proposition \ref{propcrease}). Several known splitting results follow as corollaries. Under TCC, if the two developments intersect, they fill the whole spacetime (Corollary \ref{corint}). By (d), a smooth compact Cauchy hypersurface with non-positive future mean curvature on the future development forces splitting (Theorem \ref{thmhawk}); this contains Bartnik's CMC splitting theorem (Corollary \ref{corhawk}). Local eventual equi-Lipschitz bounds for the distances to the sets $\Sigma_k$ near $\Sigma$ are equivalent to the conjecture (Corollary \ref{coreqlip}), and a global $C^1$ solution of the eikonal equation produces timelike lines (Proposition \ref{propeik}). Finally, there is a one-sided criterion which requires only future timelike completeness: if $b^+-d_\Sigma$ is bounded along some past inextendible curve, there is a $b^+$-calibrated timelike line, possibly past incomplete (Theorem \ref{thmonesided}). Under timelike completeness this condition is equivalent to the intersection of the two developments. As a consequence, every future timelike complete spacetime with compact Cauchy hypersurfaces and regular cosmological time contains $b^{\textnormal{\textsc{m}}}$-calibrated timelike lines (Corollary \ref{corbigbang}), while in a counterexample to Bartnik's conjecture $b^+-d_\Sigma$ diverges at past infinity (Corollary \ref{corpastdiv}).
\item[(f)] {\em The cosmological wall.} If the spacetime does not split, some component of the future horizon contains a lightlike line; a component contains a lightlike line if and only if it meets every Cauchy hypersurface, while every other component lies in the future of some Cauchy hypersurface, every Cauchy time function attains a minimum on it, and its chronological future is contained in the wall (Theorem \ref{thmwall}); every point common to the two horizons lies on a lightlike line contained in both; a wall with empty interior coincides with both horizons, it is ruled by lightlike lines and separates every Cauchy hypersurface, and under null completeness it is a smooth totally geodesic null hypersurface (Theorem \ref{thmthin}); the TIFs generated by the past inextendible timelike curves of $\textrm{dom}\, b^+$ are never minimal, which suggests that its past causal boundary is null (Proposition \ref{proppastnull}). Every complete timelike geodesic has infinite clock defect in the future or in the past (Corollary \ref{cordefect}).
\end{itemize}

The paper is organized as follows. Section \ref{secdist} studies the distance from a compact Cauchy hypersurface. Section \ref{secbuse} introduces the Cauchy Busemann function and proves that its finiteness domain does not depend on the Cauchy hypersurface, Section \ref{secprod} presents the product trick, and Section \ref{secreg} the regularity theory, including the canonical calibrating semi-flow of spacetimes with an initial singularity (Section \ref{seccosm}). Section \ref{sechoro} studies the horospheres and their mean curvature. Section \ref{secsplit} contains the splitting criteria and the localized rigidity of Hawking's theorem, and Section \ref{secstru} the structure of a hypothetical counterexample to Bartnik's conjecture, including the cosmological wall.

\subsection{Conventions and assumptions} \label{secconv}

Reference for background causality theory is  \cite{minguzzi18b} unless otherwise stated, and we use its terminology. In particular, Cauchy hypersurfaces are acausal by definition, a {\em time function} is a continuous function which is strictly increasing over causal curves, and $d$ denotes the Lorentzian distance. For $A,B\subset M$ we set $d(A,B)=\sup_{p\in A,q\in B} d(p,q)$ and $d(p,B)=d(\{p\},B)$. Given a closed set $S$, an {\em $S$-ray} is a future inextendible causal curve starting from $S$ such that for every point $q$ on it the segment from $S$ to $q$ has Lorentzian length $d(S,q)$; past $S$-rays are defined dually.

Throughout the paper $(M,g)$ is a connected, time-oriented, globally hyperbolic spacetime of dimension $n+1$ which admits a compact Cauchy hypersurface $\Sigma$.

Timelike geodesic completeness is not a standing assumption: every result states the completeness assumption it uses, if any. We say that $(M,g)$ is {\em future} (resp.\ {\em past}) {\em timelike complete} if every future (resp.\ past) inextendible timelike geodesic is future (resp.\ past) complete, and that it is {\em timelike complete} if both properties hold. Future timelike completeness guarantees that the level sets $\Sigma_\tau$, $\tau>0$, of the Lorentzian distance from $\Sigma$ are compact Cauchy hypersurfaces (Theorem \ref{kvvb}), which is what is needed to construct the future Cauchy Busemann function $b^+$. Thus $b^+$ is defined only for future timelike complete spacetimes, and every statement involving $b^+$ presupposes this assumption, which is nevertheless recalled in each result. Dually, the past Cauchy Busemann function $b^-$ requires past timelike completeness, and the results involving both $b^+$ and $b^-$ require timelike completeness.

The timelike convergence condition (TCC), $Ric(v,v)\ge 0$ for every timelike $v$, will be used through the Lorentzian splitting theorem and in Sections \ref{secmean} and \ref{secmeanC}; results in which it is used are explicitly marked. By continuity, TCC implies the null convergence condition (NCC), $Ric(k,k)\ge 0$ for every lightlike $k$. We say that $(M,g)$ is {\em future null complete} if every future inextendible lightlike geodesic is future complete, and {\em null geodesically complete} if, moreover, the time dual property holds. A {\em timelike line} is an inextendible timelike geodesic each of whose segments is maximizing; it is {\em complete} if it is defined for all values of its proper time. For globally hyperbolic spacetimes the Lorentzian splitting theorem does not require the timelike geodesic completeness of the whole spacetime, but it does require the completeness of the timelike line \cite{galloway89}, see also \cite{beem96} and \cite[Thm.\ 5.1]{galloway14b}: under TCC, the existence of a complete timelike line implies that $(M,g)$ is isometric to a product $(\mathbb{R}\times S,-\dd t^2+h)$. In a timelike complete spacetime every timelike line is complete. The completeness of the line cannot be dropped: the spacetime $M=(0,+\infty)\times T^n$, $g=-\dd t^2+h$, $h$ flat, is globally hyperbolic, future timelike complete and Ricci flat, every vertical geodesic is an inextendible timelike line, but $(M,g)$ is not isometric to a product $\mathbb{R}\times S$.

For the sectional curvature of a non-degenerate plane spanned by $X,Y$ we use the convention of \cite{beem96,ehrlich90}, $K(X,Y)=g(R(X,Y)Y,X)/[g(X,X)g(Y,Y)-g(X,Y)^2]$, where $R(X,Y)=\nabla_X\nabla_Y-\nabla_Y\nabla_X-\nabla_{[X,Y]}$. Then $Ric(X,X)=-\sum_{i=1}^n K(X,e_i)$ for every unit timelike vector $X$ and every orthonormal basis $\{e_i\}$ of $X^\perp$, so that non-positive timelike sectional curvature implies TCC, and the two conditions are equivalent in dimension $1+1$. In the synthetic framework of Lorentzian length spaces the triangle comparison condition corresponding to non-positive timelike sectional curvature in this convention is called non-negative timelike curvature \cite{flores24}.

The following table summarizes where the various assumptions are used.
\begin{center}
\begin{tabular}{>{\raggedright\arraybackslash}p{3cm}>{\raggedright\arraybackslash}p{8.1cm}}
\hline
Assumption & Where it is used \\ \hline
\noalign{\vskip 4pt}
future timelike completeness & Theorem \ref{kvvb} and Corollary \ref{corlipt}; construction of $b^+$ and results involving it \\ \noalign{\vskip 5pt}
past timelike completeness & construction of $b^-$ and results involving it; non-emptiness of $B_c$, $c<0$ (Proposition \ref{cllp}); Proposition \ref{prophoro} and Corollary \ref{corhoro}; splitting conclusions of Theorem \ref{thmcone} and Propositions \ref{propline}, \ref{propequiv} \\ \noalign{\vskip 5pt}
timelike completeness & results involving both $b^+$ and $b^-$, mostly in Sections \ref{secsplit} and \ref{secstru} \\ \noalign{\vskip 5pt}
TCC & Lorentzian splitting theorem; Sections \ref{secmean} and \ref{secmeanC} \\ \noalign{\vskip 5pt}
null completeness & Theorem \ref{thmrigid} and Theorem \ref{thmthin}(d) \\ \noalign{\vskip 3pt} \hline
\end{tabular}
\end{center}

\section{Distance functions from compact Cauchy hypersurfaces} \label{secdist}

The results of this section do not require any completeness assumption, with the exception of Theorem \ref{kvvb} and Corollary \ref{corlipt}, whose assumptions are stated explicitly.

For a Cauchy hypersurface $\Sigma$ we set
\[
d_\Sigma^+(q):=\sup_{p\in \Sigma} d(p,q)=d(\Sigma,q) , \qquad d_\Sigma^-(q):=\sup_{p\in \Sigma} d(q,p)=d(q,\Sigma).
\]
By the achronality of $\Sigma$, $d^+_\Sigma$ vanishes on $J^-(\Sigma)$ and $d^-_\Sigma$ vanishes on $J^+(\Sigma)$. The {\em signed distance} from $\Sigma$ is the function $d_\Sigma:=d^+_\Sigma-d^-_\Sigma$, namely $d_\Sigma=d^+_\Sigma$ on $J^+(\Sigma)$ and $d_\Sigma=-d^-_\Sigma$ on $J^-(\Sigma)$.

\begin{theorem} \label{thmcont}
Let $\Sigma$ be a Cauchy hypersurface. The function $d_\Sigma^+$ is positive, finite and continuous on $I^+(\Sigma)$ and for every $q\in I^+(\Sigma)$ there is $p\in \Sigma$ and a maximizing timelike geodesic from $p$ to $q$ of Lorentzian length $d_\Sigma^+(q)$. Moreover, $d_\Sigma^+$, which vanishes on $J^-(\Sigma)$, is finite and continuous on $M$. Consequently, the signed distance $d_\Sigma$ is finite and continuous on $M$.
\end{theorem}

See also \cite[Prop.\ 2.3]{minguzzi19d}.

\begin{proof}
Recall that $I^+(\Sigma)=\textrm{Int} D^+(\Sigma)$ since $\Sigma$ is Cauchy.
Positivity is clear since any past inextendible timelike curve ending at $q\in I^+(\Sigma)$ must intersect $\Sigma$. Finiteness has been proved in \cite[Cor.\ 3.47]{minguzzi18b} as well as the fact that for every $q$ the sup is realized, namely there is some $p$ and a maximal timelike geodesic segment $\sigma$ connecting $p$ and $q$ such that $\ell(\sigma)=d(p,q)=d^+_{\Sigma}(q)$.

The pointwise supremum of a family of lower semi-continuous functions is lower semi-continuous. As the functions $d(p,\cdot)$ are lower semi-continuous for every $p\in M$, we conclude that $d^+_{\Sigma}$ is lower semi-continuous.

Suppose it is not upper semi-continuous at $q$, then there is $\epsilon>0$ and a sequence of points $(p_n,q_n)$, $q_n\to q$, $p_n\in \Sigma$ such that $d(p_n,q_n)=d^+_\Sigma(q_n)> d^+_\Sigma(q)+\epsilon$. Let $r\in I^+(q)$, then $q\in I^-(r)$, so for large $n$, $q_n\in I^-(r)$ and finally $p_n\in I^-(r)$. But $J^-(r)\cap \Sigma$ is compact thus we can assume that $p_n\to p\in\Sigma$. As the maximizing causal curves $\sigma_n$ connecting $p_n$ to $q_n$ are contained in the compact set $J^+(\Sigma)\cap J^-(r)$, the limit curve theorem for curves with converging endpoints \cite[Thm.\ 2.53, Rem.\ 2.54]{minguzzi18b} gives a  limit  causal curve $\sigma$ connecting $p$ to $q$. Since the $\sigma_n$ are maximizing, by \cite[Thm.\ 2.58]{minguzzi18b}
\[
d^+_\Sigma(q)\ge d(p,q)=\ell(\sigma)= \lim_{n} \ell(\sigma_n)\ge  d^+_\Sigma(q)+\epsilon,
\]
a contradiction.

For the last statement it is sufficient to prove that if $q_n\to q\in \Sigma$, $q_n\in I^+(\Sigma)$, then $d^+_\Sigma(q_n)\to 0$. If not, passing to a subsequence, $d^+_\Sigma(q_n)\ge \epsilon>0$ and, arguing as above, there are $p_n\in \Sigma$ with $d(p_n,q_n)=d^+_\Sigma(q_n)$ and $p_n\to p\in \Sigma$. Since $(M,g)$ is globally hyperbolic, the Lorentzian distance is continuous \cite{beem96,minguzzi18b}, thus $\epsilon\le d(p_n,q_n)\to d(p,q)$. But $d(p,q)=0$ since $p,q\in \Sigma$ and $\Sigma$ is achronal, a contradiction. Finiteness on $M$ follows from finiteness on $I^+(\Sigma)$, since $M\backslash I^+(\Sigma)= I^-(\Sigma)\cup \Sigma\subset J^-(\Sigma)$, where $d^+_\Sigma$ vanishes.

By the time dual argument, $d^-_\Sigma$ is finite and continuous on $M$, hence so is $d_\Sigma=d^+_\Sigma-d^-_\Sigma$.
\end{proof}

Since $\Sigma$ is a Cauchy hypersurface, $J^+(\Sigma)=\Sigma\cup I^+(\Sigma)$. Indeed, the inclusion $\supset$ is clear, $M=I^-(\Sigma)\cup \Sigma \cup I^+(\Sigma)$, and $J^+(\Sigma)$ does not intersect $I^-(\Sigma)$, since $p\le x\ll q$ with $p,q\in \Sigma$ would give $p\ll q$, in contradiction with the achronality of $\Sigma$.

Recall (Section \ref{secconv}) that a future {\em $\Sigma$-ray} is a future inextendible causal curve $\sigma$ starting from $\Sigma$ such that for every point $q$ on it the segment of $\sigma$ from $\Sigma$ to $q$ has Lorentzian length $d(\Sigma,q)=d^+_\Sigma(q)$; past $\Sigma$-rays are defined dually.

\begin{lemma} \label{lemlimray}
Let $\Sigma$ be a compact Cauchy hypersurface and let $q_n\in I^+(\Sigma)$ be a sequence of points without accumulation points (i.e.\ escaping every compact set). Let $\sigma_n$ be a maximizing timelike curve from $p_n\in \Sigma$ to $q_n$ of Lorentzian length $d^+_\Sigma(q_n)$, parametrized with respect to the arc-length of a complete Riemannian metric $h$. Then there is a subsequence which converges uniformly on compact subsets to a future $\Sigma$-ray $\sigma$.
\end{lemma}

\begin{proof}
By compactness of $\Sigma$ we can assume $p_n\to p\in \Sigma$. Since the endpoints $q_n$ escape every compact set and $h$ is complete, the $h$-lengths of the curves $\sigma_n$ diverge, so the limit curve theorem \cite[Thm.\ 2.51, Rem.\ 2.52]{minguzzi18b} provides a future inextendible continuous causal curve $\sigma:[0,+\infty)\to M$ starting from $p$ and a subsequence converging to it $h$-uniformly on compact subsets. Since the curves $\sigma_n$ are $\Sigma$-maximizing, $\sigma$ is $\Sigma$-maximizing \cite[Thm.\ 2.61]{minguzzi18b}, that is, $\sigma$ is a future $\Sigma$-ray.
\end{proof}

\begin{theorem} \label{cpgx}
Let $\Sigma$ be a compact Cauchy hypersurface. Then there is a future $\Sigma$-ray and every future $\Sigma$-ray is timelike.
\end{theorem}

\begin{proof}
By \cite[Thm.\ 2.65]{minguzzi18b} (or by Lemma \ref{lemlimray}) there is a future $\Sigma$-ray $\sigma$. Let $q\in \Sigma$ be its starting point, pick $p\ll q$ and $r\in  \sigma \backslash \{q\}$, so that $r\notin \Sigma$ by the acausality of $\Sigma$. The timelike curve from $p\in I^-(\Sigma)$ to $r\in J^+(\Sigma)\backslash \Sigma=I^+(\Sigma)$ intersects the Cauchy hypersurface $\Sigma$ at some point $s\ne r$, thus $d^+_\Sigma(r)\ge d(s,r)>0$. As the segment of $\sigma$ from $q$ to $r$ is maximizing and has positive Lorentzian length, it is a timelike geodesic, hence $\sigma$ cannot be lightlike.
\end{proof}

For the following result see also \cite[Lemma 4.2]{galloway97} \cite[Lemma 3.15]{galloway14b}.
We set $\Sigma_c:=(d_\Sigma^+)^{-1}(c)$ for $c> 0$, $\Sigma_{-c}:=(d_\Sigma^-)^{-1}(c)$ for $c>0$, and $\Sigma_0=\Sigma$. Observe that $\Sigma_c=d_\Sigma^{-1}(c)$ for every $c\in \mathbb{R}$, since $d^\pm_\Sigma>0$ on $I^\pm(\Sigma)$ by Theorem \ref{thmcont} and its time dual.

\begin{theorem} \label{kvvb}
Let  $\Sigma$ be a compact Cauchy hypersurface and assume that $(M,g)$ is future timelike complete. Then for every $c>0$ the level set $\Sigma_c$ is a non-empty compact Cauchy hypersurface, $d^+_\Sigma$ is strictly increasing over causal curves contained in $J^+(\Sigma)$, and for $c\ge 0$ the set $J^+(\Sigma)\cap (d_\Sigma^+)^{-1}([0,c])=J^+(\Sigma)\cap J^-(\Sigma_c)$ is compact. Moreover, $(d_{\Sigma_a}^+)^{-1}(b)=\Sigma_{a+b}$ for $a,b>0$. The time dual statements hold if $(M,g)$ is past timelike complete. Consequently, if $(M,g)$ is timelike complete, the signed distance $d_\Sigma$ is a Cauchy time function whose level sets $\Sigma_c$, $c\in \mathbb{R}$, are compact Cauchy hypersurfaces.
\end{theorem}

The theorem does not hold without the future timelike completeness condition. Identify $x=0$ and $x=-1$ in 1+1 Minkowski spacetime, let $\Sigma$ be the slice $t=0$ and remove $J^+(q)$ from the spacetime where $q$ is a point in $I^+(\Sigma)$. The resulting spacetime is still globally hyperbolic with Cauchy hypersurface $\Sigma$, but for large $c$ the level sets $\Sigma_c$ are not Cauchy.

\begin{proof}
By Theorem \ref{thmcont} $d_\Sigma$ is continuous.
Let $c>0$. The set $\Sigma_c$ is non-empty. Indeed, by Theorem \ref{cpgx} there is a future $\Sigma$-ray $\sigma$, which is timelike and, by future timelike completeness, is defined, with proper time parametrization, on $[0,+\infty)$. Since $\sigma$ is a $\Sigma$-ray, $d^+_\Sigma(\sigma(c))=\ell(\sigma\vert_{[0,c]})=c$, namely $\sigma(c)\in \Sigma_c$.
Observe that $\Sigma_c$ is acausal. In fact, suppose there is a causal curve connecting  $q_1,q_2\in \Sigma_c$, $q_1\ne q_2$. By global hyperbolicity we can replace it with a maximizing causal geodesic $\sigma$ from $q_1$ to $q_2$. We know that there is $p\in \Sigma$ and a timelike geodesic connecting $p$ and $q_1$ of Lorentzian length $c$. Concatenating it with $\sigma$ gives a causal curve from $p$ to $q_2$ which is either longer than $c$, if $\sigma$ has positive Lorentzian length, or, if $\sigma$ has zero Lorentzian length, and hence is a lightlike geodesic, has a corner at $q_1$, and so can be deformed into a longer curve \cite{minguzzi18b}. This contradicts $q_2\in \Sigma_c$.

Suppose $c>0$ is such that $\Sigma_c$ is not a Cauchy hypersurface. Then there is an inextendible causal curve $\gamma$ which does not intersect $\Sigma_c$. It intersects $\Sigma$ at some point $u$. Let us parametrize $\gamma$ with respect to arc-length of a complete Riemannian metric so that it is future-directed and $\gamma(0)=u$. Since $d_\Sigma^+$ is continuous, we have  $d_\Sigma^+(\gamma(t))<c$ for every $t \ge 0$. Let $q_n=\gamma(n)$. By Lemma \ref{lemlimray} the  maximizing timelike curves $\sigma_n$ from $\Sigma$ to $q_n$ converge, up to a subsequence, to a future $\Sigma$-ray $\sigma$ starting from some $p\in \Sigma$. By Theorem \ref{cpgx} $\sigma$ is timelike, and by future timelike completeness it is defined, with proper time parametrization, on $[0,+\infty)$. Thus we can find some $t$ such that $d^+_\Sigma(\sigma(t))=\ell(\sigma\vert_{[0,t]})>c$, thus by the continuity of $d^+_\Sigma$, for sufficiently large $n$, $d_\Sigma^+(q_n)=\ell(\sigma_n)\ge  d^+_\Sigma(\sigma_n(t))> c$, a contradiction which proves that  $\Sigma_c$ is a Cauchy hypersurface.

If $p\le q$, $p,q\in J^+(\Sigma)$, as we have $d^+_\Sigma(q)\ge d^+_\Sigma(p)+d(p,q)$, the function $d^+_\Sigma$ is non-decreasing over causal curves on $J^+(\Sigma)$. It must be strictly increasing because if $p<q$, $p,q\in J^+(\Sigma)$ and $d^+_\Sigma(p)=d^+_\Sigma(q)=c$ then they belong to the same  Cauchy hypersurface $\Sigma_c$, $c\ge 0$, which contradicts its acausality. Under past timelike completeness, $d_\Sigma$ is similarly strictly increasing on $J^-(\Sigma)$, so under timelike completeness it is a time function.
Under timelike completeness, $d_\Sigma$ is a Cauchy time function, namely its restriction to every inextendible causal curve $\gamma$ has image $\mathbb{R}$. Indeed, let $c\in \mathbb{R}$. The level set $d_\Sigma^{-1}(c)$ is $\Sigma_c$, which is a Cauchy hypersurface by what we have proved for $c>0$, by the time dual result for $c<0$, and since $\Sigma_0=\Sigma$ for $c=0$. Thus $\gamma$ intersects $\Sigma_c$, i.e.\ $c$ belongs to the image of $d_\Sigma\circ \gamma$. The hypersurfaces $\Sigma_c$ are homeomorphic to $\Sigma$, hence compact, and  $J^+(\Sigma)\cap (d_\Sigma^+)^{-1}([0,c])$$=J^+(\Sigma)\cap J^-(\Sigma_c)$ is compact (note that $d^+_\Sigma$ vanishes on $J^-(\Sigma)$, so the intersection with $J^+(\Sigma)$ is needed).

Let $r\in \Sigma_{a+b}$, $a,b>0$. If there is $q\in \Sigma_a$ and a causal curve $\sigma$ connecting $q$ to $r$ of Lorentzian length $\ell(\sigma)>b$, since we can find a causal curve of Lorentzian length $a$ from a point $p\in \Sigma$ to $q$, we would have a causal curve connecting $\Sigma$ to $r$ of Lorentzian length larger than $a+b$, contradicting the definition of $\Sigma_{a+b}$, thus $d_{\Sigma_a}^+(r)\le b$. Now, consider the causal curve $\gamma$ connecting $\Sigma$ to $r$ of Lorentzian length $a+b$. It intersects $\Sigma$, $\Sigma_a$ at points $p,q$. Clearly $d(p,q)=d^+_\Sigma(q)=a$ otherwise we could find a longer curve contradicting the definition of $\Sigma_{a+b}$, and since
$\gamma$ is maximizing, $d(p,q)+d(q,r)=d(p,r)$, which reads $a+d(q,r)=a+b$. This shows that $d_{\Sigma_a}^+(r)\ge b$, hence $d_{\Sigma_a}^+(r)=b$. We have proved $\Sigma_{a+b}\subset (d_{\Sigma_a}^+)^{-1}(b)$. The former set is a Cauchy hypersurface and the latter is acausal (by the acausality argument above applied to $\Sigma_a$), and a Cauchy hypersurface $S$ contained in an achronal set $S'$ coincides with it: if $x\in S'\backslash S$, an inextendible timelike curve through $x$ meets $S$ at a point $y\ne x$ of $S'$, contradicting the achronality of $S'$.
\end{proof}

\begin{remark} \label{remcomplete}
In the proof of Theorem \ref{kvvb}, future timelike completeness is used only to ensure that every future $\Sigma$-ray, which is a timelike geodesic by Theorem \ref{cpgx}, has infinite Lorentzian length. Conversely, if $\Sigma_c$ is a Cauchy hypersurface for every $c>0$, then every future $\Sigma$-ray $\sigma$ meets every $\Sigma_c$ and, since $d^+_\Sigma(\sigma(t))=\ell(\sigma\vert_{[0,t]})$, it has infinite Lorentzian length, namely it is future complete. Thus the level sets $\Sigma_c$, $c>0$, are non-empty Cauchy hypersurfaces if and only if the future $\Sigma$-rays are future complete.
\end{remark}

\begin{proposition} \label{nner}
Let $a\le b$. Every causal curve from $\Sigma_a$ to $\Sigma_b$ has Lorentzian length at most $b-a$, thus $d(\Sigma_a,\Sigma_b)\le b-a$. No completeness assumption is needed.
\end{proposition}

\begin{proof}
Let $\sigma$ be a causal curve from $p\in \Sigma_a$ to $r\in \Sigma_b$. If $a\le 0\le b$, then $\sigma$ meets the Cauchy hypersurface $\Sigma$ at some point $q$, thus $\ell(\sigma)\le d(p,q)+d(q,r)\le d^-_\Sigma(p)+d^+_\Sigma(r)=b-a$. If $0\le a$, then $p\in J^+(\Sigma)$, hence $r\in J^+(\Sigma)$, and $d^+_\Sigma(r)\ge d^+_\Sigma(p)+d(p,r)$ gives $\ell(\sigma)\le d(p,r)\le b-a$. If $b\le 0$, then dually $p,r\in J^-(\Sigma)$ and $d^-_\Sigma(p)\ge d(p,r)+d^-_\Sigma(r)$ gives $\ell(\sigma)\le d(p,r)\le -a-(-b)=b-a$.
\end{proof}

This result will be used repeatedly. We shall also use the following consequence.

\begin{corollary} \label{corlipt}
Assume that $(M,g)$ is future timelike complete and let $0<\tau\le \tau'$. For every $y\in M$
\[
d(y,\Sigma_\tau)\le d(y,\Sigma_{\tau'})\le d(y,\Sigma_\tau)+\tau'-\tau ,
\]
in particular $\vert d(y,\Sigma_\tau)-d(y,\Sigma_{\tau'})\vert\le \vert \tau-\tau'\vert$.
\end{corollary}

\begin{proof}
By Theorem \ref{kvvb}, $\Sigma_\tau$ and $\Sigma_{\tau'}$ are Cauchy hypersurfaces and $\Sigma_{\tau'}\subset J^+(\Sigma_\tau)$. For the first inequality we can assume $d(y,\Sigma_\tau)>0$. Let $z\in \Sigma_\tau$ be such that $d(y,z)=d(y,\Sigma_\tau)$ (time dual of Theorem \ref{thmcont}). A future inextendible causal curve starting from $z$ meets $\Sigma_{\tau'}$ at some point $z'$, thus $d(y,\Sigma_{\tau'})\ge d(y,z')\ge d(y,z)=d(y,\Sigma_\tau)$. For the second inequality we can assume $d(y,\Sigma_{\tau'})>0$, and we let $\gamma$ be a maximizing causal curve from $y$ to a point $z'\in \Sigma_{\tau'}$ with $\ell(\gamma)=d(y,\Sigma_{\tau'})$. If $y\in J^-(\Sigma_\tau)$, then $\gamma$ meets $\Sigma_\tau$ at some point $x$, and by Proposition \ref{nner}, $\ell(\gamma)\le d(y,x)+(\tau'-\tau)\le d(y,\Sigma_\tau)+\tau'-\tau$. Otherwise $y\in I^+(\Sigma_\tau)$, so that $d(y,\Sigma_\tau)=0$ by the achronality of $\Sigma_\tau$, and $\gamma$ can be prolonged to the past to a causal curve from $\Sigma_\tau$ to $\Sigma_{\tau'}$, so that $\ell(\gamma)\le \tau'-\tau$ by Proposition \ref{nner}.
\end{proof}

\begin{definition} \label{defrush}
A function $u: M\to [-\infty, +\infty]$ is {\em rushing} \cite{minguzzi18b} (dominated) on a subset $A\subset M$ if for every $p,q\in A$, $p\le q$,
\[
u(p)+d(p,q)\le u(q).
\]
It is rushing if it is rushing on $M$.
\end{definition}

A rushing function is non-decreasing along causal curves, and it increases along every causal curve at least by its Lorentzian length.

\begin{proposition} \label{propdrush}
Let $\Sigma$ be a Cauchy hypersurface. The signed distance $d_\Sigma$ is rushing. No completeness assumption is needed.
\end{proposition}

\begin{proof}
Let $p\le q$. If $p,q\in J^+(\Sigma)$, then $d^+_\Sigma(q)\ge d^+_\Sigma(p)+d(p,q)$, as a causal curve from $\Sigma$ to $p$ can be concatenated with a causal curve from $p$ to $q$; dually if $p,q\in J^-(\Sigma)$. If $p\in J^-(\Sigma)$ and $q\in J^+(\Sigma)$, a maximizing causal curve from $p$ to $q$ crosses $\Sigma$ at some $r$, and $d(p,q)=d(p,r)+d(r,q)\le d^-_\Sigma(p)+d^+_\Sigma(q)=d_\Sigma(q)-d_\Sigma(p)$. The case $p\in J^+(\Sigma)$, $q\in J^-(\Sigma)$ can occur only if $p,q\in\Sigma$, as $\Sigma$ is achronal, and it is covered by the previous cases.
\end{proof}

We end this section with a regularity result for the signed distance.

\begin{proposition} \label{lemdsig}
Let $\Sigma$ be a Cauchy hypersurface. The signed distance $d_\Sigma$ is locally Lipschitz on $M\backslash\Sigma$, and if $\Sigma$ is a $C^{1,1}$ spacelike hypersurface it is locally Lipschitz on $M$. No completeness assumption is needed.
\end{proposition}

\begin{proof}
If $\Sigma$ is a $C^{1,1}$ spacelike hypersurface, then $d_\Sigma$ is $C^{1,1}$ on a neighborhood of $\Sigma$ \cite[Thm.\ 2.2]{minguzzi19d}, so it suffices to prove the first statement. As $M\backslash\Sigma=I^+(\Sigma)\cup I^-(\Sigma)$, by time duality it suffices to prove that $f:=d^+_\Sigma$ is locally Lipschitz on $I^+(\Sigma)$. Fix an auxiliary complete Riemannian metric $h$ on $M$. Let $K\subset I^+(\Sigma)$ be compact, let $\delta:=\min_K f>0$ (Theorem \ref{thmcont}), and let $V_K$ be the set of the final unit tangents $w\in T_qM$, $q\in K$, of the proper time parametrized maximizing timelike geodesics from $\Sigma$ to $q$ of Lorentzian length $f(q)$, which exist by Theorem \ref{thmcont}.

We claim that $V_K$ is bounded with respect to $h$. Otherwise there are $q_n\in K$ and $w_n\in V_K\cap T_{q_n}M$ with $\vert w_n\vert_h\to+\infty$ and, passing to a subsequence, $q_n\to q\in K$ and $u_n:=w_n/\vert w_n\vert_h\to u$, where $\vert u\vert_h=1$ and $u$ is future directed causal. Since $g(u_n,u_n)=-\vert w_n\vert_h^{-2}\to 0$, we have $g(u,u)=0$, so $u$ is a future directed lightlike vector at $q$. Let $\sigma_n$ be the maximizing geodesic from $\Sigma$ to $q_n$ with final tangent $w_n$, parametrized with respect to $h$-arc length. As $(M,g)$ is globally hyperbolic, hence non-totally imprisoning, and the curves $\sigma_n$ are contained in the compact set $J^+(\Sigma)\cap J^-(K)$, by the limit curve theorem for curves with converging endpoints \cite[Thm.\ 2.53(i), Rem.\ 2.54]{minguzzi18b} a subsequence of the curves $\sigma_n$ converges uniformly to a causal curve $\sigma$ connecting $\Sigma$ to $q$. As the curves $\sigma_n$ are $\Sigma$-maximizing, $\sigma$ is $\Sigma$-maximizing and $\ell(\sigma)=d(\Sigma,q)=f(q)>0$ \cite[Thm.\ 2.61]{minguzzi18b}, see the last equality in its Eq.\ (2.13). Thus $\sigma$ is a geodesic up to reparametrization \cite[Thm.\ 2.20]{minguzzi18b} and, by the continuous dependence of geodesics on the initial data, its final tangent is proportional to $u$, so $\sigma$ is lightlike, in contradiction with $\ell(\sigma)>0$.

Now let $x_0\in I^+(\Sigma)$. By \cite[Thm.\ 2.4]{minguzzi18b} there are a coordinate chart around $x_0$ and $\rho>0$ such that the coordinate ball $C:=B(x_0,\rho)$ is a convex normal neighborhood with compact closure in $I^+(\Sigma)$. Let $U:=B(x_0,\rho/2)$ and $K:=\overline{U}$. As the domain of the geodesic flow is open and, by the claim, the closure of $V_K$ is compact, there is $s_0\in(0,\delta)$ such that for every $x\in K$ and $w\in V_K\cap T_xM$ the geodesic segment $s\mapsto\exp_x(-sw)$, $s\in[0,s_0]$, is contained in $C$. For $x\in U$ let $w\in V_K\cap T_xM$ and $r_x:=\exp_x(-s_0w)$. As $s_0<\delta\le f(x)$, $r_x$ lies on a maximizing geodesic from $\Sigma$ to $x$, so that $f(r_x)=f(x)-s_0$, and the segment of this geodesic from $r_x$ to $x$ is the geodesic of $C$ connecting them. The function $D^2_{r_x}:=g(\exp_{r_x}^{-1}(\cdot),\exp_{r_x}^{-1}(\cdot))$ is smooth on $C$, with differential given by Gauss' lemma \cite[Thm.\ 2.6]{minguzzi18b}, and the function $\varphi_x:=f(r_x)+\sqrt{-D^2_{r_x}}$ is smooth on the open set $I^+_C(r_x)\ni x$ of the points of $C$ reached from $r_x$ by a future directed timelike geodesic of $C$ \cite[Cor.\ 2.10]{minguzzi18b}. As $\sqrt{-D^2_{r_x}(y)}$ is the Lorentzian length of this geodesic, $f(y)\ge f(r_x)+d(r_x,y)\ge \varphi_x(y)$ for $y\in I^+_C(r_x)$, with equality for $y=x$, so $\varphi_x$ is a lower support function for $f$ at $x$. By \cite[Eq.\ (2.6)]{minguzzi18b}, $\dd \varphi_x(x)=-g_x(w,\cdot)$, whose norm in the coordinates is bounded by a constant $L$ independent of $x\in U$, by the claim. By \cite[Sublemma 14.20]{beem96}, $f$ is Lipschitz on $U$.
\end{proof}

\section{The Cauchy Busemann function} \label{secbuse}

In this section $\Sigma$ is a compact Cauchy hypersurface. The construction of the Cauchy Busemann function uses the level sets $\Sigma_\tau$, $\tau>0$, of $d^+_\Sigma$, which are compact Cauchy hypersurfaces if $(M,g)$ is future timelike complete (Theorem \ref{kvvb}). Accordingly, this assumption appears in the definition of the Cauchy Busemann function and in all the results which involve it. Theorem \ref{cbdc}, the properties of the Busemann function $b^\sigma$ of a complete ray, and the general facts on rushing functions and calibrated curves at the beginning of Section \ref{secreg} do not need it.

Recall that $\Sigma_\tau=(d_\Sigma^+)^{-1}(\tau)$ for $\tau>0$, and $\Sigma_0=\Sigma$. If $\Sigma_\tau\ne \emptyset$, which is the case for every $\tau$ if $(M,g)$ is future timelike complete (Theorem \ref{kvvb}), then $d(\Sigma,\Sigma_\tau)=\tau$.

\begin{proposition} \label{kcqq}
Assume that $(M,g)$ is future timelike complete. For any proper time parametrized future $\Sigma$-ray, $\sigma:[0,+\infty)\to M$, we have    for $0\le t \le \tau$, $d(\sigma(t),\Sigma_\tau)=\tau-t$, in particular $d(\sigma(0),\Sigma_\tau)=\tau$.
\end{proposition}

For a generic point $p\in \Sigma$ we have just $d(p,\Sigma_\tau)\le \tau$, thus the starting points of $\Sigma$-rays are rather special.
\begin{proof}
Let  $\sigma$ be a $\Sigma$-ray, observe that $\sigma_\tau:=\sigma\vert_{[0,\tau]}$ has endpoint in $\Sigma_\tau$, as  $\ell(\sigma_\tau)=d(\Sigma, \sigma(\tau))=\tau$.
The Lorentzian distance between $\sigma(t)$ and $\Sigma_\tau$ is attained by a geodesic  $\gamma$ connecting $\sigma(t)$ to a point $q\in \Sigma_\tau$. If it were $\ell(\gamma)> \tau-t$, then the causal curve obtained by following $\sigma\vert_{[0,t]}$ and then $\gamma$ would give $d(\Sigma,q)>\tau$, a contradiction since $q\in \Sigma_\tau$. But it cannot be $\ell(\gamma)<\tau-t$ as otherwise $\gamma$ would not be maximizing the Lorentzian distance between $\sigma(t)$ and $\Sigma_\tau$ as $\sigma\vert_{[t,\tau]}$ has Lorentzian length $\tau-t$.
\end{proof}

Assume that $(M,g)$ is future timelike complete, so that the sets $\Sigma_\tau$, $\tau>0$, are compact Cauchy hypersurfaces (Theorem \ref{kvvb}), and consider for $\tau>0$ the pre-Busemann type function
\[
b_\tau(p):=\tau-d(p,\Sigma_\tau), \qquad p\in M.
\]
Observe that for $c<\tau$, $b_\tau^{-1}(c)=(d^-_{\Sigma_\tau})^{-1}(\tau-c)$, thus these level sets are acausal by the proof of Theorem \ref{kvvb} and, if $(M,g)$ is also past timelike complete, they are Cauchy hypersurfaces by the time dual of Theorem \ref{kvvb} (without past timelike completeness they might even be empty), while $b_\tau=\tau$ on the whole set $J^+(\Sigma_\tau)$.

Notice that for $p\in J^+(\Sigma)$ and sufficiently large $\tau$, $d(\Sigma,p)+d(p,\Sigma_\tau)\le d(\Sigma,\Sigma_\tau)=\tau$ thus
\[
0\le d(\Sigma,p) \le \tau-d(p,\Sigma_\tau)=b_\tau(p) .
\]
If, instead, $p\in J^-(\Sigma)$, every causal curve from $p$ to $\Sigma_\tau$ meets $\Sigma$, thus $d(p,\Sigma_\tau)\le d(p,\Sigma)+d(\Sigma,\Sigma_\tau)=d^-_\Sigma(p)+\tau$, namely
\[
 -d^-_\Sigma(p) \le \tau-d(p,\Sigma_\tau)=b_\tau(p) ,
\]
which proves that for each $p$, $b_\tau(p)$ is bounded from below uniformly in $\tau$ (past timelike completeness is not needed for this).

If $p\le q$ and $q\in J^-(\Sigma_\tau)$ then $d(p,q)+d(q,\Sigma_\tau)\le d(p,\Sigma_\tau)$, thus
\begin{equation} \label{rushtau}
b_\tau(p)+d(p,q)\le b_\tau(q), \qquad p\le q, \ q \in J^-(\Sigma_\tau),
\end{equation}
namely $b_\tau$ is rushing on $J^-(\Sigma_\tau)$ (Definition \ref{defrush}). It is not rushing on $M$ since $b_\tau$ is constant on $J^+(\Sigma_\tau)$. Moreover, $b_\tau$ is a time function on $J^-(\Sigma_\tau)$. Indeed, it is continuous, as $d(\cdot,\Sigma_\tau)$ is continuous (Theorem \ref{thmcont}), and by (\ref{rushtau}) it is non-decreasing along the causal curves contained in $J^-(\Sigma_\tau)$. It is strictly increasing along them since its level sets in $J^-(\Sigma_\tau)$ are acausal: for $c<\tau$ this was observed above, while for $c=\tau$ the level set is $\Sigma_\tau$, as $d(\cdot,\Sigma_\tau)>0$ on $J^-(\Sigma_\tau)\backslash\Sigma_\tau=I^-(\Sigma_\tau)$.

Since for $\tau'>\tau$ every causal curve from $\Sigma_\tau$ to $\Sigma_{\tau'}$ has Lorentzian length no larger than $\tau'-\tau$ (Proposition \ref{nner}), we have $d(p,\Sigma_{\tau'})\le  d(p,\Sigma_{\tau}) +\tau'-\tau$, which implies
\[
b_{\tau'}(p)=\tau'-d(p,\Sigma_{\tau'})\ge \tau-d(p,\Sigma_{\tau})=b_\tau(p).
\]
The family of functions is non-decreasing in $\tau$, while in the  study of the Lorentzian splitting theorem the analogous family is non-increasing.
\begin{definition} \label{defbuse}
Assume that $(M,g)$ is future timelike complete. The {\em Cauchy Busemann function} $b$ (dependent on $\Sigma$) is the function given at $p\in M$ by
\begin{equation} \label{eqbuse}
b(p):=\lim_{\tau\to \infty} b_\tau(p) =\sup_{\tau>0} b_\tau(p).
\end{equation}
\end{definition}
As $b$ is defined only under future timelike completeness, every statement involving $b$ presupposes this assumption; in the results it will be recalled explicitly. From Section \ref{secsplit} on, this function will be denoted $b^+$, to distinguish it from its time dual $b^-$, but in this section and in Sections \ref{secprod}--\ref{sechoro} we denote it $b$ for simplicity.
The lower bounds satisfied by $b_\tau(p)$ prove that
\begin{equation} \label{cmlr}
b\ge d_\Sigma \quad \textrm{on } M,
\end{equation}
 and $b>-\infty$ over $M$. Thus $b: M \to (-\infty,+\infty]$.
 Observe that $b_\tau$ is continuous because  $d(\cdot,\Sigma_\tau)$ is continuous (Thm.\ \ref{thmcont}), thus $b$, being the supremum of a family of  continuous functions, is lower semi-continuous.

For $p\le q$, for sufficiently large $\tau$ we have $q\in J^-(\Sigma_\tau)$, thus passing to the limit in (\ref{rushtau}) we  get (note that $d(p,q)$ is finite but the other terms can be infinite in principle)
\begin{equation} \label{fcpb}
b(p)+d(p,q)\le b(q),
\end{equation}
namely $b$ is rushing.

\begin{remark} \label{cort}
The function $b_\tau$ depends on the choice of $\Sigma$, so let us write it  $b_\tau^\Sigma$ for the moment. If we choose $\Sigma'=\Sigma_s$, $s>0$, which is a compact Cauchy hypersurface by Thm.\ \ref{kvvb}, then, again by Thm.\ \ref{kvvb}, $\Sigma'_{\tau'}=\Sigma_{\tau}$ for every $\tau'>0$, where $\tau=\tau'+s$. Thus, for $\tau'>0$,
\[
b^{\Sigma'}_{\tau'}(p)=\tau'-d(p,\Sigma'_{\tau'})=\tau-s-d(p,\Sigma_{\tau})=b_\tau^\Sigma(p)-s
\]
and taking the limit
$
b^{\Sigma'}=b^{\Sigma}-s.
$
This means that the level sets of $b$ corresponding to non-negative values can be seen as zero level sets of other Cauchy Busemann functions, but this fact does not extend to the negative values.
\end{remark}

\begin{remark}
The set $b_k^{-1}(0)$ is what Galloway and Vega  \cite{galloway14b} call a Cauchy prehorosphere. We shall see later that, under past timelike completeness, $b^{-1}(0)$ is the Cauchy horosphere, namely the achronal boundary of the future set $\cup_k I^+(b_{k}^{-1}(0))$. The proof  requires first to establish the continuity of $b$. Galloway and Vega studied directly the regularity of the horospheres rather than that of $b$.
\end{remark}

It is important to prove that $b$ is not identically $+\infty$. This fact follows from the existence of future $\Sigma$-rays (Theorem \ref{cpgx}).
\begin{proposition} \label{ncho}
Assume that $(M,g)$ is future timelike complete. Let $\sigma$ be any future proper time-parametrized $\Sigma$-ray. We have $b(\sigma(t))=t$.
\end{proposition}

\begin{proof}
We have proved in Prop.\ \ref{kcqq} that $d(\sigma(t),\Sigma_\tau)=\tau-t$ which reads $b_\tau(\sigma(t))=t$ for $\tau\ge t$, and so $b(\sigma(t))=t$.
\end{proof}

Recall that a {\em past set} \cite{hawking73,minguzzi18b} is a set $P\subset M$ such that $I^-(P)\subset P$, so that $I^-(P)=P$ if $P$ is open, and that an open past set is {\em indecomposable} (an IP) if it is non-empty and cannot be written as the union of two open past sets which are proper subsets of it. A {\em terminal indecomposable past set} (TIP) is an IP which is not of the form $I^-(p)$, $p\in M$. By \cite[Thm.\ 2.1, 2.3]{geroch72} the TIPs are precisely the sets of the form $I^-(\gamma)$ where $\gamma$ is a future inextendible timelike curve. The following fact will be used in Corollary \ref{cormintip} and Proposition \ref{proppastnull}.

\begin{theorem} \label{cbdc}
Let $\Sigma$ be a compact Cauchy hypersurface. Every TIP contains another TIP of the form $I^-(\sigma)$ where $\sigma$ is a future $\Sigma$-ray (hence timelike).
\end{theorem}

\begin{proof}
Let $I^-(\tilde \gamma)$ be a TIP, where $\tilde \gamma$ is a future inextendible timelike curve. Since $\tilde \gamma$ intersects the Cauchy hypersurface $\Sigma$, denoting with $\gamma$ the portion in $I^+(\Sigma)$ we have $I^-(\tilde \gamma)=I^{-}(\gamma)$. Let $\gamma$ be parametrized with respect to the arc-length of a complete Riemannian metric so that it is future-directed and let  $q_n=\gamma(n)$. Let $\sigma_n$ be a maximizing timelike curve connecting $p_n\in \Sigma$ to $q_n$ of Lorentzian length $d^+_\Sigma(q_n)$. Observe that $\sigma_n\subset J^-(q_n)\subset I^-(\gamma)$. By Lemma \ref{lemlimray} a subsequence converges to a future $\Sigma$-ray $\sigma$, which is timelike by Theorem \ref{cpgx}, and $\sigma\subset \overline{I^-(\gamma)}$. Since $\sigma$ is timelike, every point of $\sigma$ belongs to $I^-(\overline{I^-(\gamma)})=I^-(\gamma)$, thus $I^-(\sigma)\subset I^-(\gamma)$.
\end{proof}

Let $\sigma$ be a proper time parametrized future $\Sigma$-ray defined on $[0,+\infty)$, which is always the case if $(M,g)$ is future timelike complete, and let $b^\sigma:M\to (-\infty,+\infty]$ be its Busemann function
\[
b^\sigma(p)=\lim_{t\to+\infty}[d(\sigma(0),\sigma(t))-d(p,\sigma(t))] = \lim_{t\to+\infty}[t-d(p,\sigma(t))].
\]
The limit exists: if $p\notin I^-(\sigma)$ the bracket equals $t$ and $b^\sigma(p)=+\infty$; if $p\in I^-(\sigma)$, the bracket is non-increasing in $t$ once $p\ll \sigma(t)$ since $d(p,\sigma(t'))\ge d(p,\sigma(t))+t'-t$ for $t<t'$.
It is  known \cite[Eq.\ (5-2)]{galloway96} that $b^\sigma$ is finite on $I^-(\sigma)$. Indeed,  on
 $I^-(\sigma)\cap J^+(\Sigma)$ it is non-negative, for if $p$ belongs to that set, by the reverse triangle inequality, $d^+_\Sigma(p)\le b^\sigma(p)$.
 If $p\in J^-(\Sigma)$ then for $t\ge 0$ every maximizing curve from $p$ to $\sigma(t)$ meets $\Sigma$, thus $d(p,\sigma(t))\le d(p,\Sigma)+d(\Sigma,\sigma(t))=  d(p,\Sigma)+t$, hence $b^\sigma(p)\ge -d(p,\Sigma)>-\infty$.

Let us denote with $\textrm{dom} \,b$ the effective domain of $b$, namely the subset of $M$ where it is finite.

\begin{proposition} \label{cmol}
Assume that $(M,g)$ is future timelike complete. Let $\sigma$ be any future $\Sigma$-ray. We have $b\le b^\sigma$ on $I^-(\sigma)$,  $b=b^\sigma$ over $\sigma$, and $I^-(\sigma) \subset \textrm{dom}\, b$.
The effective domain $\textrm{dom} \,b$ intersects every Cauchy hypersurface and satisfies $J^-(\textrm{dom}\, b)\subset \textrm{dom}\, b$ (hence it is a past set).
\end{proposition}

\begin{proof}
Let $p\in I^-(\sigma)$ then there is a sufficiently large
$t>0$ such that $p\ll \sigma(t)$. For $n\ge t$, $\sigma(n)\in \Sigma_n$ and $d(p,\Sigma_n)\ge d(p, \sigma(n))$. Then using Prop.\ \ref{kcqq}
\[
b(p)=\lim_{n\to \infty} [n-d(p,\Sigma_n)]
\le \lim_{n\to \infty} [n-d(p, \sigma(n))]=b^\sigma(p)<+\infty.
\]
This proves $I^-(\sigma)\subset \textrm{dom}\, b$.
By Prop.\ \ref{ncho} and Prop.\ \ref{kcqq} we have $b=b^\sigma$ over $\sigma$.

Let $S$ be a Cauchy hypersurface. Let us extend $\sigma$ to an inextendible timelike geodesic. By the rushing property (\ref{fcpb}) $b$ is finite on the  extension to the past of $\sigma(0)$, and it is finite on $\sigma\subset I^-(\sigma)$. Since the extended geodesic is inextendible it intersects $S$. Thus $\textrm{dom}\, b\cap S \ne \emptyset$.
The inclusion $J^-(\textrm{dom}\, b)\subset \textrm{dom}\, b$ follows from the rushing property and the fact that $b$ does not take the value $-\infty$.
\end{proof}

For compact Cauchy hypersurfaces $\Sigma$ and $\Sigma'$ let
\[
R(\Sigma,\Sigma'):= \sup_{p\in \Sigma,p'\in \Sigma'} \max(d(p,p'),d(p',p)).
\]

\begin{lemma} \label{lemRdist}
Let $\Sigma$ and $\Sigma'$ be compact Cauchy hypersurfaces, possibly intersecting. Then $R(\Sigma,\Sigma')<+\infty$ and for every $p\in M$,
\[
\vert d(p,\Sigma)-d(p,\Sigma')\vert\le R(\Sigma,\Sigma'), \quad \textrm{and} \quad \vert d(\Sigma,p)-d(\Sigma',p)\vert\le R(\Sigma,\Sigma') .
\]
(No completeness assumption is needed.)
\end{lemma}

\begin{proof}
The supremum $R(\Sigma,\Sigma')$ is attained, hence finite, by the compactness of $\Sigma$ and $\Sigma'$ and the continuity of $d$. Let $R:=R(\Sigma,\Sigma')$. We prove the former inequality, the latter being analogous. There are various cases to consider depending on whether $p$ belongs to $J^-(\Sigma)\cap J^-(\Sigma')$, to $J^-(\Sigma)\cap J^+(\Sigma')$, to $J^+(\Sigma)\cap J^-(\Sigma')$, or to $J^+(\Sigma)\cap J^+(\Sigma')$.

In the first case consider the  sets $\Sigma^-=\p[J^+(\Sigma)\cup J^+(\Sigma')]$ and $\Sigma^+=\p[J^-(\Sigma)\cup J^-(\Sigma')]$, which are  Cauchy hypersurfaces contained in $\Sigma\cup \Sigma'$ with $\Sigma, \Sigma'\subset J^+(\Sigma^-)\cap J^-(\Sigma^+)$ and $p\in J^-(\Sigma^-)$. (They are acausal: e.g.\ if $x\in \Sigma\cap \Sigma^-$ and $y\in \Sigma'\cap \Sigma^-$ with $x<y$, then $y\in J^+(\Sigma)$ and $y\in J^-(\Sigma)$ as $y\in \Sigma^-$, thus $y\in \Sigma$, which contradicts the acausality of $\Sigma$.) Now,
$
d(p,\Sigma^-)\le d(p,\Sigma)\le d(p,\Sigma^+)
$
and similarly for $\Sigma'$ in place of $\Sigma$. Since every point of $\Sigma^-$ and of $\Sigma^+$ belongs to $\Sigma\cup\Sigma'$ and $\Sigma$, $\Sigma'$ are achronal, $d(\Sigma^-,\Sigma^+)\le R$. Consider the causal curve $\sigma: [0,1]\to M$ that maximizes the distance between $p=\sigma(0)$ and $\Sigma^+$. It intersects $\Sigma^-$ at some point $\sigma(a)$ and
\begin{align*}
d(p,\Sigma^+)-d(p,\Sigma^-)&\le \ell(\sigma\vert_{[0,1]})-\ell(\sigma\vert_{[0,a]})\\
&=\ell(\sigma\vert_{[a,1]})\le d(\Sigma^-,\Sigma^+)\le R.
\end{align*}
As a consequence, $d(p,\Sigma)\le d(p,\Sigma^+)\le d(p,\Sigma^-)+R\le d(p,\Sigma')+R$ and exchanging the roles of $\Sigma$ and $\Sigma'$ we get the desired result.

In the second case $d(p,\Sigma')=0$ and $d(p,\Sigma)\le R$ because the maximizing causal curve between $p$ and $\Sigma$ can be prolonged to the past to a causal curve connecting $\Sigma'$ and $\Sigma$. The third case is similar with $\Sigma'$ and $\Sigma$ switched. In the fourth case $d(p,\Sigma)=d(p,\Sigma')=0$ and the inequality is clear.
\end{proof}

The effective domain of $b$ does not depend on the choice of the compact Cauchy hypersurface.

\begin{theorem} \label{xxxt}
Assume that $(M,g)$ is future timelike complete. Let $\Sigma$ and $\Sigma'$ be compact Cauchy hypersurfaces, let $b$ and $b'$ be their Cauchy Busemann functions, and let $R:=R(\Sigma,\Sigma')$ be as in Lemma \ref{lemRdist}. Then $\textrm{dom}\, b=\textrm{dom}\, b'$ and $\vert b-b'\vert \le R$ on this set. In particular, the past set $\textrm{dom}\, b$ and its boundary $\p\, \textrm{dom}\, b$ do not depend on $\Sigma$.
\end{theorem}

The time dual statement holds for the past Cauchy Busemann functions of Section \ref{secsplit} under past timelike completeness.

\begin{proof}
By future timelike completeness, $b$ and $b'$ are defined and, for every $\tau>0$, $\Sigma_\tau$ and $\Sigma'_\tau$ are compact Cauchy hypersurfaces (Theorem \ref{kvvb}).

Let $0<a\le c$ and let $p$ be such that $d(\Sigma',p)\in [a,c]$. The function $d(\Sigma',\cdot)$ is continuous (Theorem \ref{thmcont}). Along a maximizing causal curve from $\Sigma'$ to $p$ it takes all the values in $[0,d(\Sigma',p)]$, thus $p\in J^+(\Sigma'_a)$. Along a future complete timelike geodesic $\gamma$ starting from $p$, which exists by future timelike completeness, it goes to $+\infty$, as $d(\Sigma',\gamma(t))\ge d(\Sigma',p)+t$, thus $p\in J^-(\Sigma'_c)$. Similarly with $\Sigma$ in place of $\Sigma'$.

Let $\tau>R$ be a real number. Recall that $\Sigma_\tau$ consists of the points at distance $\tau$ from $\Sigma$. By Lemma \ref{lemRdist}, if $p\in \Sigma_\tau$ then $d(\Sigma',p)\in [\tau-R,\tau+R]$, hence, by the previous paragraph, $\Sigma_\tau\subset J^+(\Sigma'_{\tau-R})\cap J^-(\Sigma'_{\tau+R})$. Analogously, $\Sigma'_\tau\subset J^+(\Sigma_{\tau-R})\cap J^-(\Sigma_{\tau+R})$.
Every causal curve from $\Sigma_\tau$ to $\Sigma'_\tau$ can be prolonged to a causal curve from $\Sigma_\tau$ to $\Sigma_{\tau+R}$, so its Lorentzian length is at most $R$ by Proposition \ref{nner}. Similarly, every causal curve from $\Sigma'_\tau$ to $\Sigma_\tau$ has Lorentzian length at most $R$, and so
\[
R(\Sigma_\tau,\Sigma'_\tau)\le R.
\]
Since $\Sigma_\tau$ and $\Sigma'_\tau$ are compact Cauchy hypersurfaces, we can apply Lemma \ref{lemRdist} to them, and we get that for every $p\in M$,
$
\vert d(p,\Sigma_\tau)-d(p,\Sigma'_\tau)\vert\le R.
$

But $b_\tau(p)=\tau-d(p,\Sigma_\tau)$ and $b'_\tau(p)=\tau-d(p,\Sigma'_\tau)$, thus for every $p\in M$ we have $\vert b_\tau(p)-b'_\tau(p)\vert\le R$ and taking the limit $\tau\to\infty$ we get that one function is $+\infty$ iff  the other is, hence $\textrm{dom}\, b=\textrm{dom}\, b'$, and for each $p$ in this common domain we have $\vert b(p)-b'(p)\vert\le R$.
\end{proof}

Thus, under future timelike completeness, $\textrm{dom}\, b$ and its boundary are natural objects of the spacetime $(M,g)$, which do not involve any auxiliary choice, and hence carry physical significance. We call $\textrm{dom}\, b$ the {\em future cosmological development} and $\p\, \textrm{dom}\, b$ the {\em future cosmological horizon}. We shall prove that $\textrm{dom}\, b$ is open (Theorem \ref{xxhg}(ii)) and that $\p\, \textrm{dom}\, b$ is a $C^0$ future null hypersurface (Proposition \ref{propbdry}). An observer that traverses the horizon from $\textrm{dom}\, b$ cannot return to $\textrm{dom}\, b$, because $\textrm{dom}\, b$ is a past set (Proposition \ref{cmol}), so that $M\backslash \textrm{dom}\, b$ is a future set. In this respect the horizon behaves as the event horizon of a black hole and, under TCC, or under NCC and future null completeness, it also satisfies an area theorem, see the companion paper \cite{minguzzi26g}.

\section{The product trick} \label{secprod}

There is an important product spacetime interpretation that becomes apparent after introducing $\tilde M=M\times \mathbb{R}$ endowed with the metric $\tilde g=g+\dd z^2$ with the time orientation given by the vector field $(V,0)$ where $V$ is a vector field providing the time orientation for $(M,g)$. Causal relations and sets on $\tilde M$ will be denoted with a tilde, e.g.\ $\tilde J$, $\tilde I$.

An absolutely continuous curve $x: I \to M$ is said to be {\em causal} if $\dot x$ is causal almost everywhere. The absolutely continuous causal curves coincide with the, suitably parametrized, continuous causal curves of Hawking and Ellis \cite[Thm.\ 2.12]{minguzzi18b}. Indeed, they can be parametrized to be locally Lipschitz. When using the term {\em causal curve} we always mean piecewise $C^1$, otherwise we shall most often use the term continuous causal curve when the parametrization is not important. Low regularity causal curves are important for applications in the limit curve theorem. Timelike curves are considered only in piecewise $C^1$ regularity. Less regular notions are not particularly useful. Various types of causal sets (causal futures, timelike futures, domains of dependence) and also the Lorentzian distance can be defined using $C^1$ (timelike/causal) curves.

Observe that a ($C^1$) curve $t \to \tilde x(t)=(x(t),z(t))$ is causal iff $\dot x$ is causal and $g(\dot x, \dot x)+\dot z^2\le 0$, which is equivalent to: $x$ is a causal curve and  $\vert \dot z\vert\le \sqrt{- g(\dot x,\dot x)}$ (so $x$ is timelike wherever $\dot z\ne 0$, while it is merely causal where $\dot z=0$). In particular, for $\tilde x$ causal we have necessarily  $z(t)\le z(0)+\ell(x\vert_{[0,t]})$. Conversely, for every  causal curve $x$ on $M$ and every $z_0\in \mathbb{R}$ the curve $t\mapsto (x(t), z_0+\ell(x\vert_{[0,t]}))$ is a lightlike (hence causal) curve on $\tilde M$, called the {\em lightlike lift} of $x$.
Similarly, a  curve $t \to \tilde x(t)=(x(t),z(t))$ is timelike iff  $x$ is timelike and  $\vert \dot z\vert< \sqrt{- g(\dot x,\dot x)}$. In particular, for $\tilde x$ timelike we have necessarily  $z(t)< z(0)+\ell(x\vert_{[0,t]})$.

The reason we consider only this direction of the inequality is that we shall only consider sets $\tilde S\subset \tilde M$ with the property: $(p,z)\in \tilde S$ implies $(p,z+a)\in \tilde S$ for every $a\le 0$. As a consequence, we shall only be interested in their top boundaries.

\begin{lemma} \label{lemgh}
If $\tilde x=(x,z)$ is a future inextendible causal curve on $\tilde M$ then $x$ is a future inextendible causal curve on $M$. The spacetime $(\tilde M,\tilde g)$ is globally hyperbolic with Cauchy hypersurface $\Sigma\times \mathbb{R}$.
\end{lemma}

\begin{proof}
If $x$ had a future endpoint $e$ then $\ell(x\vert_{[0,t)})\le d(x(0),e)<\infty$, so by $\vert \dot z\vert\le \sqrt{- g(\dot x,\dot x)}$ the function $z$ would have bounded variation, hence a limit, and $\tilde x$ would be extendible. The same argument in the past direction shows that projections of inextendible causal curves are inextendible, so they meet $\Sigma$, namely every inextendible causal curve on $\tilde M$ meets $\Sigma\times \mathbb{R}$. A causal curve on $\tilde M$ cannot have constant projection since the vector $(0,\dot z)$ is spacelike, thus the projection of a causal curve joining two points of $\Sigma\times \mathbb{R}$ would be a non-trivial causal curve joining two points of $\Sigma$, which is impossible by acausality. Similarly, $\tilde M$ is causal because $M$ is. Thus $\Sigma\times \mathbb{R}$ is an acausal set met by every inextendible causal curve, namely a Cauchy hypersurface, and $\tilde M$ is globally hyperbolic.
\end{proof}

Let $\tilde x: I \to \tilde M$ be a geodesic. Since $\p_z$ is Killing, $\tilde g (\p_z, \dot{\tilde x})=\dot z$ is conserved. Hence, every causal geodesic on $\tilde M$ reads $\tilde x(t)=(x(t), a+b t)$ where $x$ is a causal geodesic of $(M,g)$, $a\in \mathbb{R}$ and $b^2 \le - g(\dot x,\dot x)$ (a constant along $x$). Moreover, if $x$ is lightlike then $\tilde x$ can only be lightlike and $z$ is constant, namely $\tilde x$ is the {\em horizontal lift} of $x$.

\begin{proposition} \label{nhtr}
The projection of a maximizing causal geodesic of $\tilde M$ is a maximizing causal geodesic of $M$.
\end{proposition}

This observation will not be used below, but it is a useful property of the product.

\begin{proof}
Let $\tilde x:[a,b]\to \tilde M$, $\tilde x(t)=(x(t),c+dt)$, be an affinely parametrized maximizing causal geodesic, so that $d^2 \le - g(\dot x,\dot x)=[\ell(x)/(b-a)]^2$. Suppose that $x$ is not maximizing. Since $M$ is globally hyperbolic there is a maximizing timelike geodesic $y: [a,b]\to M$, $y(a)=x(a)$, $y(b)=x(b)$, affinely parametrized,  with $\ell(y)> \ell(x)$. Then $-g(\dot y,\dot y)=[{\ell(y)}/(b-a)]^2$ and the curve $\tilde y(t)=(y(t), c+ d t)$ is causal, has the same endpoints as $\tilde x$  and has Lorentzian length
\[
\ell(\tilde y)=\int_a^b \sqrt{[{\ell(y)}/(b-a)]^2 -d^2}\, \dd t>  \int_a^b \sqrt{[{\ell(x)}/(b-a)]^2 -d^2}\, \dd t=\ell(\tilde x),
\]
which contradicts the maximality of $\tilde x$.
\end{proof}

We recall that a set $F$ is a {\em future set} if $I^+(F)\subset F$, which is the time dual of the notion of past set recalled before Theorem \ref{cbdc}. If $F$ is a future set then $J^+(\overline F)\subset \overline F$, and if $F$ is open and future, then $I^+(F)=F$ \cite[Prop.\ 2.84]{minguzzi18b}, and hence also $J^+(F)=F$, since $J^+(I^+(F))=I^+(F)$. The same terminology is used on $\tilde M$.

\begin{proposition} \label{propprod}
 Let $\tilde \Sigma=\{(p,z): p\in \Sigma, z \le 0\}$. Then the following identities hold, the last two under the assumption that $(M,g)$ is future timelike complete:
 \begin{align}
 \tilde J^+(\tilde \Sigma)&=\{(p,z): p \in J^+( \Sigma), z\le d^+_\Sigma(p)\}, \label{kcrj}\\
 \tilde I^+(\tilde \Sigma)&=\{(p,z): p \in I^+(\Sigma), z< d^+_\Sigma(p)\}, \label{kcrt}\\
 \tilde D^{-}(\tilde I^+(\tilde \Sigma))
&=\{(p,z): p \in M, z< b(p)\}, \label{mxwe} \\
\tilde H^{-}(\tilde I^+(\tilde \Sigma))&= \p \operatorname{epi}  b.\label{mxwh}
\end{align}
Moreover, without any completeness assumption, $\tilde D^{-}(\tilde I^+(\tilde \Sigma))$ is a future open set and $\tilde H^{-}(\tilde I^+(\tilde \Sigma))$ is its  boundary, hence an achronal boundary (thus edgeless). The set $\tilde D^{-}(\tilde I^+(\tilde \Sigma))$  contains $\tilde I^+(\tilde \Sigma)$.
\end{proposition}

\begin{remark} \label{remfol}
A fundamental observation is that formula (\ref{mxwe}) characterizes the Busemann function without any reference to the auxiliary foliation by the level sets $\Sigma_s$ used in its construction, so that $b$ is a natural object which depends solely on the starting hypersurface $\Sigma$. This does not mean that replacing $\Sigma_\tau$ by the level sets of an arbitrary temporal function in the expression $\tau-d(\cdot,\Sigma_\tau)$ would give the same limit: in a static product $-\dd t^2+h$ with $\Sigma=\{t=0\}$, the choice of the slices $\{t=2\tau\}$ gives $\tau-d(\cdot,\{t=2\tau\})\to -\infty$. Rather, the normalization provided by the Lorentzian distance from $\Sigma$ is the natural one, as the foliation-free characterizations of Proposition \ref{propprod} and Theorem \ref{thmintr} show.
\end{remark}

Note that the epigraph of $b$ is closed since $b$ is lower semi-continuous. For shortness, we set
\[
\tilde D:=\tilde D^{-}(\tilde I^+(\tilde \Sigma)).
\]

\begin{proof}[Proof of Proposition \ref{propprod}]
Let $(p,r) \ge (q,a)$, $a\le 0$, $q\in \Sigma$. This means that there is a causal curve $\tilde x: [0,1]\to \tilde M$, $\tilde x(t)=(x(t),z(t))$, $(x(0),z(0))=(q,a)$, $(x(1),z(1))=(p,r)$. Thus $x$ is causal which implies $p\in J^+(q)\subset J^+(\Sigma)$. By the formula proved above $r=z(1)\le z(0)+\ell(x)\le a+d^+_\Sigma(p)\le d^+_\Sigma(p)$. This proves the inclusion $\subset$ in (\ref{kcrj}).
For the other direction, suppose $p\ge q\in \Sigma$ and $r\le d^+_\Sigma(p)$. Let $x$ be a maximizing causal curve from $\Sigma$ to $p$, $\ell(x)=d^+_\Sigma(p)$, starting from $q'\in \Sigma$. Its lightlike lift  connects $(q',0)$ to $(p,d^+_\Sigma(p))$ and so, by the translational invariance of the cone structure of $\tilde M$, $(q',r-d^+_\Sigma(p)) \le (p,r)$, where  $(q',r-d^+_\Sigma(p))\in \tilde \Sigma$.

Let $(p,r) \gg (q,a)$, $a\le 0$, $q\in \Sigma$. Then there is a timelike curve $\tilde x$ from $(q,a)$ to $(p,r)$ whose projection $x$ is timelike, which implies $p\in I^+(q)\subset I^+(\Sigma)$ and $r< a+\ell(x)\le d^+_\Sigma(p)$. This proves the inclusion $\subset$ in (\ref{kcrt}).
For the other direction, suppose $p\in I^+(\Sigma)$ and $r< d^+_\Sigma(p)$. There is a maximizing affinely parametrized timelike geodesic $x:[0,1]\to M$ connecting some $q\in \Sigma$ to $p$, of Lorentzian length $\ell(x)=d^+_\Sigma(p)$. Observe that $g(\dot x,\dot x)=-(d^+_\Sigma(p))^2$, thus the curve $\tilde x(t):=(x(t), a+ \alpha t d^+_\Sigma(p))$ is timelike for every $a\le 0$ and $\alpha\in [0,1)$ and starts from  $\tilde \Sigma$. For $r\le 0$ it is sufficient to choose $a=r$, $\alpha=0$. For $0<r<d^+_\Sigma(p)$ it is sufficient to choose $a=0$, $\alpha=r/d^+_\Sigma(p)$.

{\em $\tilde D$ is open.} Its complement $C$ is closed. Indeed, let $P_k\in C$, $P_k\to P$. There are future-inextendible causal curves $\tilde x_k$ starting from $P_k$ and not intersecting the open set $\tilde I^+(\tilde \Sigma)$, hence contained in the closed set $\tilde M\backslash \tilde I^+(\tilde \Sigma)$. As $\tilde M$ is globally hyperbolic (Lemma \ref{lemgh}), by the limit curve theorem there is a future-inextendible  causal curve $\tilde x$ starting from $P$ and contained in $\tilde M\backslash \tilde I^+(\tilde \Sigma)$, which proves  $P\in C$.

{\em Proof of $\supset$ in (\ref{mxwe}).} Let $(q,u)$ be such that $u<b(q)$. Since $b(q)=\sup_k b_k(q)$ there is $k>0$ such that $u<b_k(q)=k-d(q,\Sigma_k)$. Let $\tilde x=(x,z)$ be a future inextendible causal curve starting from $(q,u)$. If $q\in I^+(\Sigma_k)$ then $d^+_\Sigma(q)>k>u$ so $\tilde x(0)\in \tilde I^+(\tilde \Sigma)$ by (\ref{kcrt}). Otherwise $q\in J^-(\Sigma_k)$ and by Lemma \ref{lemgh} $x$ is future inextendible, thus it intersects $\Sigma_k$ at some point $x(\tau)$. Then $z(\tau)\le u+\ell(x\vert_{[0,\tau]})\le u+d(q,\Sigma_k)<k=d^+_\Sigma(x(\tau))$, thus $\tilde x(\tau)\in \tilde I^+(\tilde \Sigma)$. By the arbitrariness of $\tilde x$, $(q,u)\in \tilde D$.

{\em Proof of $\subset$ in (\ref{mxwe}).} Let $(p,r)\in \tilde D$. As $\tilde D$ is open there is $\epsilon>0$ such that $(p,r+\epsilon)\in \tilde D$. It is sufficient to prove $r+\epsilon\le b(p)$. Suppose by contradiction that $r+\epsilon> b(p)$, so that $r+\epsilon>b_k(p)=k- d(p,\Sigma_k)$ for every $k$. For large $k$, $p\in I^-(\Sigma_k)$ and there is a maximizing causal curve $\sigma_k:[0,d_k]\to M$ ($h$-arc length parametrization with $h$ complete) from $p$ to $\Sigma_k$, $\ell(\sigma_k)=d(p,\Sigma_k)$. As $\Sigma_k$ escapes every compact set, by the limit curve theorem \cite[Thm.\ 2.51, Rem.\ 2.52]{minguzzi18b} a subsequence, not relabeled, converges uniformly on compact subsets to a future-inextendible causal curve $\sigma$ starting from $p$. We claim that the lightlike lift $\tilde x(t)=(\sigma(t), r+\epsilon+\ell(\sigma\vert_{[0,t]}))$, which starts from $(p,r+\epsilon)$ and is future inextendible, does not intersect $\tilde I^+(\tilde \Sigma)$, which gives the desired contradiction.

Suppose that $\tilde x(t)\in \tilde I^+(\tilde \Sigma)$, namely $\sigma(t)\in I^+(\Sigma)$ and $r+\epsilon+\ell(\sigma\vert_{[0,t]})< d^+_\Sigma(\sigma(t))$. Since the $\sigma_k$ are maximizing, $\ell(\sigma_k\vert_{[0,t]})\to \ell(\sigma\vert_{[0,t]})$ \cite[Thm.\ 2.58]{minguzzi18b}, thus, by the continuity of $d^+_\Sigma$ (Theorem \ref{thmcont}), for sufficiently large $k$
\[
r+\epsilon+\ell(\sigma_k\vert_{[0,t]})< d^+_\Sigma(\sigma_k(t)).
\]
Since $\sigma_k$ is maximizing from $p$ to $\Sigma_k$, its restriction to $[t,d_k]$ maximizes the distance from $\sigma_k(t)$ to $\Sigma_k$, thus $\ell(\sigma_k\vert_{[0,t]})=d(p,\Sigma_k)-d(\sigma_k(t),\Sigma_k)$. Using $r+\epsilon>k- d(p,\Sigma_k)$ we get
\begin{align*}
k&< r+\epsilon+d(p,\Sigma_k)=r+\epsilon+\ell(\sigma_k\vert_{[0,t]})+d(\sigma_k(t),\Sigma_k)\\
&<d^+_\Sigma(\sigma_k(t))+d(\sigma_k(t),\Sigma_k)\le d(\Sigma,\Sigma_k)=k,
\end{align*}
a contradiction.

The set $\tilde D$ contains $\tilde I^+(\tilde \Sigma)$ due to  (\ref{cmlr}), (\ref{kcrt}) and (\ref{mxwe}).

The set $\tilde D$ is a future set. Indeed, let $P\in \tilde D$ and $Q\gg P$. If it were $Q\notin \tilde D$ there would be a future inextendible causal curve $\tilde \sigma$ starting at $Q$ not intersecting $\tilde I^+(\tilde \Sigma)$. Let $\tilde \gamma$ be a timelike curve from $P$ to $Q$. It does not intersect $\tilde I^+(\tilde \Sigma)$ otherwise $Q$  would belong to it, as $\tilde I^+(\tilde \Sigma)$ is future, and so would $\tilde \sigma$. Thus, joining $\tilde \gamma$ and $\tilde \sigma$ we get a future inextendible causal curve  starting from $P$ not intersecting $\tilde I^+(\tilde \Sigma)$, a contradiction.

By definition the past Cauchy horizon is
\[
\tilde H^{-}(\tilde I^+(\tilde \Sigma))= \overline{\tilde D}\backslash \tilde I^+(\tilde D)= \overline{\tilde D}\backslash \tilde D= \p \tilde D.
\]
It is the boundary of $\tilde D$, hence the  boundary of  $\tilde M\backslash \tilde D$ which, by (\ref{mxwe}), is the epigraph of $b$. As $\tilde D$ is a future set, its boundary is an achronal boundary, hence edgeless.
\end{proof}

\begin{theorem}[Foliation-free formula] \label{thmintr}
Assume that $(M,g)$ is future timelike complete. For every future inextendible causal curve $x:[0,+\infty)\to M$ the function
\[
t\mapsto F_x(t):=d^+_\Sigma(x(t))-\ell(x\vert_{[0,t]}),
\]
defined for those $t$ such that $x(t)\in J^+(\Sigma)$, is non-decreasing, so the limit $L(x):=\lim_{t\to +\infty} F_x(t)\in (-\infty,+\infty]$ exists. For every $p\in M$
\begin{equation} \label{intrinsic}
b(p)=\inf \{L(x): x \textrm{ future inextendible causal curve starting from } p\}.
\end{equation}
Moreover, if $p\in \textrm{dom}\, b$, a future inextendible causal curve $x$ starting from $p$ attains the infimum, $L(x)=b(p)$, if and only if $b$ increases by the Lorentzian length along $x$, namely $b(x(t_2))=b(x(t_1))+\ell(x\vert_{[t_1,t_2]})$ for every $t_1<t_2$.
\end{theorem}

The last property for $x$ is the $b$-calibration property of the subsequent Definition \ref{defcal} and Proposition \ref{bber}, which also shows that such curves exist.

The quantity $L(x)$ can be interpreted as the asymptotic delay of a clock carried along $x$ with respect to the ``fastest'' clocks leaving $\Sigma$; $b(p)$ is the minimal such delay among all observers leaving $p$.
Formula (\ref{intrinsic}) does not involve the level sets $\Sigma_k$ and, in fact, it does not even require timelike completeness to make sense.

\begin{proof}
The set $J^+(\Sigma)$ is a future set, and every future inextendible causal curve enters it and never leaves it. If $x(t),x(t')\in J^+(\Sigma)$, $t<t'$, then $d^+_\Sigma(x(t'))\ge d^+_\Sigma(x(t))+d(x(t),x(t'))\ge d^+_\Sigma(x(t))+\ell(x\vert_{[t,t']})$, which proves the monotonicity of $F_x$.

Let $z< \inf_x L(x)$ and let $\tilde x=(x,z(\cdot))$ be a future inextendible causal curve starting from $(p,z)$. By Lemma \ref{lemgh} $x$ is future inextendible. As $z<L(x)$ and $F_x$ is non-decreasing, for sufficiently large $t$ we have $x(t)\in I^+(\Sigma)$ and $z+\ell(x\vert_{[0,t]})<d^+_\Sigma(x(t))$, thus $z(t)\le z+\ell(x\vert_{[0,t]})< d^+_\Sigma(x(t))$, namely $\tilde x(t)\in \tilde I^+(\tilde \Sigma)$ by (\ref{kcrt}). Thus $(p,z)\in \tilde D$ and by (\ref{mxwe}) $z<b(p)$. We conclude $\inf_x L(x)\le b(p)$.

Conversely, let $z<b(p)$, so that $(p,z)\in \tilde D$, and let $x$ be a future inextendible causal curve starting from $p$. Its lightlike lift $(x(t), z+\ell(x\vert_{[0,t]}))$ is future inextendible and so it intersects $\tilde I^+(\tilde \Sigma)$, namely there is $t$ such that $x(t)\in I^+(\Sigma)$ and $z<F_x(t)\le L(x)$. By the arbitrariness of $z$ and $x$, $b(p)\le \inf_x L(x)$.

Finally, let $p\in \textrm{dom}\, b$ and let $x$ be a future inextendible causal curve starting from $p$. If $b$ increases by the Lorentzian length along $x$, then for $x(t)\in J^+(\Sigma)$, by (\ref{cmlr}), $F_x(t)=d^+_\Sigma(x(t))-b(x(t))+b(p)\le b(p)$, thus $L(x)\le b(p)$ and hence $L(x)=b(p)$. Conversely, suppose $L(x)=b(p)$ and let $t\ge 0$. The curve $x_t:=x\vert_{[t,+\infty)}$ satisfies $F_{x_t}=F_x+\ell(x\vert_{[0,t]})$, hence $L(x_t)=b(p)+\ell(x\vert_{[0,t]})$, and by (\ref{intrinsic}) applied at $x(t)$ and by (\ref{fcpb})
\[
b(p)+\ell(x\vert_{[0,t]})\le b(p)+d(p,x(t))\le b(x(t))\le L(x_t)=b(p)+\ell(x\vert_{[0,t]}).
\]
Thus $b(x(t))=b(p)+\ell(x\vert_{[0,t]})$ for every $t\ge 0$, which gives the claim by taking differences.
\end{proof}

\begin{remark} \label{remasym}
If $x$ is a proper time parametrized $b$-calibrated ray starting from $p\in \textrm{dom}\, b$, then $d^+_\Sigma(x(t))-t\to b(p)$ and $b(x(t))=b(p)+t$, thus
\[
b(x(t))-d^+_\Sigma(x(t))\to 0 \quad \textrm{for } t\to +\infty .
\]
Thus, asymptotically along calibrated rays, $b$ coincides with the distance from $\Sigma$.
\end{remark}

\section{Horizon, calibrated curves and regularity} \label{secreg}

A {\em $C^0$ future null hypersurface} \cite{galloway00} is a locally achronal topological hypersurface such that  through every point in it starts  a future inextendible lightlike geodesic contained in the hypersurface.
For shortness, let us set $H:=\tilde H^{-}(\tilde I^+(\tilde \Sigma))=\p \tilde D$.

\begin{proposition} \label{cnph}
The Cauchy horizon $H$  is a $C^0$ future null hypersurface. Indeed, for every $P\in H$ we can find an achronal future-inextendible lightlike geodesic $\tilde x$ starting from $P$ and entirely contained in $H$.
\end{proposition}

The geodesics mentioned in this statement are  called {\em generators} of the horizon. Observe that if prolonged in the past direction they can escape the horizon.
The proof is a standard limit curve argument but it is repeated here because $H$ is not of the familiar form $H^-(S)$ with $S$ closed and achronal.
\begin{proof}
The set $H$ is an achronal boundary, hence a locally Lipschitz topological hypersurface.
Let $P\in H$.
Since $H=\p \tilde D$ and $\tilde D$ is open, $P\notin \tilde D$ and there are future-inextendible causal curves  starting from $P$ and not intersecting $\tilde I^+(\tilde \Sigma)$. (Alternatively, one can take a sequence $P_k\to P$, $P_k\notin \tilde D$, and use the limit curve theorem.) Let $\tilde x$ be one of them. It does not intersect $\tilde D$, since the portion of $\tilde x$ after a point of $\tilde D$ would have to intersect $\tilde I^+(\tilde \Sigma)$. Suppose  $\tilde x$ is not contained in $H$, so that there is some point $Q\in \tilde x\backslash  \overline{\tilde D}$. Let $O$ be a neighborhood of $Q$ disjoint from $\overline{\tilde D}$ and let $R\in O$, $R\gg Q$. Since $P\le Q\ll R$ we have $P\ll R$, thus $\tilde I^-(R)$ is an open neighborhood of $P\in \p \tilde D$ and so it contains some point $P'\in \tilde D$. As $\tilde D$ is a future set and $R\gg P'$, $R\in \tilde D$, a contradiction with $R\in O$. Thus $\tilde x$ is entirely contained in $H$. Since $H$ is achronal, $\tilde x$ is achronal and hence, once reparametrized, a lightlike geodesic.
\end{proof}

Future-inextendible maximizing causal curves are also called future causal rays. Rays are not necessarily affinely parametrized. To stress that the curve is affinely parametrized we might also use the word geodesic.

The following notion comes from weak KAM theory.
\begin{definition} \label{defcal}
Let $u: M\to [-\infty, +\infty]$  be a function.
A causal curve $x: I \to M$ with image in $\textrm{dom}\, u:=u^{-1}(\mathbb{R})$ is said to be {\em $u$-calibrated} if for every $t_1,t_2\in I$, $t_2>t_1$,
\[
u(x(t_2))=u(x(t_1))+\ell(x\vert_{[t_1,t_2]}).
\]
\end{definition}
Here $\textrm{dom}\, u$ is the effective domain of $u$, namely the subset of $M$ where it is finite. The requirement that calibrated curves be contained in $\textrm{dom}\, u$ is essential: with the usual conventions on infinite values, every causal curve contained in $\{u=+\infty\}$ would otherwise satisfy the calibration identity.
The following result is immediate.

\begin{proposition}
Let $u: M\to [-\infty, +\infty]$  be a function and let $x:I\to \textrm{dom}\, u$ be a causal curve. Then $x$ is $u$-calibrated iff its lightlike lift $\tilde x(t)=(x(t), u(x(a)) +\ell(x\vert_{[a,t]}))$, $t\in I$, (for some and hence every $a\in I$) has image in $\textrm{graph}\, u$.
\end{proposition}

Recall that a function $u: M\to [-\infty, +\infty]$ is rushing if $u(p)+d(p,q)\le u(q)$ for every $p,q\in M$, $p\le q$ (Definition \ref{defrush}).

\begin{proposition}
The function $u: M\to [-\infty, +\infty]$ is rushing iff the set $\tilde U\subset \tilde M$, $\tilde U=\{(p,z): z\le u(p)\}$ satisfies $\tilde J^+(\tilde U)\subset \tilde U$.
\end{proposition}

The graph is defined in the usual way over $\textrm{dom}\, u$. Observe that for $u$ rushing, $\textrm{graph}\, u$ is included in the achronal boundary\footnote{An achronal boundary is the boundary of a future set.} $\p \tilde U$ because for each $p\in \textrm{dom}\, u$, $(p,u(p))\in \tilde U$, while for each $\epsilon>0$, $(p,u(p)+\epsilon)\notin \tilde U$.

\begin{proof}
Suppose $u$ is rushing and let $(p,a)\in \tilde U$, $(q,c)\ge (p,a)$. There is a causal curve $\tilde x=(x,z)$ from $(p,a)$ to $(q,c)$. The projected curve $x$ is causal and $c\le a+\ell(x)\le a+d(p,q)\le u(p)+d(p,q)\le u(q)$, which proves that $(q,c)\in \tilde U$.
Conversely, suppose that $\tilde U$ satisfies $\tilde J^+(\tilde U)\subset \tilde U$, and let $p\le q$. If $u(p)=-\infty$ there is nothing to prove, so let us assume $u(p)>-\infty$. Let $a\le u(p)$ be a finite number.
For every causal curve $x$ from $p$ to $q$ the lightlike lift of $x$ starting at $(p,a)$ ends at $(q,a+\ell(x))$, thus  $(q,a+\ell(x))\in \tilde U$, which implies $a+\ell(x)\le u(q)$. Taking the supremum over $x$ and $a$ we get $u(p)+d(p,q)\le u(q)$.
\end{proof}

\begin{proposition} \label{jjtr}
Let $u: M\to [-\infty, +\infty]$ be rushing and $x:[a,b]\to \textrm{dom}\, u$ a causal curve such that $u(x(b))= u(x(a))+\ell(x\vert_{[a,b]})$. Then $x$ is maximizing and $u$-calibrated.
\end{proposition}

\begin{proof}
This follows from
$
u(x(b))= u(x(a))+\ell(x)\le u(x(a))+d(x(a),x(b))\le u(x(b))
$
which proves $\ell(x)=d(x(a),x(b))$, namely $x$ is maximizing. Thus for every $t_1<t_2$,
$\ell(x\vert_{[t_1,t_2]})=d(x(t_1),x(t_2))$, so
$
u(x(t_2))\ge u(x(t_1))+\ell(x\vert_{[t_1,t_2]}).
$
The same inequality holds on $[a,t_1]$ and on $[t_2,b]$; if one of them were strict, summing we would get $u(x(b))>u(x(a))+\ell(x)$, a contradiction.
\end{proof}

Let us return to the rushing function $b$.

\begin{proposition} \label{bber}
Assume that $(M,g)$ is future timelike complete. Let $p\in \textrm{dom}\, b$. There is a generator of $H$ starting from $(p,b(p))$, and the projection of every such generator is a future inextendible causal geodesic $x: [0,a)\to M$, $a\in (0,+\infty]$, starting from $p$ which is maximizing, hence a causal ray, and $b$-calibrated, in the sense that it is  entirely contained in $\textrm{dom}\, b$ and for every $t_1<t_2$
\[
b(x(t_2))=b(x(t_1))+d(x(t_1),x(t_2)).
\]
In particular, from $p$ starts a future inextendible $b$-calibrated causal ray. Conversely, the lightlike lift of every future-inextendible $b$-calibrated causal curve $x$ starting from $p$, namely $t\mapsto (x(t),b(x(t)))$, is a generator of $H$.
No future inextendible $b$-calibrated causal curve can be intersected in the interior of its domain by another one, namely, if $x_1:[0,a_1)\to M$ and $x_2:[0,a_2)\to M$ are future inextendible $b$-calibrated causal curves and $x_2(s_2)=x_1(s_1)$ for some $s_1\in (0,a_1)$ and $s_2\in [0,a_2)$, then $x_2\vert_{[s_2,a_2)}$ coincides with $x_1\vert_{[s_1,a_1)}$ up to reparametrization.
\end{proposition}

From one point can start many $b$-calibrated causal curves, as it happens for the generators of a $C^0$ future null hypersurface.

\begin{proof}
Let $p\in \textrm{dom}\, b$ and $c:=b(p)$. We have $(p,c)\in \textrm{epi}\, b$ while $(p,c-1/n)\in \tilde D$, thus $(p,c)\in \p \tilde D=H$. By Proposition \ref{cnph} there is a generator of $H$ starting from $(p,c)$. Let $\tilde x\subset H$ be any such generator, which, affinely parametrized, reads $\tilde x(t)=(x(t),c+dt)$, $t\in [0,a)$, where $x$ is a causal geodesic with $d^2= -g(\dot x,\dot x)$ and $a\in (0,+\infty]$, since at this stage the generator might be incomplete. By Lemma \ref{lemgh} $x$ is future inextendible. Since $\textrm{epi}\, b$ is closed, $H\subset \textrm{epi}\, b$, thus
\[
b(x(t))\le c+dt.
\]
On the other hand, by the rushing property,
\[
b(x(t))\ge b(p)+d(p,x(t))\ge c+\ell(x\vert_{[0,t]})=c+\vert d\vert t.
\]
The two inequalities give $d\ge 0$, $b(x(t))=c+dt<+\infty$ and $d(p,x(t))=\ell(x\vert_{[0,t]})$. By Proposition \ref{jjtr}, $x$ is maximizing and $b$-calibrated. This argument does not use the inclusion $H\subset \textrm{graph}\, b$, which is proved only in Theorem \ref{xxhg}. That $x$ is timelike with infinite proper time is proved in Theorem \ref{xxhg}(i).

Conversely, if $x$ is $b$-calibrated and future inextendible, its lightlike lift $(x(t),b(x(t)))$ lies in $\textrm{graph}\, b\subset \p \textrm{epi}\, b=H$. It is a causal curve contained in the achronal set $H$, hence a future inextendible (Lemma \ref{lemgh}) lightlike geodesic contained in $H$, namely a generator.

For the last statement, let $x_2(s_2)=x_1(s_1)$ with $s_1\in (0,a_1)$. The concatenation of the lift of $x_1\vert_{[0,s_1]}$ with the lift of $x_2\vert_{[s_2,a_2)}$ is a causal curve contained in the achronal set $H$, hence a lightlike geodesic without corners. Thus the lift of $x_2\vert_{[s_2,a_2)}$ continues that of $x_1\vert_{[0,s_1]}$ and, by uniqueness of geodesics, the two lifts coincide after $x_1(s_1)$, and so do $x_1$ and $x_2$.

\end{proof}

\begin{remark}
The previous proof shows that the affinely parametrized generators of $H$ starting from $\textrm{graph}\, b$ are of the form $\tilde x(t)=(x(t),b(x(0))+dt)$ where $x$ is a future-inextendible maximizing causal geodesic and $d=\sqrt{- g(\dot x,\dot x)}\ge 0$ is a constant.
\end{remark}

\begin{proposition}
Assume that $(M,g)$ is future timelike complete. The proper time-parametrized future $\Sigma$-rays are $b$-calibrated, hence the projections of generators of $H$.
\end{proposition}

\begin{proof}
We know that a future $\Sigma$-ray $\sigma$ is timelike. By Prop.\ \ref{cmol} $\sigma\subset I^-(\sigma)\subset \textrm{dom}\, b$ and by Prop.\ \ref{ncho}, $b(\sigma(t_2))=b(\sigma(t_1))+(t_2-t_1)$, where $t_2-t_1=\ell(\sigma\vert_{[t_1,t_2]})$.
\end{proof}

We shall use the following standard facts on the Lorentzian cut locus in globally hyperbolic spacetimes \cite[Ch.\ 9]{beem96}. Let $\gamma:[0,+\infty)\to M$ be a future inextendible timelike geodesic starting from $q$ and let $t_0:=\sup\{t: d(q,\gamma(t))=\ell(\gamma\vert_{[0,t]})\}$. If $t_0<+\infty$, the point $\gamma(t_0)$ is the {\em future cut point} of $q$ along $\gamma$, and the {\em future timelike cut locus} $\textrm{Cut}^+(q)$ is the set of future cut points of $q$ along the future directed timelike geodesics starting from $q$ \cite[Def.\ 9.3, 9.9]{beem96}. Thus $\gamma$ is maximizing up to its cut point and not beyond it, and for $0<t<t_0$ the segment $\gamma\vert_{[0,t]}$ is the unique maximizing curve from $q$ to $\gamma(t)$ \cite[Cor.\ 9.4]{beem96}. The cut point comes no later than the first conjugate point \cite[Thm.\ 9.10, Cor.\ 9.11]{beem96}. In a globally hyperbolic spacetime, the future cut point of $q$ along $\gamma$ is either the first future conjugate point of $q$ along $\gamma$, or it is joined to $q$ by at least two maximizing geodesics \cite[Thm.\ 9.12]{beem96}. Let
\[
\Omega:=\{(q,r)\in M\times M: q\ll r, \ r\notin \textrm{Cut}^+(q)\}.
\]
By these facts, $(q,r)\in \Omega$ if and only if $q\ll r$, there is exactly one maximizing geodesic $\gamma:[0,1]\to M$ from $q$ to $r$, and $r$ is not conjugate to $q$ along $\gamma$.

\begin{lemma} \label{lemomega}
Let $(M,g)$ be globally hyperbolic. The set $\Omega$ is open and $d$ is smooth on $\Omega$.
\end{lemma}

\begin{proof}
Let $(q_0,r_0)\in \Omega$, let $\gamma_0:[0,1]\to M$ be the maximizing geodesic from $q_0$ to $r_0$ and let $v_0=\dot\gamma_0(0)$. Since $r_0$ is not conjugate to $q_0$ along $\gamma_0$, the map $(q,v)\mapsto (q,\exp_q v)$ restricts to a diffeomorphism $\Phi$ from a neighborhood $A$ of $(q_0,v_0)$ in $TM$ onto a neighborhood $B$ of $(q_0,r_0)$, and we can assume that $A$ consists of future directed timelike vectors. We claim that there is a neighborhood $B'\subset B$ of $(q_0,r_0)$ such that for $(q,r)\in B'$ we have $q\ll r$ and the initial vector of every maximizing geodesic $\gamma:[0,1]\to M$ from $q$ to $r$ belongs to $A$, namely $(q,\dot\gamma(0))\in A$. Since the relation $\ll$ is open, the condition $q\ll r$ holds on a neighborhood of $(q_0,r_0)$. If the claim were false, there would be $(q_n,r_n)\to (q_0,r_0)$, $q_n\ll r_n$, and maximizing geodesics $\gamma_n:[0,1]\to M$ from $q_n$ to $r_n$ with $(q_n,\dot\gamma_n(0))\notin A$. Let $u_n:=\dot\gamma_n(0)/d(q_n,r_n)$ be the normalized initial vectors. By \cite[Lemma 9.6]{beem96}, every subsequence of $\{u_n\}$ has a further subsequence converging to the initial vector of a unit speed maximizing geodesic from $q_0$ to $r_0$, which by uniqueness is $u_0:=v_0/d(q_0,r_0)$. Thus $u_n\to u_0$ and, by the continuity of $d$, $\dot\gamma_n(0)=d(q_n,r_n)\,u_n\to v_0$, namely $(q_n,\dot\gamma_n(0))\to (q_0,v_0)$. As $A$ is an open neighborhood of $(q_0,v_0)$, we get $(q_n,\dot\gamma_n(0))\in A$ for large $n$, a contradiction. Such a geodesic is then unique, its endpoint $r$ is not conjugate to $q$ along it, and $d(q,r)=\sqrt{-g(v,v)}$, where $(q,v)=\Phi^{-1}(q,r)$. Thus $B'\subset \Omega$ and $d$ is smooth on $B'$.
\end{proof}

For the next result see also \cite[Lemma 14.26]{beem96}.

\begin{lemma} \label{lemeik}
Let $(M,g)$ be globally hyperbolic, let $q\in M$ and let $\Omega_q:=\{x: (x,q)\in \Omega\}$, an open set by Lemma \ref{lemomega}. On $\Omega_q$ the smooth function $r:=d(\cdot,q)$ satisfies the eikonal equation $g(\nabla r,\nabla r)=-1$, and $\nabla r(x)$ is the future directed unit tangent at $x$ of the maximizing geodesic from $x$ to $q$. Moreover, if $\alpha:[0,s']\to M$ is a maximizing timelike geodesic and $0<s<s'$, then $\alpha([0,s))\subset \Omega_{\alpha(s)}$, so that $d(\cdot,\alpha(s))$ is smooth and satisfies the eikonal equation on a neighborhood of $\alpha([0,s))$.
\end{lemma}

\begin{proof}
By the reverse triangle inequality, $-r(y)\ge -r(x)+d(x,y)$ for $x\le y$, $y\in J^-(q)$, so $-r$ is rushing on the open set $I^-(q)\supset \Omega_q$. A rushing function is steep at its points of differentiability \cite[Thm.\ 1.28]{minguzzi18b}, namely $N:=\nabla r$ is future directed timelike with $-g(N,N)\ge 1$ \cite[Thm.\ 1.23]{minguzzi18b}. Let $x\in \Omega_q$ and let $\gamma$ be the proper time parametrized maximizing geodesic from $x$ to $q$. Since $r(\gamma(t))=r(x)-t$, we have $-g(N,\dot\gamma(0))=1$, while the reverse Cauchy-Schwarz inequality gives $-g(N,\dot\gamma(0))\ge \sqrt{-g(N,N)}\ge 1$. Thus $-g(N,N)=1$ and, by the equality case of the reverse Cauchy-Schwarz inequality, $N=\dot\gamma(0)$.

For the last statement, let $t\in[0,s)$. The segment $\alpha\vert_{[t,s]}$ is the unique maximizing geodesic from $\alpha(t)$ to $\alpha(s)$ \cite[Lemma 9.1]{beem96}, and $\alpha(s)$ is not conjugate to $\alpha(t)$ along $\alpha$, otherwise $\alpha$ would not be maximizing on $[t,s']$ \cite[Thm.\ 9.10]{beem96}. Thus $(\alpha(t),\alpha(s))\in \Omega$, and the claim follows from the openness of $\Omega_{\alpha(s)}$ and from the first part.
\end{proof}

\begin{lemma} \label{lemunif}
Let $W\subset TM$ be a compact set such that $\exp(w)$ is defined and $(\pi(w),\exp(w))\in \Omega$ for every $w\in W$. Let $p\in M$ and consider a coordinate chart with domain $U\ni p$. Then there is a closed coordinate ball $C\subset U$ centered at $p$ such that the set
\[
T:=\{(q,\exp(w)): q\in C,\ w\in W,\ \pi(w)\in C\}
\]
is a compact subset of $\Omega$, and the same holds for every smaller concentric coordinate ball. In particular, for every $w\in W$ with $\pi(w)\in C$, we have $C\subset I^-(\exp(w))$, the function $d(\cdot,\exp(w))$ is smooth on a neighborhood of $C$, and there is a constant $G$, independent of $w$, which bounds on $C$ the first and second coordinate derivatives of all these functions.
\end{lemma}

\begin{proof}
Let $h$ be a complete Riemannian metric and $d_h$ its distance. We claim that there is $\rho>0$ such that $(q,\exp(w))\in \Omega$ for every $w\in W$ and every $q$ with $d_h(q,\pi(w))\le \rho$. Otherwise there would be $w_k\in W$ and $q_k\in M$ with $d_h(q_k,\pi(w_k))\le 1/k$ and $(q_k,\exp(w_k))\notin \Omega$. As $W$ is compact, passing to a subsequence $w_k\to w\in W$, so that $q_k\to \pi(w)$ and, by the continuity of the exponential map, $\exp(w_k)\to \exp(w)$. Thus $(q_k,\exp(w_k))\to (\pi(w),\exp(w))$, which belongs to the open set $\Omega$ (Lemma \ref{lemomega}), so that $(q_k,\exp(w_k))\in \Omega$ for large $k$, a contradiction.

Let $C\subset U$ be a closed coordinate ball centered at $p$ contained in the $d_h$-ball of radius $\rho/2$ centered at $p$. The set $\{(q,w)\in C\times W: \pi(w)\in C\}$ is compact, so its image $T$ under the continuous map $(q,w)\mapsto (q,\exp(w))$ is compact. It is contained in $\Omega$, since for $q\in C$ and $\pi(w)\in C$ we have $d_h(q,\pi(w))\le \rho$. These properties are clearly preserved by passing to a smaller concentric ball.

Let $w\in W$ with $\pi(w)\in C$ and $y:=\exp(w)$. Since $C\times\{y\}\subset T\subset \Omega$, we have $C\subset I^-(y)$ and, $\Omega$ being open, the function $d(\cdot,y)$ is smooth on the open neighborhood $\{q: (q,y)\in \Omega\}$ of $C$ by Lemma \ref{lemomega}. Finally, by the same lemma, the functions $(q,y)\mapsto \p_{i}d(q,y)$ and $(q,y)\mapsto \p_{i}\p_{j}d(q,y)$, where $\p_i$ denotes the coordinate derivative with respect to the $i$-th coordinate of $q$, are continuous on the open set $\Omega\cap (U\times M)$. This set contains the compact set $T$, so these functions are bounded on $T$ by some constant $G$. As $(q,\exp(w))\in T$ for every $q\in C$ and every $w\in W$ with $\pi(w)\in C$, the bound $G$ holds for the first and second coordinate derivatives of $d(\cdot,\exp(w))$ on $C$, uniformly in $w$.
\end{proof}

\begin{theorem} \label{xxhg}
Assume that $(M,g)$ is future timelike complete. The following properties hold.
\begin{itemize}
\item[(i)] For $p\in \textrm{dom}\, b$, every future inextendible $b$-calibrated causal curve starting from $p$ is a timelike ray of infinite Lorentzian length, hence, once parametrized with proper time, a future complete timelike geodesic.
\item[(ii)] The set $\textrm{dom}\, b$ is open.
\item[(iii)] Let $p\in \textrm{dom}\, b$ and let $\alpha(t)$, $\alpha(0)=p$, be a proper-time parametrized $b$-calibrated curve starting from $p$. For each $s> 0$
\begin{equation} \label{cmby}
b_{p,s}(q)=b(p)+s-d(q,\alpha(s))
\end{equation}
is a smooth upper support function for $b$ at $p$, namely it is smooth in a neighborhood of $p$, $b_{p,s}(p)=b(p)$ and $b\le b_{p,s}$ on $I^-(\alpha(s))$.
\item[(iv)] $b: M\to (-\infty,+\infty]$ is continuous, $H$ is the graph of $b\vert_{\textrm{dom}\, b}$, and every generator of $H$ projects to a future inextendible $b$-calibrated curve.
\item[(v)] For every compact subset $K\subset \textrm{dom}\, b$, the set $U_K$ of unit future-directed timelike vectors tangent to some $b$-calibrated curve at its starting point in $K$ is compact in $TM$.
\item[(vi)] $b\vert_{\textrm{dom}\, b}$ is locally semi-concave (locally upper-$C^2$) and hence locally Lipschitz.
\end{itemize}
\end{theorem}

The regularity of $b$ is inherited from that of horizons which are known to be locally Lipschitz and semi-convex when expressed as graphs over a spacelike hypersurface in the direction of a timelike vector field \cite{chrusciel01,minguzzi14d}. Here the  horizon is expressed as a graph where the ``vertical'' variable $z$ is spacelike, $\tilde g(\p_z,\p_z)>0$, which explains the upper (rather than lower) support functions. One could obtain the  regularity of $b$ from that of $H$ through a change of coordinates on $\tilde M$, taking advantage of the fact that $\p_z$ stays away from the light cone; however, the direct proof presented below is shorter.

\begin{proof}
(i). Let $x$ be a $b$-calibrated future inextendible causal curve starting from $p$. By Proposition \ref{jjtr} it is maximizing, hence a timelike or a lightlike geodesic. Suppose it is lightlike, then $b(x(t))=c:=b(p)$ for every $t\ge 0$. Let $\Sigma':=\Sigma_{\vert c\vert+1}$. For those $t$ such that $x(t)\in J^+(\Sigma)$ we have, by (\ref{cmlr}), $d^+_\Sigma(x(t))\le c<\vert c\vert +1$, thus $x(t)\in J^-(\Sigma')$; the same holds if $x(t)\in J^-(\Sigma)$. Thus $x$ is imprisoned in the compact set $J^+(p)\cap J^-(\Sigma')$, which contradicts non-total imprisonment. Thus $x$ is timelike. Finally, $x$ meets every Cauchy hypersurface $\Sigma_k$, $k\in \mathbb{N}$, where by (\ref{cmlr}) $b\ge k$. As $b$ increases by the Lorentzian length along $x$, the Lorentzian length of $x$ is infinite.

(ii). Let $p\in \textrm{dom}\, b$, let $\gamma$ be a $b$-calibrated curve starting from $p$ (Prop.\ \ref{bber}) and take $r=\gamma(\epsilon)$, $\epsilon>0$. Then $p\in I^-(r)$ by (i), and for $q\in I^-(r)$ we have by (\ref{fcpb}) $b(q)\le b(r)-d(q,r)\le b(r)<+\infty$, thus $I^-(r)\subset \textrm{dom}\, b$, which proves that $\textrm{dom}\, b$ is open. This argument also shows that $b$ is bounded from above in a neighborhood of each point of $\textrm{dom}\, b$.

(iii). Since $\alpha$ is maximizing and $b$-calibrated, $b_{p,s}(p)=b(p)+s-s=b(p)$ and, for $q\le \alpha(s)$, $b(q)+ d(q,\alpha(s))\le b(\alpha(s))=b(p)+s$ which means $b(q)\le b_{p,s}(q)$. Since $\alpha$ is a timelike ray, by Lemma \ref{lemeik} the function $d(\cdot,\alpha(s))$, hence $b_{p,s}$, is smooth in a neighborhood of $p$.

(iv). By (iii), $\limsup_{q\to p} b(q)\le \lim_{q\to p} b_{p,s}(q)=b(p)$, thus $b$ is upper semi-continuous and hence continuous on $\textrm{dom}\, b$, as we proved that $b$ is lower semi-continuous. At $p\in M\backslash \textrm{dom}\, b$, $b(p)=+\infty$ and lower semi-continuity implies $b(q)\to +\infty$ for $q\to p$. Hence $b$ is continuous.

Suppose $(p,r)$ belongs to $H$ but not to the graph of $b$. As $H\subset \textrm{epi}\, b$, $r\ge b(p)$, thus $b(p)<+\infty$ and $r>b(p)$. By continuity of $b$ there are $\epsilon>0$ and an open set $O\ni p$ such that $b(q)< r-\epsilon$ for $q\in O$, which means that the open set $O\times (r-\epsilon,r+\epsilon)$ does not intersect $\tilde D$ and contains $(p,r)$, a contradiction  since $H$ is the boundary of $\tilde D$. We conclude that $H\subset \textrm{graph}\, b$, the other inclusion having been already observed. In particular every generator starts from a point $(p,b(p))$ of the graph of $b$, with $p\in \textrm{dom}\, b$, and, by Proposition \ref{bber}, projects to a $b$-calibrated curve.

(v). Let us introduce a complete Riemannian metric $h$ on $M$. The elements of $U_K$ must be  bounded in $h$-norm, otherwise we could find a  sequence $v_k\in U_K$, $v_k\in T_{p_k}M$, $p_k\to p\in K$, $h(v_k,v_k)\to \infty$. Let $v_k':=v_k/\sqrt{h(v_k,v_k)}$, so that $h(v_k',v_k')=1$ and $g(v_k',v_k')=-1/h(v_k,v_k)\to 0$. Passing to a subsequence, $v_k'\to v'\in T_pM$, where $v'$ is a future directed lightlike vector. Let $x_k(t)=\exp(tv_k')$, $t\ge 0$, which is an affine reparametrization of the $b$-calibrated ray with initial vector $v_k$, and let $x:[0,a)\to M$, $x(t)=\exp(tv')$, be the future inextendible lightlike geodesic with initial vector $v'$. By the continuous dependence of geodesics on the initial data, $x_k\to x$ uniformly on compact subsets of $[0,a)$. By the continuity of $b$ and $d$ the  identity $b(x_k(t_2))=b(x_k(t_1))+d(x_k(t_1),x_k(t_2))$, $0\le t_1<t_2<a$, passes to the limit, showing that $x$ is a future inextendible lightlike $b$-calibrated causal curve starting from $p\in \textrm{dom}\, b$, a contradiction with (i). Closedness of $U_K$ follows from the same type of limit argument, using (i) to exclude lightlike limits.

(vi). Let $p\in \textrm{dom}\, b$ and let $K\subset \textrm{dom}\, b$ be a compact neighborhood of $p$ contained in a coordinate chart. For $v\in U_K$ let $\alpha_v$ be the proper time parametrized $b$-calibrated ray with initial vector $v$, so that $\alpha_v(1)=\exp(v)$; by (i) and Lemma \ref{lemeik}, $(\pi(v),\alpha_v(1))\in \Omega$. As $U_K$ is compact by (v), Lemma \ref{lemunif} applied with $W=U_K$ provides a closed coordinate ball $C\subset K$ centered at $p$ such that, for $v\in U_C\subset U_K$, $C\subset I^-(\alpha_v(1))$ and the functions $d(\cdot,\alpha_v(1))$ are smooth on a neighborhood of $C$ with first and second coordinate derivatives bounded on $C$ uniformly in $v$.
For $q_0\in C$ and $v\in U_C\cap T_{q_0}M$ the function
\[
b_v(q):=b(q_0)+1-d(q,\alpha_v(1))
\]
is smooth on $C$, it is an upper support function for $b$ at $q_0$ (by (iii), since $C\subset I^-(\alpha_v(1))$), and its first and second derivatives in coordinates are bounded on $C$ uniformly in $v$. This is the local upper-$C^2$ property,  which implies that in a neighborhood of any point $b$ is semi-concave \cite[Thm.\ 10.33]{rockafellar09}, see \cite[Prop.\ 3]{minguzzi14d} and references therein for the various equivalences. Local Lipschitzness follows from the uniform bound $G$ on the gradients of the support functions: for $q,q'\in C$ choose $v\in U_C\cap T_qM$, which exists by Proposition \ref{bber}, so that $b_v$ is the upper support function at $q_0=q$, namely $b_v(q)=b(q)$ and $b\le b_v$ on $C$; then $b(q')-b(q)\le b_v(q')-b_v(q)\le G\vert q'-q\vert$, and exchanging the roles of $q$ and $q'$ we get $\vert b(q')-b(q)\vert \le G\vert q'-q\vert$ (in coordinates, for $C$ convex).

\end{proof}

\begin{proposition} \label{coruni}
Assume that $(M,g)$ is future timelike complete. $\textrm{dom}\, b$ is the union of the sets $I^-(\gamma)$ where $\gamma$ runs over the future inextendible $b$-calibrated curves. Each of these sets is a TIP.
\end{proposition}

\begin{proof}
If $\gamma$ is $b$-calibrated then $\gamma\subset \textrm{dom}\, b$ and $I^-(\gamma)\subset \textrm{dom}\, b$ as $\textrm{dom}\, b$ is a past set. Conversely, let $p\in \textrm{dom}\, b$. The key point is that from $p$ starts a future inextendible $b$-calibrated curve $\gamma$ (Proposition \ref{bber}). It is timelike by Theorem \ref{xxhg}(i), so $p\in I^-(\gamma)$. The sets $I^-(\gamma)$ with $\gamma$ future inextendible are TIPs \cite{geroch72}.
\end{proof}

The next result shows that the failure of the local equi-Lipschitz property for the Lorentzian distances to the level sets of Cauchy temporal functions, established in \cite{galloway26}, cannot happen on $\textrm{dom}\, b$ for the foliation by the level sets $\Sigma_\tau$. We say that a family of functions $\{f_\tau\}_{\tau>0}$ is {\em locally eventually equi-Lipschitz} on an open set $O$ if every point of $O$ has a neighborhood $V$ and a number $\tau_V>0$ such that the functions $f_\tau$, $\tau\ge \tau_V$, are Lipschitz on $V$ with a common constant (with respect to a coordinate metric). This form of the property suffices for the arguments below.

\begin{theorem} \label{thmequi}
Assume that $(M,g)$ is future timelike complete. Let $K\subset \textrm{dom}\, b$ be compact. There is $\tau_0>0$ such that $K\subset I^-(\Sigma_{\tau})$ for every $\tau\ge \tau_0$, and for any such $\tau_0$ the set $V_K$ of unit future-directed timelike vectors which are initial tangents of maximizing timelike geodesics from a point $q\in K$ to some $\Sigma_\tau$, $\tau\ge \tau_0$, (of Lorentzian length $d(q,\Sigma_\tau)$) is relatively compact in $TM$. As a consequence, the family $\{d(\cdot,\Sigma_\tau)\}_{\tau>0}$ is locally eventually equi-Lipschitz on $\textrm{dom}\, b$.
\end{theorem}

\begin{proof}
The function $d_\Sigma$ is continuous (Theorem \ref{thmcont}), so it is bounded from above on $K$ by some $\tau_0>0$. For $\tau> \tau_0$, by Theorem \ref{kvvb} $d_\Sigma\ge \tau$ on $J^+(\Sigma_\tau)$, thus $K$ does not intersect $J^+(\Sigma_\tau)$ and, since $\Sigma_\tau$ is a Cauchy hypersurface, $K\subset I^-(\Sigma_\tau)$. This proves the existence of $\tau_0$ (after replacing $\tau_0$ by $\tau_0+1$).

Suppose that $V_K$ is not relatively compact, then there are $q_n\in K$, $\tau_n\ge \tau_0$, and maximizing timelike geodesics $\sigma_n$ from $q_n$ to $\Sigma_{\tau_n}$ with unit initial tangents $v_n$ such that $h(v_n,v_n)\to \infty$. Passing to a subsequence, $q_n\to q\in K$, the $h$-normalized initial vectors converge to a lightlike vector $u$ at $q$ and, parametrizing the $\sigma_n$ with respect to $h$-arc length, they converge uniformly on compact subsets to a curve $\sigma$ starting from $q$ which is a suitable reparametrization of the lightlike geodesic with initial vector $u$.

If $\tau_n$ has a bounded subsequence, we can assume $\tau_n\to \tau_*\ge \tau_0$. For large $n$ the curves $\sigma_n:[0,L_n]\to M$ are contained in the compact set $J^+(K)\cap J^-(\Sigma_{\tau_*+1})$, so their $h$-lengths $L_n$ are bounded, as otherwise $\sigma$ would be a future inextendible causal curve imprisoned in this compact set, in contradiction with non-total imprisonment. Passing to a subsequence, $L_n\to L$, so that the endpoints $e_n=\sigma_n(L_n)\in \Sigma_{\tau_n}$ converge to $e:=\sigma(L)$, which belongs to $\Sigma_{\tau_*}$ by the continuity of $d_\Sigma$, and the curves $\sigma_n$ converge uniformly to $\sigma\vert_{[0,L]}$, still denoted $\sigma$. Since the $\sigma_n$ are maximizing, by \cite[Thm.\ 2.58]{minguzzi18b} $\sigma$ is maximizing and $\ell(\sigma)=\lim_n \ell(\sigma_n)=\lim_n d(q_n,\Sigma_{\tau_n})$, and the last limit equals $d(q,\Sigma_{\tau_*})$ by Corollary \ref{corlipt} and the continuity of $d(\cdot,\Sigma_{\tau_*})$. It is positive as $q\in K\subset I^-(\Sigma_{\tau_*})$. Thus $\sigma$ is a maximizing causal curve of positive Lorentzian length, hence, up to reparametrization, a timelike geodesic \cite[Thm.\ 2.20]{minguzzi18b}. This contradicts the fact that $\sigma$ is a reparametrization of a lightlike geodesic.

Thus $\tau_n\to +\infty$, the $h$-lengths of the $\sigma_n$ diverge and $\sigma$ is a future inextendible causal curve starting from $q$, which is a reparametrization of a lightlike geodesic. Fix $t>0$ and an integer $m$. For $n$ large, $\tau_n\ge m$ and, since $\sigma_n$ is maximizing from $q_n$ to $\Sigma_{\tau_n}$, $\ell(\sigma_n\vert_{[0,t]})=d(q_n,\Sigma_{\tau_n})-d(\sigma_n(t),\Sigma_{\tau_n})$, namely
\[
b_m(\sigma_n(t))\le b_{\tau_n}(\sigma_n(t))=b_{\tau_n}(q_n)+\ell(\sigma_n\vert_{[0,t]})\le b(q_n)+\ell(\sigma_n\vert_{[0,t]}).
\]
Letting $n\to \infty$, using the continuity of $b_m$, the continuity of $b$ at $q\in \textrm{dom}\, b$ (Theorem \ref{xxhg}) and the convergence $\ell(\sigma_n\vert_{[0,t]})\to \ell(\sigma\vert_{[0,t]})$, which holds as the $\sigma_n$ are maximizing \cite[Thm.\ 2.58]{minguzzi18b}, we get $b_m(\sigma(t))\le b(q)+\ell(\sigma\vert_{[0,t]})$. Letting $m\to \infty$, $b(\sigma(t))\le b(q)+\ell(\sigma\vert_{[0,t]})$, while the opposite inequality follows from the rushing property. By Proposition \ref{jjtr} $\sigma$ is a $b$-calibrated lightlike curve starting from $q\in \textrm{dom}\, b$, which contradicts Theorem \ref{xxhg}(i).

It remains to prove the equi-Lipschitz property, for which we use lower support functions. Let $p\in \textrm{dom}\, b$, let $K\subset \textrm{dom}\, b$ be a compact neighborhood of $p$ contained in a coordinate chart, and let $\tau_0$ be as above. Since $\Sigma_{\tau_0}$ is a Cauchy hypersurface and $K\subset I^-(\Sigma_{\tau_0})$, the continuous function $d(\cdot,\Sigma_{\tau_0})$ is positive on $K$, so that $\delta:=\min_K d(\cdot,\Sigma_{\tau_0})>0$. For $q\in K$ and $\tau\ge \tau_0$ we have $d(q,\Sigma_\tau)\ge d(q,\Sigma_{\tau_0})\ge \delta$ by Corollary \ref{corlipt}. Let $s_0:=\delta/2$ and let $\overline{V_K}$ be the closure of $V_K$, a compact subset of $TM$ consisting of unit timelike vectors. Every $w\in \overline{V_K}$ is the initial vector of a unit speed timelike geodesic which is maximizing on $[0,\delta]$: indeed it is a limit of initial vectors of unit speed maximizing geodesics of Lorentzian length at least $\delta$, which converge to a geodesic with initial vector $w$, maximizing on $[0,\delta]$ by \cite[Cor.\ 2.59]{minguzzi18b}. By Lemma \ref{lemeik}, applied with $s=s_0<\delta=s'$, $(\pi(w),\exp(s_0w))\in \Omega$ for every $w\in \overline{V_K}$. Thus Lemma \ref{lemunif} applies to the compact set $W:=\{s_0w: w\in \overline{V_K}\}$ and gives a closed coordinate ball $C\subset K$ centered at $p$ such that, for $w\in \overline{V_K}$ with $\pi(w)\in C$, we have $C\subset I^-(\exp(s_0w))$ and the coordinate gradients of the functions $d(\cdot,\exp(s_0w))$ on $C$ are bounded by a constant $G$ independent of $w$.

Let $q,q'\in C$ and $\tau\ge \tau_0$, and let $\gamma$ be a unit speed maximizing timelike geodesic from $q$ to $\Sigma_\tau$ of Lorentzian length $d(q,\Sigma_\tau)$, which exists by the time dual of Theorem \ref{thmcont}. Its initial vector $w$ belongs to $V_K$; let $y:=\gamma(s_0)=\exp(s_0w)$. Then $d(q,\Sigma_\tau)=d(q,y)+d(y,\Sigma_\tau)$, while $q'\ll y$ and $d(q',\Sigma_\tau)\ge d(q',y)+d(y,\Sigma_\tau)$. Therefore
\[
d(q',\Sigma_\tau)-d(q,\Sigma_\tau)\ge d(q',y)-d(q,y)\ge -G\vert q'-q\vert ,
\]
and exchanging the roles of $q$ and $q'$ we get $\vert d(q',\Sigma_\tau)-d(q,\Sigma_\tau)\vert \le G\vert q'-q\vert$ for every $\tau\ge \tau_0$. This is the local eventual equi-Lipschitz property.
\end{proof}

\begin{proposition} \label{propbdry}
Assume that $(M,g)$ is future timelike complete. The achronal boundary $\p\, \textrm{dom}\, b$ is a $C^0$ future null hypersurface.
\end{proposition}

\begin{proof}
Since $\textrm{dom}\, b$ is an open past set, its boundary is an achronal boundary.
Let $p\in \p\, \textrm{dom}\, b$. We can find a sequence $p_k\in \textrm{dom}\, b$, $p_k\to p$, and corresponding timelike $b$-calibrated rays $\sigma_k$ starting from $p_k$, which converge uniformly on compact subsets to a causal future ray $\sigma$ starting from $p$ \cite[Thm.\ 2.51, Rem.\ 2.52, Cor.\ 2.59]{minguzzi18b}. By construction $\sigma\subset \overline{\textrm{dom}\, b}$. If there were a point $q\in \sigma\cap \textrm{dom}\, b$ then, as $b$ is rushing, $b(p)$ would be finite, a contradiction with $p\in \p\, \textrm{dom}\, b$ and the openness of $\textrm{dom}\, b$. Thus $\sigma$ is entirely contained in $\p\, \textrm{dom}\, b$ and since this set is achronal, $\sigma$ is a lightlike geodesic.
\end{proof}

\begin{lemma} \label{lemgrad}
Assume that $(M,g)$ is future timelike complete. Let $p\in \textrm{dom}\, b$ be a point where $b$ is differentiable. Then $T:=-\nabla b(p)$ is a future directed unit timelike vector and there is exactly one $b$-calibrated future inextendible curve starting from $p$, namely the timelike geodesic with initial velocity $T$.
\end{lemma}

\begin{proof}
Let $\alpha$ be a proper time parametrized $b$-calibrated ray starting from $p$ and $s>0$. The smooth function $b_{p,s}$ of Theorem \ref{xxhg}(iii) satisfies $b\le b_{p,s}$ near $p$ with equality at $p$, thus $\dd b_p=\dd (b_{p,s})_p$. Since $\alpha$ is a timelike ray, the function $d(\cdot,\alpha(s))$ satisfies the eikonal equation near $p$ (Lemma \ref{lemeik}), so $g(\nabla b,\nabla b)=-1$ at $p$. By the rushing property, $b(\exp_p(tv))\ge b(p)+t$ for every future directed unit timelike $v$ and small $t>0$, so $\dd b(v)\ge 1>0$, which shows that $T$ is future directed. Finally, $\dd b(\dot \alpha(0))=1$ by calibration, namely $-g(T,\dot \alpha(0))=1$, and by the reverse Cauchy-Schwarz inequality for future directed unit timelike vectors $\dot \alpha(0)=T$. As $\alpha$ was arbitrary, the calibrated curve from $p$ is unique.
\end{proof}

The following result shows that a counterexample to Bartnik's conjecture would necessarily have Busemann functions with non-differentiability points.

\begin{theorem} \label{thmcone}
Assume that $(M,g)$ is future timelike complete. If $b$ is $C^1$ on $\textrm{dom}\, b$ then through every point of $\textrm{dom}\, b$ passes a timelike line which is $b$-calibrated. The line is future complete. If $(M,g)$ is also past timelike complete, the line is complete and, under TCC, the spacetime splits. The analogous statement holds for $b^-$, with the roles of future and past timelike completeness exchanged.
\end{theorem}

\begin{proof}
Let $T=-\nabla b$, a continuous future directed unit timelike vector field on $\textrm{dom}\, b$ by Lemma \ref{lemgrad}. Let $p\in \textrm{dom}\, b$ and let $\alpha:(a,+\infty)\to M$, $a\in [-\infty,0)$, be the inextendible timelike geodesic with $\alpha(0)=p$, $\dot \alpha(0)=T(p)$, parametrized with proper time (it is future complete by Theorem \ref{xxhg}(i), while $a$ might be finite). By Lemma \ref{lemgrad} $\alpha\vert_{[0,+\infty)}$ is the $b$-calibrated ray from $p$, and since for $t\ge 0$ the restriction $\alpha\vert_{[t,+\infty)}$ is calibrated, $\dot \alpha(t)=T(\alpha(t))$. As $\textrm{dom}\, b$ is a past set, $\alpha(t)\in \textrm{dom}\, b$ for every $t\in (a,+\infty)$.

Let $S:=\{t\in (a,0]: \alpha\vert_{[t,+\infty)} \textrm{ is } b\textrm{-calibrated}\}$. $S$ is non-empty and closed in $(a,0]$ by the continuity of $b$. Let us show that it is open in $(a,0]$. Let $t_0\in S$ and $x=\alpha(t_0)$. By Peano's theorem there is an integral curve $\beta:(t_0-\epsilon,t_0]\to \textrm{dom}\, b$ of $T$ with $\beta(t_0)=x$, where $t_0-\epsilon\ge a$. Along it $g(\dot \beta,\dot \beta)=-1$ and $\frac{\dd}{\dd t} b(\beta(t))=\dd b(T)=-g(\nabla b,\nabla b)=1$. Thus the concatenation of $\beta$ with $\alpha\vert_{[t_0,+\infty)}$ is a causal curve along which $b$ increases by the Lorentzian length, and by Proposition \ref{jjtr} it is maximizing, hence a geodesic without corners. Therefore $\beta=\alpha\vert_{(t_0-\epsilon,t_0]}$ and $(t_0-\epsilon,t_0]\subset S$. We conclude $S=(a,0]$, so $\alpha$ is $b$-calibrated on its whole domain $(a,+\infty)$ and, again by Proposition \ref{jjtr}, it is a timelike line, namely an inextendible maximizing timelike geodesic. It is future complete by Theorem \ref{xxhg}(i), and it is complete if $(M,g)$ is past timelike complete, in which case the Lorentzian splitting theorem applies under TCC.
\end{proof}

The following class of examples shows that, under future timelike completeness alone, $\textrm{dom}\, b$ can be the whole spacetime, with $b$ smooth and calibrated curves which are future complete timelike lines, without the spacetime being a product.

\begin{example}[Warped products with compact fiber] \label{exwarp}
Let $M=I\times S$, $I=(\alpha,\omega)$, $-\infty\le \alpha<\omega\le +\infty$, $g=-\dd t^2+a(t)^2h$, where $(S,h)$ is a compact Riemannian manifold of dimension $n$ and $a:I\to (0,+\infty)$ is smooth. A future directed causal curve $\gamma(s)=(t(s),x(s))$ satisfies $\dot t>0$ and $a(t)\vert \dot x\vert_h\le \dot t$, so it can be parametrized with $t$, and then $\vert \dd x/\dd t\vert_h\le 1/a(t)$. If $\gamma$ is future inextendible, then $t\to \omega$ along it, for otherwise $x$ would have finite $h$-length, and hence a limit by the compactness of $S$, and $\gamma$ would be extendible; dually in the past. Thus the slices $S_\tau:=\{\tau\}\times S$, $\tau\in I$, are compact Cauchy hypersurfaces, and $t$ is a Cauchy time function.

{\em Future timelike completeness.} Let $\tau\mapsto (t(\tau),x(\tau))$ be a timelike geodesic parametrized with proper time. It is a critical point of the energy functional $\int (-\dot t^2+a(t)^2 h(\dot x,\dot x))\,\dd\tau$ with respect to variations with fixed endpoints. Consider a variation $x_s$ of $x$ alone, $t$ being kept fixed, with variational field $V=\p_s x_s\vert_{s=0}$ vanishing at the endpoints, and let $D_\tau$, $D_s$ be the covariant derivatives of $\nabla^h$ along the curves of the variation. Since $D_s\dot x=D_\tau V$, we get
\[
0=\frac{1}{2}\frac{\dd}{\dd s}\Big\vert_{s=0}\int a^2 h(\dot x_s,\dot x_s)\,\dd\tau=\int a^2 h(\dot x,D_\tau V)\,\dd\tau=-\int h(D_\tau(a^2\dot x),V)\,\dd\tau ,
\]
thus $D_\tau(a^2\dot x)=0$: the momentum $a^2\dot x$ is $\nabla^h$-parallel along $x$. In particular $x$ is a reparametrization of an $h$-geodesic and $c:=\vert a^2\dot x\vert_h$ is constant. Since $-\dot t^2+a^2\vert\dot x\vert_h^2=-1$, we get $\dd \tau/\dd t=a/(a^2+c^2)^{1/2}$, and since $t\to\omega$ along future inextendible geodesics, the geodesic is future complete if and only if $\int^{\omega} a(a^2+c^2)^{-1/2}\dd t=+\infty$. For $c>0$ the integrand is comparable with $\min\{a,1\}$, up to positive factors depending only on $c$, while for $c=0$ (the comoving observers $t\mapsto (t,x)$) it equals $1$. We conclude that $(M,g)$ is future timelike complete if and only if
\[
\omega=+\infty \quad \textrm{and}\quad \int^{+\infty} \min\{a(t),1\}\,\dd t=+\infty .
\]
In particular, a scale factor which decreases at least exponentially for large $t$ is incompatible with future timelike completeness.

{\em Energy conditions.} The Ricci tensor of the warped product is \cite[Cor.\ 7.43]{oneill83}
\begin{align*}
&Ric(\p_t,\p_t)=-n\frac{\ddot a}{a},\qquad Ric(\p_t,V)=0,\\
&Ric(V,W)=Ric_h(V,W)+\Big(\frac{\ddot a}{a}+(n-1)\frac{\dot a^2}{a^2}\Big)g(V,W),
\end{align*}
for $V,W$ tangent to $S$. Under TCC, $\ddot a\le 0$. A positive concave function on $(\alpha,+\infty)$ is non-decreasing, since $\dot a(t_1)<0$ would give $a(t)\le a(t_1)+\dot a(t_1)(t-t_1)<0$ for large $t$. Thus, under TCC, future timelike completeness implies $\dot a\ge 0$; moreover $\dot a$ is non-increasing, so either $\dot a>0$ on $I$ or $a$ is constant on some interval $[t_1,+\infty)$. A similar conclusion follows, for $n\ge 2$, from the weaker null convergence condition, which is implied both by TCC and by the dominant energy condition, provided $h$ does not have positive Ricci curvature: for a unit vector $e\in TS$ the vector $N=\p_t+e/a$ is null and, with $\nu:=\dot a/a$,
\[
Ric(N,N)=a^{-2}Ric_h(e,e)-(n-1)\dot \nu .
\]
As $\dot \nu$ depends only on $t$, for $n\ge 2$ the null convergence condition gives $\dot \nu\le m/((n-1)a^2)$, where $m$ is the minimum of $Ric_h(e,e)$ over the unit vectors $e\in TS$. If $m\le 0$, namely if $Ric_h(e,e)\le 0$ for some unit vector $e$ at some point of $S$, then $\dot \nu\le 0$; if it were $\nu(t_1)<0$ for some $t_1$, we would have $a(t)\le a(t_1)e^{\nu(t_1)(t-t_1)}$ for $t\ge t_1$, in contradiction with future timelike completeness. Thus, for $n\ge 2$, under the null convergence condition and $m\le 0$, future timelike completeness implies $\dot a\ge 0$, while acceleration, $\ddot a>0$, is allowed. For $n=1$ the null convergence condition holds trivially, as $Ric$ is proportional to $g$, and the conclusion fails: $a=(1+t)^{-1}$, $t>0$, over the circle gives a future timelike complete spacetime with strictly decreasing scale factor. For $n\ge 2$ and $m>0$ the conclusion fails as well: the de Sitter spacetime, $a=\cosh t$ over the round sphere, satisfies the dominant energy condition and its scale factor contracts and then expands, and a direct computation shows that $a=1+e^{-t}$, $t>0$, over the round $3$-sphere satisfies the dominant energy condition, gives a future timelike complete spacetime, and is strictly decreasing. These classes contain the non-contracting Friedmann--Lema\^itre--Robertson--Walker models with compact spatial sections, including those with an initial singularity.

{\em The Cauchy Busemann function.} No energy condition is needed here. Assume that $(M,g)$ is future timelike complete and let $\Sigma=S_{t_0}$, $t_0\in I$. For a future directed causal curve $\gamma$ from $p$ to $q$,
\[
\ell(\gamma)=\int \sqrt{\dot t^2-a^2\vert \dot x\vert_h^2}\, \dd s\le \int \dot t\, \dd s=t(q)-t(p),
\]
with equality if and only if $\dot x=0$ almost everywhere, namely if and only if $\gamma$ is a reparametrization of an integral curve of $\p_t$. Thus $d(p,q)\le t(q)-t(p)$, with equality when $q$ lies on the integral curve of $\p_t$ through $p$. As a consequence $d^+_\Sigma=t-t_0$ on $J^+(\Sigma)$, $\Sigma_\tau=S_{t_0+\tau}$, and $d(p,\Sigma_k)=t_0+k-t(p)$ whenever $t(p)\le t_0+k$. Therefore $b_k=t-t_0$ on $\{t\le t_0+k\}$, and
\[
b=t-t_0 \quad \textrm{on } M .
\]
In particular $\textrm{dom}\, b=M$ and $b$ is smooth. Moreover, by the equality case above, a causal curve is $b$-calibrated if and only if it is a reparametrization of an integral curve of $\p_t$: the calibrated curves are the comoving observers, and from every point starts exactly one of them, in agreement with Proposition \ref{propcrease}. By Theorem \ref{thmcone} they are future complete timelike lines, as can also be checked directly. They are complete if and only if $\alpha=-\infty$, in which case, under TCC, $a$ is a positive concave function on $\mathbb{R}$, hence constant, and the spacetime is the static product. If $\alpha>-\infty$ the spacetime is past timelike incomplete and it is not a product, although $\textrm{dom}\, b=M$.
\end{example}

The differentiability of $b$ has a simple characterization in terms of calibrated curves.

\begin{proposition} \label{propcrease}
Assume that $(M,g)$ is future timelike complete. Let $p\in \textrm{dom}\, b$. Then $b$ is differentiable at $p$ if and only if exactly one future inextendible $b$-calibrated curve starts from $p$ (in particular, by Proposition \ref{bber}, $b$ is differentiable at every point $x(t)$, $t>0$, of every future inextendible $b$-calibrated curve $x$). Consequently, $b$ is $C^1$ on an open set $O\subset \textrm{dom}\, b$ if and only if from every point of $O$ starts exactly one future inextendible $b$-calibrated curve. In this case the level sets of $b$ in $O$ are $C^1$ spacelike hypersurfaces orthogonal to the $b$-calibrated curves.
\end{proposition}

The proof shows that if exactly one future inextendible $b$-calibrated curve starts from $p$, then $b$ is strictly differentiable at $p$.

\begin{proof}
The direct implication is Lemma \ref{lemgrad}. For the converse, let $\alpha$ be the unique $b$-calibrated curve starting from $p$, parametrized with proper time, and let $v=\dot \alpha(0)$. We work in a coordinate chart around $p$, where $b$ is locally Lipschitz (Theorem \ref{xxhg}(vi)). Let $p_j\to p$ be points of differentiability of $b$ such that the differentials $\dd b_{p_j}$ converge to some covector $\omega$. By Lemma \ref{lemgrad}, $-\nabla b(p_j)=v_j$, where $v_j$ is the initial tangent of the unique $b$-calibrated curve starting from $p_j$. By Theorem \ref{xxhg}(v), applied to a compact neighborhood $K$ of $p$, a subsequence of $v_j$ converges to the initial tangent of a $b$-calibrated curve starting from $p$, which is necessarily $v$. Thus $\omega=-g(v,\cdot)$. We have shown that all the limits of differentials of $b$ along sequences of differentiability points converging to $p$ coincide. For a locally Lipschitz function the Clarke generalized gradient at $p$ is the convex hull of these limits, and the function is strictly differentiable at $p$ if the generalized gradient is a singleton \cite[Ch.\ 9]{rockafellar09}. Thus $b$ is differentiable at $p$, with $\dd b_p=-g(v,\cdot)$. For the last statement, a function which is strictly differentiable at every point of an open set is $C^1$ there \cite[Ch.\ 9]{rockafellar09}; the converse follows from Lemma \ref{lemgrad}. In this case, for $q\in O$, $\nabla b(q)$ is timelike by Lemma \ref{lemgrad}, so the level set of $b$ through $q$ is, near $q$, a $C^1$ spacelike hypersurface with tangent space $\ker \dd b_q=\nabla b(q)^\perp$. A $b$-calibrated curve through $q$ restricts, from $q$ on, to the unique future inextendible $b$-calibrated curve starting from $q$, whose unit tangent at $q$ is $-\nabla b(q)$ by Lemma \ref{lemgrad}, so it is orthogonal to the level set.
\end{proof}

\begin{remark}
By Theorem \ref{thmcone} and Proposition \ref{propcrease}, if $(M,g)$ is timelike complete, satisfies TCC and does not split, then there is a point of $\textrm{dom}\, b^+$ from which at least two distinct future $b^+$-calibrated rays start, and dually for $b^-$.
\end{remark}

The fine differentiability theory of horizons \cite{minguzzi14d} (see also \cite{beem98,chrusciel98b,chrusciel01}) gives an alternative proof of Lemma \ref{lemgrad} and Proposition \ref{propcrease}, and stronger results. By Theorem \ref{xxhg}(iv) and Proposition \ref{bber}, the horizon $H\subset \tilde M$ is the graph of $b$ over $\textrm{dom}\, b$, and its generators are the lightlike lifts $t\mapsto (x(t),b(x(t)))$ of the future inextendible $b$-calibrated curves $x$. For $p\in \textrm{dom}\, b$ let $U_p\subset T_pM$ be the set of unit initial tangents of the $b$-calibrated rays starting from $p$. The semitangents of $H$ at $P=(p,b(p))$, namely the future directed lightlike vectors tangent to the generators starting from $P$, are then the vectors proportional to $(v,1)$, $v\in U_p$. In particular, $\tilde g(\p_z,(v,1))=1\ne 0$, so the spacelike vector field $\p_z$ is transverse to every null hyperplane containing a semitangent. If $H$ is differentiable at $P$, its tangent space $\ker(\dd z-\dd b)$ is a null hyperplane, which gives $g(\nabla b,\nabla b)=-1$, and the unique semitangent is proportional to its normal $\p_z-\nabla b$, which gives Lemma \ref{lemgrad}.

\begin{theorem} \label{thmfine}
Assume that $(M,g)$ is future timelike complete. Let $n+1=\dim M$ and let $S\subset \textrm{dom}\, b$ be the set of non-differentiability points of $b$.

\noindent {\em Upper bounds on the size of $S$.}
\begin{itemize}
\item[(a)] $b$ is differentiable at $p\in \textrm{dom}\, b$ if and only if $U_p$ is a singleton, and the map $p\mapsto \nabla b(p)$ is continuous on $\textrm{dom}\, b\backslash S$ with the induced topology.
\item[(b)] $b$ is $C^{1,1}_{loc}$ on the interior of $\textrm{dom}\, b\backslash S$; in particular, under the assumption of Theorem \ref{thmcone}, $b$ is $C^{1,1}_{loc}$ on $\textrm{dom}\, b$.
\item[(c)] For $0\le k\le n+1$, the set $S^k$ of points $p\in \textrm{dom}\, b$ for which the affine span of $U_p$ has dimension $k$ is countably $\mathcal{H}^{n+1-k}$-rectifiable; in particular $S^{n+1}$ is countable. Thus $S=\cup_{k\ge 1} S^k$ is countably $\mathcal{H}^{n}$-rectifiable and its Hausdorff dimension is at most $n$. {For $p\in S^k$ let $V_p\subset T_pM$ be the $k$-dimensional direction space of the affine span of $U_p$. If $x_j\to p$, $x_j\ne p$, and there are $u^i_j\in U_{x_j}$, $i=0,\dots,k$, converging to affinely independent vectors, then every limit direction of the sequence $x_j$ at $p$ belongs to $V_p^\perp$, a subspace of dimension $n+1-k$.\footnote{{By a {\em limit direction} of the sequence $x_j$ at $p$ we mean a limit of a convergent subsequence of the unit vectors $(x_j-p)/\vert x_j-p\vert$, where differences and norms are computed in a chart. The whole sequence of unit vectors need not converge, so there can be several limit directions, but by the compactness of the unit sphere there is at least one. A change of chart acts on the limit directions through its differential at $p$, up to positive factors, so the statement does not depend on the chart. Similarly, a limit direction of a curve $\sigma$ at $p=\sigma(0)$ is a limit direction of a sequence $\sigma(s_j)$, $s_j\to 0^+$. The statement is thus a constraint on all the directions along which $p$ can be approached by the sequence. Sequences $x_j$ as in (c) exist for $k=1$ at every $p\in S^1$, by (d); for $k\ge 2$ their existence at a given $p\in S^k$ is not claimed.}} For $k=1$, $V_p$ is spanned by the spacelike vector $v_+-v_-$, where $v_\pm$ are the initial tangents of the two calibrated rays starting from $p$, so that $V_p^\perp$ is a timelike hyperplane containing $v_++v_-$. Consequently, every achronal set $A\subset M$ satisfies $\mathcal{H}^n(A\cap S)=0$; this applies to every Cauchy hypersurface and to every level set of a time function.}
\end{itemize}

\noindent {\em Lower bounds on the size of $S$: propagation of singularities.}
\begin{itemize}
\item[(d)] Every $p\in S$ is the starting point of a Lipschitz curve $\sigma:[0,\rho]\to S$, $\sigma(s)\ne p$ for $s>0$, such that, with respect to a Riemannian metric, the diameter of $U_{\sigma(s)}$ is bounded from below by a positive constant independent of $s$. In particular, the points of $\textrm{dom}\, b$ from which more than one $b$-calibrated ray starts are never isolated. {If $p\in S^1$, every limit direction of $\sigma$ at $p$, namely every limit of the unit vectors $(\sigma(s_j)-p)/\vert \sigma(s_j)-p\vert$ along a sequence $s_j\to 0^+$, of which there is at least one, belongs to the timelike hyperplane $(v_+-v_-)^\perp$.}
\item[(e)] If $p\in S^k$, $1\le k\le n$, there is a Lipschitz map $\sigma: D\to S$, where $D\subset \mathbb{R}^{n+1-k}$ is a disk centered at the origin, such that $\sigma(0)=p$, $\sigma(D)$ has a tangent space at $p$, and the $\mathcal{H}^{n+1-k}$-density of $\sigma(D)$ at $p$ is positive.
\end{itemize}
\end{theorem}

\begin{proof}
Let $\Phi:\textrm{dom}\, b\to H$, $\Phi(p)=(p,b(p))$. It is a homeomorphism whose inverse is the restriction to $H$ of the projection $\pi:\tilde M\to M$. Since $b$ is locally Lipschitz on $\textrm{dom}\, b$ (Theorem \ref{xxhg}(vi)), $\Phi$ is locally Lipschitz, and $\pi$ is smooth, hence locally Lipschitz. Thus $\Phi$ is locally bi-Lipschitz (with respect to any Riemannian metrics on $M$ and $\tilde M$, or to coordinate metrics). The Hausdorff measures and dimensions, and the countable $\mathcal{H}^d$-rectifiability of a set, namely the property of being contained, up to an $\mathcal{H}^d$-null set, in a countable union of Lipschitz images of subsets of $\mathbb{R}^d$, equivalently of $C^1$ $d$-dimensional submanifolds \cite[3.2.29]{federer69}, are preserved by locally bi-Lipschitz maps, since a Lipschitz map with constant $L$ increases the $\mathcal{H}^d$ measure at most by the factor $L^d$ \cite[2.10.11]{federer69}, \cite[Sect.\ 2.4]{evans92}. Also Lipschitz curves and Lipschitz maps from disks are mapped to Lipschitz ones, and positivity of the $\mathcal{H}^d$-density at a point is preserved, since locally bi-Lipschitz maps distort the $\mathcal{H}^d$ measures, and the radii of balls, by bounded factors. Moreover, the existence of a tangent space is preserved by the smooth projection $\pi$ restricted to $H$: as $H$ is a locally Lipschitz graph along $\p_z$, the contingent cone of a subset of $H$ at a point does not contain $\pm\p_z$, so $\dd\pi$ is injective on it and maps it onto the contingent cone of the projected set, hence it maps a tangent space to a tangent space. The set $S^k$ is mapped by $\Phi$ to the set of points of $H$ at which the semitangents span a vector space of dimension $k+1$, since the semitangents at $\Phi(p)$ are proportional to the vectors $(v,1)$, $v\in U_p$, which lie on the affine hyperplane $\{\dd z=1\}$ of $T_{\Phi(p)}\tilde M$. Since $\dim H=n+1$, (c) follows from \cite[Thm.\ 7(4)]{minguzzi14d} together with \cite[Thm.\ 6(3)]{minguzzi14d} for the countability of $S^{n+1}$, (d) from \cite[Thm.\ 9(a)]{minguzzi14d} and (e) from \cite[Thm.\ 9(b)]{minguzzi14d}, once we observe that $\Phi$ maps $S$ onto the set of non-differentiability points of $H$ (see below). For (d) we also note that the hypothesis of \cite[Thm.\ 9(a)]{minguzzi14d} is satisfied at every non-differentiability point $P$: since no semitangent of $H$ is of the form $(u,0)$, the convex cone generated by the semitangents at $P$ does not coincide with the future causal cone of $\tilde M$ at $P$. Moreover, the aperture angle of this cone at $\Phi(q)$ is bounded from below by a positive constant, locally uniformly in $q$, if and only if the diameter of $U_q$ is, since the cone is generated by the vectors $(v,1)$, $v\in U_q$, and the vectors $v\in U_q$, for $q$ in a compact subset of $\textrm{dom}\, b$, belong to a compact set of timelike vectors (Theorem \ref{xxhg}(v)).

It remains to compare the differentiability of $b$ with that of $H$. In \cite{minguzzi14d}, $H$ is locally expressed as the graph of a locally Lipschitz function $h$ over a spacelike hypersurface along a timelike direction, and its differentiability properties are those of $h$. Since $b$ is also a locally Lipschitz graphing function of $H$, along the direction $\p_z$, which is transverse to $H$ as observed above, these properties coincide. Indeed, a locally Lipschitz graphing function is differentiable at a point if and only if the contingent cone of its graph at the corresponding point is a hyperplane, a condition which does not refer to the graphing direction, since contingent cones are mapped to contingent cones by the differentials of diffeomorphisms; in that case the tangent hyperplane $\ker(\dd z-\dd b)$ of $H$ determines $\dd b$. The equivalence is elementary. Let $f$ be Lipschitz with constant $L$ near $\bar x\in \mathbb{R}^m$ and let $T$ be the contingent cone of its graph at $(\bar x,f(\bar x))$, namely the set of limits of $t_k^{-1}(x_k-\bar x,f(x_k)-f(\bar x))$ with $t_k\to 0^+$ and $x_k\to \bar x$. Then $T\subset\{(u,w):\vert w\vert\le L\vert u\vert\}$ and $T$ contains, for every $u$, some vector $(u,w)$. If $f$ is differentiable at $\bar x$, $T$ is the graph of $\dd f_{\bar x}$. Conversely, if $T$ is a hyperplane, it is the graph of a linear map $A$, and if $f$ were not differentiable at $\bar x$ there would be $\epsilon>0$ and $x_k\to \bar x$ with $\vert f(x_k)-f(\bar x)-A(x_k-\bar x)\vert\ge \epsilon \vert x_k-\bar x\vert$; setting $t_k=\vert x_k-\bar x\vert$ and passing to a subsequence, $t_k^{-1}(x_k-\bar x, f(x_k)-f(\bar x))\to (u,w)\in T$ with $\vert w-Au\vert\ge \epsilon$, a contradiction. By \cite[Thm.\ 7(1)]{minguzzi14d}, $H$ is differentiable at $P$ if and only if exactly one generator starts from $P$, namely if and only if $U_p$ is a singleton, and by \cite[Thm.\ 7(2)]{minguzzi14d} the tangent hyperplanes of $H$ depend continuously on the point of the set of differentiability points, so that $\dd b$, and hence $\nabla b$, is continuous on the set of differentiability points of $b$. This proves (a).

For (b), let $O\subset \textrm{dom}\, b$ be an open set of differentiability points of $b$. Then $H$ is differentiable on the open set $\Phi(O)$, hence $C^{1,1}_{loc}$ there \cite[Thm.\ 8]{minguzzi14d}, namely locally $H=\{F=0\}$ with $F$ of class $C^{1,1}$ and $\dd F\ne 0$. As $\p_z$ is transverse to the tangent hyperplanes of $H$, we have $\p_zF\ne 0$ on $H$, and the implicit function theorem shows that the solution $z=b(p)$ of $F(p,z)=0$ is $C^1$, with $\dd b=-(\p_z F)^{-1}\dd_pF$ evaluated at $(p,b(p))$. The right-hand side is locally Lipschitz, being the composition of locally Lipschitz maps, so $b$ is $C^{1,1}_{loc}$ on $O$. Under the assumption of Theorem \ref{thmcone} we can take $O=\textrm{dom}\, b$.

{{\em Tangent directions of the strata.} Let $p\in S^k$, $k\ge 1$, and let $C$ be a closed coordinate ball centered at $p$ as in the proof of Theorem \ref{xxhg}(vi). For $x\in C$ and $u\in U_x$, the function $b_{x,1}$ of (\ref{cmby}), built with the calibrated ray with initial tangent $u$, satisfies $b\le b_{x,1}$ on $C$, $b_{x,1}(x)=b(x)$ and $\dd (b_{x,1})_x=-g(u,\cdot)$, since the gradient of $d(\cdot,y)$ at a point $x\in I^-(y)$ at which it is smooth is the future directed unit tangent at $x$ of the maximizing geodesic to $y$. Moreover, its second coordinate derivatives are bounded on $C$ uniformly in $x$ and $u$. By Taylor's formula along the coordinate segments, there is $G'>0$ such that
\[
b(y)\le b(x)-g_x(u,y-x)+G'\vert y-x\vert^2, \qquad x,y\in C,\ u\in U_x,
\]
where $y-x$ and $\vert y-x\vert$ are computed in coordinates. Let $x_j\to p$ and $u^i_j\in U_{x_j}$ be as in (c), with limits $u^i$. A limit of calibrated rays is a calibrated ray (Theorem \ref{xxhg}(v) and its proof), so $u^i\in U_p$, and, as the $u^i$ are $k+1$ affinely independent vectors of $U_p$, whose affine span has dimension $k$, the differences $u^i-u^0$ span $V_p$. Applying the inequality with $(x,y,u)=(p,x_j,u^i)$ and with $(x,y,u)=(x_j,p,u^l_j)$ and summing, we get
\[
0\le g_{x_j}(u^l_j,x_j-p)-g_p(u^i,x_j-p)+2G'\vert x_j-p\vert^2 .
\]
Let $\xi$ be a limit direction of the sequence $x_j$ at $p$. Dividing by $\vert x_j-p\vert$ and passing to a subsequence along which $(x_j-p)/\vert x_j-p\vert$ converges to $\xi$, we get $g(u^l-u^i,\xi)\ge 0$ for all $i,l$, hence, exchanging $i$ and $l$, $g(u^l-u^i,\xi)=0$, namely $\xi\in V_p^\perp$. This is the Lorentzian version of an argument of Alberti, Ambrosio and Cannarsa \cite[Prop.\ 2.2]{alberti92}. For $k=1$, $U_p=\{v_+,v_-\}$ and $g(v_+-v_-,v_+-v_-)=-2-2g(v_+,v_-)>0$ by the reverse Cauchy--Schwarz inequality, so $V_p$ is spacelike and $V_p^\perp$ is a timelike hyperplane; it contains $v_++v_-$, as $g(v_++v_-,v_+-v_-)=0$. The additional statement in (d) follows. Indeed, let $s_j\to 0^+$ be a sequence along which the unit vectors $(\sigma(s_j)-p)/\vert \sigma(s_j)-p\vert$ converge to some $\xi$. Two vectors of $U_{\sigma(s_j)}$ at distance bounded from below belong to a compact set (Theorem \ref{xxhg}(v)), so, passing to a further subsequence, which does not change $\xi$, they converge to two distinct, hence affinely independent, vectors of $U_p$. The previous argument, applied to the sequence $x_j=\sigma(s_j)$ with $k=1$, gives $\xi\in V_p^\perp$.

{\em Achronal sets.} Let $A$ be achronal. As $\mathcal{H}^n$ is an outer measure, $A$ need not be measurable. The set $A$ is contained in $Y:=\p I^+(A)$, since its points do not belong to $I^+(A)$, by achronality, but are limits of points of $I^+(A)$, so it suffices to show that $\mathcal{H}^n(Y\cap S)=0$. The set $Y$ is an achronal boundary, hence a locally Lipschitz hypersurface which, in suitable charts in which $\p_0$ is timelike, is the graph $x^0=f(\bar x)$ of a Lipschitz function \cite[Prop.\ 6.3.1]{hawking73}. Fix an auxiliary Riemannian metric and, for $m\in\mathbb{N}$, let $F_m:=\{x\in Y\cap \textrm{dom}\, b: \textrm{diam}\, U_x\ge 1/m\}$. These sets cover $Y\cap S$, and they are closed in $Y\cap \textrm{dom}\, b$: if $x_j\in F_m$ converge to $x\in\textrm{dom}\, b$, two vectors of $U_{x_j}$ at distance at least $1/m$ converge, along a subsequence, to two vectors of $U_x$ at distance at least $1/m$, by the compactness of Theorem \ref{xxhg}(v) and the fact that a limit of calibrated rays is a calibrated ray. Suppose, by contradiction, that $\mathcal{H}^n(Y\cap S)>0$. As the sets $S^k$, $k\ge 2$, are $\mathcal{H}^n$-null, being countably $\mathcal{H}^{n+1-k}$-rectifiable, there are $m$ and a chart as above, whose domain has compact closure in $\textrm{dom}\, b$, such that $F:=F_m\cap \textrm{graph}\, f$ satisfies $\mathcal{H}^n(F\cap S^1)>0$. The projection $F'$ of $F$ on the $\bar x$ coordinates is $\sigma$-compact, hence Lebesgue measurable, and it has positive Lebesgue measure, while the projection of $F\backslash S^1\subset \cup_{k\ge 2}S^k$ is Lebesgue null, the projection being Lipschitz. Thus there is a Lebesgue density point $\bar x_0$ of $F'$ at which $f$ is differentiable, which exists by Rademacher's theorem, and such that $x_0:=(f(\bar x_0),\bar x_0)\in S^1$. For every unit vector $e$ there are $\bar x_j\in F'$ such that $(\bar x_j-\bar x_0)/\vert\bar x_j-\bar x_0\vert\to e$, for otherwise the intersection of a small ball centered at $\bar x_0$ with a cone of vertex $\bar x_0$ and axis $e$ would contain no point of $F'$, contradicting the density. The directions of $x_j-x_0$, where $x_j:=(f(\bar x_j),\bar x_j)\in F$, then converge to that of $(\dd f_{\bar x_0}(e),e)$, and these vectors span the tangent hyperplane $T_{x_0}Y$. As $U_{x_j}$ contains two vectors at distance at least $1/m$, which belong to a compact set (Theorem \ref{xxhg}(v)), along a subsequence they converge to two distinct vectors of $U_{x_0}$, so the first part gives $T_{x_0}Y\subset V_{x_0}^\perp$. The two hyperplanes thus coincide, and $T_{x_0}Y$ is timelike. This is impossible: if $T_{x_0}Y$ contained a future directed timelike vector $X$, the points of $Y$ obtained from $x_0$ by small displacements in the direction of $X$ would belong to $I^+(x_0)$, contradicting the achronality of $Y$. Thus $\mathcal{H}^n(Y\cap S)=0$, and a fortiori $\mathcal{H}^n(A\cap S)=0$. Finally, Cauchy hypersurfaces and level sets of time functions are achronal.}
\end{proof}

{
{\em Causal character of the strata.} For $k\ge 2$ the subspace $V_p^\perp$ of Theorem \ref{thmfine}(c) contains a timelike vector if and only if $V_p$ is spacelike, namely if and only if the affine span of $U_p$ is a spacelike affine subspace. This always holds for $k=1$, but for $k\ge 2$ it does not follow from the local properties of $b$ used above. For instance, on $2+1$ Minkowski spacetime with coordinates $(t,x,y)$ consider the vectors $u_i=(\cosh s_i,\sinh s_i,0)$, $s_i=-1,0,1$, and the function $\min_i(-g(u_i,\cdot))$, a local model which shares with $b$ the properties used above: it is locally semi-concave, its gradient is a past directed unit timelike vector where it exists, and the future rays with velocities $u_i$ starting from the points at which the three linear functions coincide are calibrated. These points form the spacelike line in the direction of $\p_y$, which is the set $S^2$ of this function, and it is contained in the slice $\{t=0\}$. Thus the local properties of $b$ used above do not exclude spacelike strata $S^k$, $k\ge 2$, possibly contained in a Cauchy hypersurface, although the example does not show that they occur for the Cauchy Busemann functions of spacetimes with compact Cauchy hypersurfaces. In any case, this does not affect the last statement in Theorem \ref{thmfine}(c), as these strata are $\mathcal{H}^n$-null.

{\em Global propagation.} Theorem \ref{thmfine}(d) is a local statement. Singularities do not propagate toward the future along the calibrated rays, whose interior points are differentiability points (Proposition \ref{propcrease}), and the singular set may end in the future, as happens to the cut locus of a hypersurface at its focal points. Toward the past the propagation is global \cite{minguzzi26f}. In a chart, the superdifferential of $b$ at $x$ is $D^+b(x)=\textrm{co}\{-g(u,\cdot): u\in U_x\}$, so the Lorentzian analogs of the generalized characteristics of semi-concave solutions of Hamilton--Jacobi equations \cite[Def.\ 2.8]{cannarsa17} are the absolutely continuous arcs $x$ in $\textrm{dom}\, b$ whose velocity is, almost everywhere, a positive multiple of a vector of $-\textrm{co}\, U_x$. They are past directed timelike curves, since $\textrm{co}\, U_x$ consists of future directed timelike vectors. Let $(M,g)$ be future timelike complete with compact Cauchy hypersurfaces. Then \cite{minguzzi26f}: (i) an arc of this type starting from a point of $S$ remains in $S$; (ii) from every point of $\textrm{dom}\, b$ there starts exactly one past inextendible arc of this type, up to reparametrizations. Consequently, every $p\in S$ is the future endpoint of a past inextendible timelike curve contained in $S$, and $S$ meets every Cauchy hypersurface $S'$ such that $p\in J^+(S')$. For $p\in S^1$ the right velocity of this curve at $p$ is proportional to $-(v_++v_-)$, which belongs to the hyperplane $(v_+-v_-)^\perp$ of Theorem \ref{thmfine}(d). These results do not follow from the global propagation theorem for generalized characteristics \cite[Thm.\ 3.11]{cannarsa17}, which is proved for Tonelli Hamiltonians, while the Hamiltonian $\frac12 g^{-1}$ is not convex; the proof is based instead on the monotonicity, along these arcs, of the maximal Lorentzian length of the vectors of $\textrm{co}\, U_x$, which equals one precisely at the differentiability points.\par}

The argument of Theorem \ref{thmcone} also shows that a globally defined $C^1$ solution of the eikonal equation produces timelike lines.

\begin{proposition} \label{propeik}
Let $u:M\to \mathbb{R}$ be a $C^1$ function such that $g(\nabla u,\nabla u)=-1$ and $\nabla u$ is past-directed. Then every integral curve of the continuous vector field $-\nabla u$ extends to an inextendible integral curve, which is a timelike line along which $u$ increases by the proper time. No completeness assumption is needed for this conclusion. If one of these lines is complete, for instance if $(M,g)$ is timelike complete, then under TCC the spacetime splits. Conversely, in a product $(\mathbb{R}\times S,-\dd t^2+h)$ the function $u=t$ has these properties.
\end{proposition}

The $C^1$ regularity required in Proposition \ref{propeik}, or in Theorem \ref{thmcone}, could be obtained by imposing a suitable elliptic equation on a Lipschitz eikonal function, such as $b^+$ on $\textrm{dom}\, b^+$, and then invoking standard elliptic regularity theory \cite{gilbarg01}. For instance, setting $\vert \nabla u\vert:=(-g(\nabla u,\nabla u))^{1/2}$, on eikonal functions the equation $\Box u=0$ coincides with $\textrm{div}(\vert \nabla u\vert^{p-2}\nabla u)=0$, which for $p<1$ is elliptic on functions with timelike gradient \cite{braun24}, whereas the d'Alembertian $\Box$ itself is a hyperbolic operator. Lipschitz eikonal weak solutions of this equation are then smooth. This is a local statement, which does not require the compactness of $M$ or of $\Sigma$: on a compact subset of the domain the Lipschitz bound and the eikonal equation confine $\nabla u$ to a compact set of timelike vectors, on which the operator is uniformly elliptic, so the difference quotient method, the De Giorgi--Nash--Moser estimates for the derivatives of $u$ and the Schauder theory apply \cite[Chs.\ 6, 8, 13]{gilbarg01}; this is the regularity argument used in the elliptic proof of the splitting theorem \cite{braun24}. By Proposition \ref{propeik} these solutions produce timelike lines; if one of these lines is complete, for instance if $(M,g)$ is timelike complete, then under TCC the spacetime splits. Smoothness alone is not sufficient: the function $u=t$ on the half-cylinder of Section \ref{secconv} is a smooth solution of this equation. Note, however, that the geometrical content of the argument lies in the eikonal equation, which provides the timelike lines, and it is not clear that this requirement could be dispensed with: we do not know whether, in a timelike complete spacetime satisfying TCC, a global solution of such an elliptic equation whose gradient is timelike, but not of constant Lorentzian length, forces the splitting.\footnote{Recently, this problem was solved by Ohanyan in the setting of compact Cauchy hypersurfaces: in a timelike complete spacetime with compact Cauchy hypersurfaces satisfying TCC, every locally Lipschitz function whose gradient is locally uniformly timelike and which weakly solves this equation for some $p<1$ is smooth and Cauchy temporal, and the spacetime splits \cite[Thms.\ 1.5, 1.6, 1.7]{ohanyan26b}.}

\begin{proof}
Let $X:=-\nabla u$, a continuous future-directed unit timelike vector field. The function $u$ is rushing. Indeed, for every future-directed causal vector $w$ the reverse Cauchy--Schwarz inequality gives $\dd u(w)=-g(X,w)\ge \sqrt{-g(w,w)}$, so $u(q)-u(p)\ge \ell(\gamma)$ for every future-directed causal curve $\gamma$ from $p$ to $q$, hence $u(q)-u(p)\ge d(p,q)$.

Let $\beta$ be an integral curve of $X$, which exists by Peano's theorem. Then $g(\dot \beta,\dot \beta)=-1$ and $\frac{\dd}{\dd t}u(\beta(t))=\dd u(X)=1$, so $u$ increases by the Lorentzian length along $\beta$, and by Proposition \ref{jjtr} every segment of $\beta$ is maximizing. Proposition \ref{jjtr} uses only the rushing property and the reverse triangle inequality, so no completeness assumption enters this argument. Hence $\beta$ is a timelike geodesic. Let $\gamma:(a_-,a_+)\to M$ be the inextendible timelike geodesic extending $\beta$, and let $A:=\{t\in (a_-,a_+): \dot \gamma(t)=X(\gamma(t))\}$. The set $A$ is non-empty and closed. It is open: if $t_0\in A$, an integral curve of $X$ through $\gamma(t_0)$ defined near $t_0$ is, by the previous argument, a timelike geodesic with the same initial data as $\gamma$, hence it coincides with $\gamma$ near $t_0$. Thus $A=(a_-,a_+)$, $\gamma$ is an inextendible integral curve of $X$, and every segment of it is maximizing, namely $\gamma$ is a timelike line. If $\gamma$ is complete, under TCC the spacetime splits by the Lorentzian splitting theorem. The converse is clear since $\nabla t=-\p_t$.
\end{proof}

Note that integral curves of $-\nabla u$ are thus unique, even though $-\nabla u$ is only continuous. Locally Lipschitz functions satisfying the eikonal equation almost everywhere always exist: for instance, the signed Lorentzian distance from a smooth spacelike compact Cauchy hypersurface, which exists by \cite{bernal04}. Thus the $C^1$ regularity is essential. By Theorem \ref{thmcone} and Example \ref{exstatic} below, if $(M,g)$ is timelike complete and satisfies TCC, the existence of such a function $u$ is equivalent to $b^+$ being $C^1$ on $\textrm{dom}\, b^+$.

\subsection{The canonical calibrating semi-flow for initial singularity spacetimes} \label{seccosm}

By Example \ref{exwarp}, under future timelike completeness alone the Cauchy Busemann function can be finite on the whole spacetime, as it happens in cosmological models with an initial singularity. In this setting there is a natural choice of the reference hypersurfaces, which removes the dependence of $b$ on $\Sigma$. Andersson, Galloway and Howard \cite{andersson98} introduced the {\em cosmological time function}
\[
\tau(q):=\sup\{d(p,q): p\in J^-(q)\},
\]
the supremum of the Lorentzian lengths of the past directed causal curves starting from $q$, and called it {\em regular} if $\tau(q)<+\infty$ for every $q\in M$ and $\tau\to 0$ along every past inextendible causal curve \cite[Def.\ 1.1]{andersson98}. In particular, a spacetime with regular cosmological time is past timelike incomplete. If $\tau$ is regular, then $\tau$ is a time function \cite[Thm.\ 1.2(ii)]{andersson98}, it is rushing, namely $\tau(q)\ge \tau(p)+d(p,q)$ for $p\le q$, which follows from the definition, and through every $q\in M$ passes a future directed timelike unit speed geodesic $\gamma_q:(0,\tau(q)]\to M$, maximizing on each segment, such that $\gamma_q(\tau(q))=q$ and $\tau(\gamma_q(t))=t$ \cite[Thm.\ 1.2(iii), Prop.\ 2.3]{andersson98}. In particular, since $\gamma_q$ is maximizing and has unit speed,
\begin{equation} \label{gaq}
d(\gamma_q(s),\gamma_q(t))=t-s=\tau(\gamma_q(t))-\tau(\gamma_q(s)), \qquad 0<s\le t\le \tau(q).
\end{equation}
In this subsection $\tau$ denotes the cosmological time and not the parameter of the level sets $\Sigma_\tau$. For $a>0$ let $T_a:=\tau^{-1}(a)$, and set
\begin{equation} \label{defbm}
b^{\textnormal{\textsc{m}}}(p):=\lim_{r\to +\infty}[r-d(p,T_r)], \qquad p\in M .
\end{equation}
The existence of the limit is proved in Proposition \ref{propcosm}(b) below. The superscript stresses that $b^{\textnormal{\textsc{m}}}$ depends only on the spacetime $(M,g)$, and not on the choice of a Cauchy hypersurface. Let $\tilde M_-:=M\times (-\infty,0]\subset \tilde M$.

\begin{proposition} \label{propcosm}
Let $(M,g)$ be future timelike complete with regular cosmological time $\tau$.
\begin{itemize}
\item[(a)] For every $a>0$, $T_a$ is a compact Cauchy hypersurface, $d^+_{T_a}=\tau-a$ on $J^+(T_a)$, and the level sets of $d^+_{T_a}$ are $(T_a)_s=T_{a+s}$, $s>0$.
\item[(b)] The limit (\ref{defbm}) exists in $(-\infty,+\infty]$ for every $p\in M$, and $b^{\textnormal{\textsc{m}}}=b^{T_a}+a$ for every $a>0$, where $b^{T_a}$ is the Cauchy Busemann function of $T_a$. Moreover, $b^{\textnormal{\textsc{m}}}\ge \tau$.
\item[(c)] For every compact Cauchy hypersurface $S$, setting $c_-:=\min_S\tau$ and $c_+:=\max_S \tau$,
\[
b^{\textnormal{\textsc{m}}}-c_+\le b^S\le b^{\textnormal{\textsc{m}}}-c_- .
\]
In particular $\textrm{dom}\, b^S=\textrm{dom}\, b^{\textnormal{\textsc{m}}}$, and if $S_n$ are compact Cauchy hypersurfaces with $\max_{S_n}\tau\to 0$, then $b^{S_n}\to b^{\textnormal{\textsc{m}}}$ uniformly on $\textrm{dom}\, b^{\textnormal{\textsc{m}}}$.
\item[(d)] {\em Foliation-free formula.} For every future inextendible causal curve $x:[0,+\infty)\to M$ the function $t\mapsto \tau(x(t))-\ell(x\vert_{[0,t]})$ is non-decreasing, so the limit $L^{\textnormal{\textsc{m}}}(x):=\lim_{t\to +\infty}[\tau(x(t))-\ell(x\vert_{[0,t]})]\in (-\infty,+\infty]$ exists. For every $p\in M$
\[
b^{\textnormal{\textsc{m}}}(p)=\inf\{L^{\textnormal{\textsc{m}}}(x): x \textrm{ future inextendible causal curve starting from } p\},
\]
and if $p\in \textrm{dom}\, b^{\textnormal{\textsc{m}}}$ the infimum is attained precisely on the $b^{\textnormal{\textsc{m}}}$-calibrated curves.
\item[(e)] $\tilde I^+(\tilde M_-)=\{(p,z): z<\tau(p)\}$ and $\tilde D^-(\tilde I^+(\tilde M_-))=\{(p,z): z<b^{\textnormal{\textsc{m}}}(p)\}$. Thus the past Cauchy horizon $\tilde H^-(\tilde I^+(\tilde M_-))$ is the graph of $b^{\textnormal{\textsc{m}}}$ over $\textrm{dom}\, b^{\textnormal{\textsc{m}}}$; for every $a>0$ it is the image of the horizon $H$ associated with $T_a$ under the translation $(p,z)\mapsto (p,z+a)$.
\end{itemize}
\end{proposition}

That $T_a$ is a compact Cauchy hypersurface was also proved by Galloway and Garc\'ia-Heveling \cite[Thm.\ 5]{galloway25}.

\begin{proof}
(a). The set $T_a$ is acausal since $\tau$ is a time function. Let $\Sigma$ be a compact Cauchy hypersurface. For $q\in J^+(\Sigma)$ and $y\in \Sigma\cap J^-(q)$, the rushing property gives $\tau(q)\ge \min_\Sigma \tau+d(y,q)$, so that $\tau(q)\ge \min_\Sigma \tau+d^+_\Sigma(q)$. Along a future inextendible causal curve $d^+_\Sigma\to +\infty$, as it meets every $\Sigma_k$ (Theorem \ref{kvvb}), so $\tau\to +\infty$, while $\tau\to 0$ in the past direction by regularity. As $\tau$ is continuous and strictly increasing along inextendible causal curves, each of them meets $T_a$ exactly once, so $T_a$ is a Cauchy hypersurface, compact as it is homeomorphic to $\Sigma$. For $q\in J^+(T_a)$ the rushing property gives $d(T_a,q)\le \tau(q)-a$, while the geodesic $\gamma_q$ meets $T_a$ at $\gamma_q(a)$, with $d(\gamma_q(a),q)=\tau(q)-a$ by (\ref{gaq}). Thus $d^+_{T_a}=\tau-a$ on $J^+(T_a)$ and $(T_a)_s=T_{a+s}$.

(b). By (a), $b^{T_a}(p)=\lim_{k}[k-d(p,T_{a+k})]$, so $b^{T_a}(p)+a=\lim_k [(a+k)-d(p,T_{a+k})]$. The function $r\mapsto r-d(p,T_r)$ is non-decreasing, since, by the rushing property of $\tau$, every causal curve from $p$ to $T_{r'}$, $r<r'$, either starts in $J^+(T_r)$ and has Lorentzian length at most $r'-r$, or crosses $T_r$, so that $d(p,T_{r'})\le d(p,T_r)+r'-r$. Thus the limit (\ref{defbm}) exists and coincides with $b^{T_a}+a$. By (\ref{cmlr}) applied to $T_a$, $b^{\textnormal{\textsc{m}}}\ge d_{T_a}+a$, which equals $\tau$ on $J^+(T_a)$; letting $a\to 0$ we get $b^{\textnormal{\textsc{m}}}\ge \tau$ on $M$.

(c). For $q\in J^+(S)$ the rushing property gives $d^+_S(q)\le \tau(q)-c_-$, while the geodesic $\gamma_q$, which is past inextendible as $\tau\to 0$ along it, meets $S$ at a point $y$, so that, by (\ref{gaq}), $d^+_S(q)\ge d(y,q)=\tau(q)-\tau(y)\ge \tau(q)-c_+$. Thus for $k>0$ the level set $S_k$ is contained in $J^+(T_{c_-+k})\cap J^-(T_{c_++k})$. Let $p\in M$. A maximizing curve from $p$ to $T_{c_-+k}$ can be prolonged to the future up to a point of the Cauchy hypersurface $S_k$, and a maximizing curve from $p$ to $S_k$ can be prolonged up to a point of $T_{c_++k}$, so that
\[
d(p,T_{c_-+k})\le d(p,S_k)\le d(p,T_{c_++k}),
\]
where the distances vanish if the corresponding sets do not intersect $J^+(p)$. Thus $(c_++k)-d(p,T_{c_++k})-c_+\le k-d(p,S_k)\le (c_-+k)-d(p,T_{c_-+k})-c_-$, and letting $k\to +\infty$ we get the claimed inequalities, which imply $\textrm{dom}\, b^S=\textrm{dom}\, b^{\textnormal{\textsc{m}}}$ and $\vert b^{S_n}-b^{\textnormal{\textsc{m}}}\vert\le \max_{S_n}\tau$ on $\textrm{dom}\, b^{\textnormal{\textsc{m}}}$.

(d). The monotonicity follows from the rushing property of $\tau$. Let $a>0$. By (a), for $t$ so large that $x(t)\in J^+(T_a)$, the function $F_x$ of Theorem \ref{thmintr} relative to $T_a$ is $F_x(t)=\tau(x(t))-a-\ell(x\vert_{[0,t]})$, thus $L(x)=L^{\textnormal{\textsc{m}}}(x)-a$. Since $b^{\textnormal{\textsc{m}}}=b^{T_a}+a$ by (b), the claim follows from Theorem \ref{thmintr} applied to $T_a$, as $b^{\textnormal{\textsc{m}}}$-calibration coincides with $b^{T_a}$-calibration.

(e). Recall that a curve $t\mapsto (x(t),z(t))$ in $\tilde M$ is causal (resp.\ timelike) if and only if $x$ is causal and $\vert \dot z\vert\le \sqrt{-g(\dot x,\dot x)}$ (resp.\ $x$ is timelike and $\vert \dot z\vert< \sqrt{-g(\dot x,\dot x)}$) (Section \ref{secprod}). Let $X:=\tilde I^+(\tilde M_-)$.

{\em Step 1: $X=\{(p,z): z<\tau(p)\}$.} Let $(p,z)\in X$, so that there is a timelike curve $(x(t),\zeta(t))$, $t\in[0,1]$, from some $(q,w)$, $w\le 0$, to $(p,z)$. Then $x$ is a timelike curve from $q$ to $p$, and integrating $\dot\zeta<\sqrt{-g(\dot x,\dot x)}$ we get $z-w<\ell(x)\le d(q,p)\le \tau(p)$, hence $z<\tau(p)$. Conversely, let $z<\tau(p)$. By the definition of $\tau$ there is $q\in I^-(p)$ with $d(q,p)>\max\{z,0\}$, hence a timelike curve $x:[0,1]\to M$ from $q$ to $p$ with $\ell:=\ell(x)>\max\{z,0\}$. Choose $w\le 0$ with $\vert z-w\vert<\ell$, which is possible: for $z>0$ take $w=0$, for $z\le 0$ just take $w=z$. Let $\zeta(t):=w+(z-w)\ell(x\vert_{[0,t]})/\ell$. Then $\vert\dot\zeta\vert=\vert z-w\vert\,\ell^{-1}\sqrt{-g(\dot x,\dot x)}<\sqrt{-g(\dot x,\dot x)}$, so $(x,\zeta)$ is a timelike curve from $(q,w)\in \tilde M_-$ to $(p,z)$, and $(p,z)\in X$. In particular $X$ is downward closed: $(p,z)\in X$ and $z'\le z$ imply $(p,z')\in X$.

{\em Step 2: lifts.} Let $(p,z)\in \tilde M$ and let $x:[0,+\infty)\to M$ be a future inextendible causal curve starting from $p$. Its lift $\hat x(t):=(x(t),z+\ell(x\vert_{[0,t]}))$ is a causal curve of $\tilde M$, future inextendible because its projection is. By Step 1, $\hat x$ meets $X$ if and only if $z+\ell(x\vert_{[0,t]})<\tau(x(t))$ for some $t$, namely, the bracket in (d) being non-decreasing, if and only if $z<L^{\textnormal{\textsc{m}}}(x)$.

{\em Step 3: reduction to lifts.} Every future inextendible causal curve of $\tilde M$ starting from $(p,z)$ reads $(x(t),\zeta(t))$, where $x$ is a future inextendible causal curve of $M$ starting from $p$ (Lemma \ref{lemgh}) and, integrating $\dot\zeta\le\sqrt{-g(\dot x,\dot x)}$, $\zeta(t)\le z+\ell(x\vert_{[0,t]})$. Thus, since $X$ is downward closed, the curve meets $X$ whenever the lift $\hat x$ does. Conversely, the lifts are themselves future inextendible causal curves starting from $(p,z)$. Therefore $(p,z)\in\tilde D^-(X)$ if and only if every lift $\hat x$ meets $X$, namely, by Step 2, if and only if $z<L^{\textnormal{\textsc{m}}}(x)$ for every future inextendible causal curve $x$ starting from $p$.

{\em Step 4: conclusion.} If $p\notin \textrm{dom}\, b^{\textnormal{\textsc{m}}}$, by (d) $L^{\textnormal{\textsc{m}}}(x)=+\infty$ for every $x$, so $(p,z)\in\tilde D^-(X)$ for every $z$. If $p\in \textrm{dom}\, b^{\textnormal{\textsc{m}}}$, by (d) and Proposition \ref{bber} the infimum of $L^{\textnormal{\textsc{m}}}(x)$ is attained, at a $b^{\textnormal{\textsc{m}}}$-calibrated ray, and equals $b^{\textnormal{\textsc{m}}}(p)$, so $(p,z)\in\tilde D^-(X)$ if and only if $z<b^{\textnormal{\textsc{m}}}(p)$. This proves $\tilde D^-(X)=\{(p,z): z<b^{\textnormal{\textsc{m}}}(p)\}$. Since $b^{\textnormal{\textsc{m}}}=b^{T_a}+a$, this set is the image under $(p,z)\mapsto(p,z+a)$ of the set (\ref{mxwe}) relative to $T_a$, so the last statement follows from Theorem \ref{xxhg}(iv) applied to $T_a$.
\end{proof}

The quantity $\tau(x(t))-\ell(x\vert_{[0,t]})$ is the delay, at time $t$, of a clock carried along $x$ with respect to the cosmological time, namely with respect to the clocks which have been running since the initial singularity along the fastest histories. Thus $L^{\textnormal{\textsc{m}}}(x)$ is the asymptotic delay of the observer $x$, and $b^{\textnormal{\textsc{m}}}(p)$ is the minimal asymptotic delay among all observers leaving $p$, a quantity which does not involve any choice.

Thus, in a future timelike complete spacetime with regular cosmological time, $b^{\textnormal{\textsc{m}}}$ is a canonical Cauchy Busemann function, independent of any choice of Cauchy hypersurface. Since $b^{\textnormal{\textsc{m}}}=b^{T_a}+a$, the $b^{\textnormal{\textsc{m}}}$-calibrated curves are the $b^{T_a}$-calibrated curves for any $a>0$, and all the results of Sections \ref{secdist}--\ref{sechoro} apply to $b^{\textnormal{\textsc{m}}}$, except for the statements which assume past timelike completeness, which fails here. Indeed, these results use only that $\Sigma$ is a compact Cauchy hypersurface, so they apply to the hypersurfaces $T_a$, which are in general only Lipschitz. The results of the paper which require $\Sigma$ to be smooth and spacelike, namely Proposition \ref{prop39}, Lemma \ref{lemsmoothT}, Theorem \ref{thmhawk} and Corollary \ref{corhawk}, apply to $T_a$ only if $T_a$ is smooth and spacelike. In particular, from every point of $\textrm{dom}\, b^{\textnormal{\textsc{m}}}$ starts at least one $b^{\textnormal{\textsc{m}}}$-calibrated ray, which is timelike (Proposition \ref{bber}, Theorem \ref{xxhg}); $b^{\textnormal{\textsc{m}}}$ is locally semi-concave on $\textrm{dom}\, b^{\textnormal{\textsc{m}}}$ (Theorem \ref{xxhg}(vi)); it is differentiable precisely at the points from which exactly one $b^{\textnormal{\textsc{m}}}$-calibrated ray starts (Proposition \ref{propcrease}); and Theorem \ref{thmfine} holds for $b^{\textnormal{\textsc{m}}}$ without changes: the set of non-differentiability points of $b^{\textnormal{\textsc{m}}}$ is countably $\mathcal{H}^n$-rectifiable, $b^{\textnormal{\textsc{m}}}$ is $C^{1,1}_{loc}$ on the interior of its complement, and the singularities propagate as in Theorem \ref{thmfine}(d).

{\em The canonical calibrating semi-flow.} Let $D\subset \textrm{dom}\, b^{\textnormal{\textsc{m}}}$ be the set of differentiability points of $b^{\textnormal{\textsc{m}}}$. For $p\in D$ let $\alpha_p$ be the unique proper time parametrized $b^{\textnormal{\textsc{m}}}$-calibrated ray starting from $p$, so that $\dot\alpha_p(0)=-\nabla b^{\textnormal{\textsc{m}}}(p)$ (Lemma \ref{lemgrad}), and set $\Phi_s(p):=\alpha_p(s)=\exp_p(-s\nabla b^{\textnormal{\textsc{m}}}(p))$, $s\ge 0$. By Proposition \ref{propcrease}, $\Phi_s(D)\subset D$ and, by uniqueness, $\Phi_{s+t}=\Phi_s\circ \Phi_t$ for $s,t\ge 0$. Moreover $b^{\textnormal{\textsc{m}}}(\Phi_s(p))=b^{\textnormal{\textsc{m}}}(p)+s$, and $\Phi:[0,+\infty)\times D\to D$ is continuous for the topology induced on $D$, being the composition of the continuous map $(s,p)\mapsto -s\nabla b^{\textnormal{\textsc{m}}}(p)\in TM$ (Theorem \ref{thmfine}(a)) with the exponential map, which is smooth on its open domain. The openness of $D$ is not needed, and no uniqueness theorem for differential equations is involved, since $\Phi$ is defined through geodesics. Every point of $\textrm{dom}\, b^{\textnormal{\textsc{m}}}\backslash D$ is the starting point of at least two $b^{\textnormal{\textsc{m}}}$-calibrated rays, all of which enter $D$ immediately, so the branching of the calibrating congruence happens only at the initial points. The semi-flow $\Phi$ is thus a canonical congruence of geodesic observers on $D$, synchronized by $b^{\textnormal{\textsc{m}}}$, and characterized by (d) as the observers with minimal asymptotic delay with respect to the cosmological time. In any case $\textrm{dom}\, b^{\textnormal{\textsc{m}}}\backslash D$ has Hausdorff dimension at most $n$, so $\Phi$ is defined on a set of full measure in $\textrm{dom}\, b^{\textnormal{\textsc{m}}}$; when $\textrm{dom}\, b^{\textnormal{\textsc{m}}}=M$, as in Example \ref{exwarp}, it is defined on a set of full measure in $M$. {Moreover, by Theorem \ref{thmfine}(c), $\mathcal{H}^n(Y\backslash D)=0$ for every achronal hypersurface $Y\subset \textrm{dom}\, b^{\textnormal{\textsc{m}}}$, for instance for every slice $T_a\cap \textrm{dom}\, b^{\textnormal{\textsc{m}}}$: the canonical semi-flow is defined at $\mathcal{H}^n$-almost every point of every Cauchy hypersurface intersected with $\textrm{dom}\, b^{\textnormal{\textsc{m}}}$.}

By Proposition \ref{propcosm}(e), $\tau$ and $b^{\textnormal{\textsc{m}}}$ are the graphing functions of the boundary of the set $\tilde I^+(\tilde M_-)$ and of the boundary of its past domain of dependence, respectively; the generators of the former are the lightlike lifts $t\mapsto (\gamma_q(t),t)$ of the geodesics $\gamma_q$, which are past rays, while the generators of the latter are the lightlike lifts of the $b^{\textnormal{\textsc{m}}}$-calibrated curves, which are future rays. For the warped products of Example \ref{exwarp} with $\alpha>-\infty$ we have $\tau=t-\alpha$, as the integral curves of $\p_t$ maximize the distance, and $b^{\textnormal{\textsc{m}}}=b^{S_{t_0}}+t_0-\alpha=\tau$, so that $\Phi$ is the semi-flow of the comoving observers. In general $b^{\textnormal{\textsc{m}}}$ might be larger than $\tau$.

\section{Cauchy horospheres} \label{sechoro}

In this section $b$ is the Cauchy Busemann function of a compact Cauchy hypersurface $\Sigma$. The following result is related to \cite[Cor.\ 5.6]{galloway96}.

\begin{proposition} \label{cllp}
Assume that $(M,g)$ is future timelike complete. For every constant $c\in \mathbb{R}$, the set
\begin{equation} \label{defbc}
B_c:=b^{-1}(c)
\end{equation}
is an acausal edgeless set, namely a partial Cauchy hypersurface, and $J^-(B_c)=D^-(B_c)= b^{-1}((-\infty,c])$. The set $B_c$ is non-empty for $c\ge 0$, and for every $c\in \mathbb{R}$ if $(M,g)$ is also past timelike complete. For every Cauchy hypersurface $S$, $J^+(S)\cap J^-(B_c)$ is compact.

Let $x: [0,+\infty)\to M$ be a future-inextendible $b$-calibrated curve which starts from $B_c$ (projection of a generator), then $x$ is a  $B_c$-ray. In particular, from every point of $B_c$ starts a $B_c$-ray. Moreover, $x$ intersects all $B_d$ with $d>c$ and for each $d\ge c$, with suitably redefined domain, $x$ is a $B_d$-ray.
\end{proposition}

Observe that we are not making claims on the connectedness of $B_c$.

\begin{proof}
Let $\sigma$ be a proper time parametrized future $\Sigma$-ray. By Prop.\ \ref{ncho}, $b(\sigma(c))=c$ for every $c\ge 0$, thus $B_c\ne \emptyset$ for $c\ge 0$. If $(M,g)$ is past timelike complete, the timelike geodesic $\gamma$ extending $\sigma$ to the past is defined on $(-\infty,0]$ and, by the rushing property of $b$, $b(\gamma(t))\le b(\sigma(0))-d(\gamma(t),\sigma(0))\le t$ for $t\le 0$. Since $b$ is continuous and finite on $\gamma$ (Prop.\ \ref{cmol}), for every $c<0$ we can find a point of $\gamma$ where $b=c$, hence $B_c\ne \emptyset$.

Suppose, by contradiction, that $B_c$ is not acausal, then there is a causal curve $x:[0,1]\to M$ with distinct endpoints in $B_c$. Its horizontal lift $\tilde x(t)=(x(t),c)$ is a causal curve connecting two points of $H$. Since $H$ is the boundary of the open future set $\tilde D$, every causal curve joining two points of $H$ is contained in $H$: an intermediate point belongs to $\tilde J^+(\overline{\tilde D})\subset \overline{\tilde D}$ \cite[Prop.\ 2.84(2c)]{minguzzi18b}, and if it belonged to $\tilde D$ then the final point would belong to $\tilde J^+(\tilde D)=\tilde D$, as $\tilde D$ is an open future set (Section \ref{secprod}). Thus $\tilde x$ is contained in $H=\textrm{graph}\, b$ (Theorem \ref{xxhg}(iv)). In particular $b=c$ along $x$, which also follows from the rushing property of $b$. Being a causal curve in the achronal set $H$, $\tilde x$ is a lightlike geodesic. Since $\dot z=0$ along it, $x$ is a lightlike geodesic segment. Let $y$ be a future inextendible $b$-calibrated curve starting from $x(1)$ (Proposition \ref{bber}). The lifts of $x$ and $y$ join into a causal curve contained in the achronal set $H$, hence into a lightlike geodesic of $\tilde M$ without corners, along which $\dot z$ is constant, thus zero. Therefore $y$ is lightlike, a contradiction with Theorem \ref{xxhg}(i).
The set $B_c$ is edgeless. Indeed, let $p\in B_c$ and let $U$ be a neighborhood of $p$. If $q_-\in I^-(p,U)$ and $q_+\in I^+(p,U)$, the rushing property gives $b(q_-)\le c-d(q_-,p)<c$ and $b(q_+)\ge c+d(p,q_+)>c$, thus, by the continuity of $b$ as a map to $(-\infty,+\infty]$, every timelike curve in $U$ from $q_-$ to $q_+$ meets $B_c$.

Let $p\in J^-(B_c)$, then there is $q\in B_c$ such that $b(p)+d(p,q)\le b(q)=c$, which implies $b(p)\le c$. Let $x$ be a future inextendible causal curve starting from $p$. By Eq.\ (\ref{cmlr}) and Thm.\ \ref{kvvb}, $b(x(t))\to +\infty$ for $t\to +\infty$, thus $x$ intersects $B_c$ by the continuity of $b$. This proves $b^{-1}((-\infty,c])\subset D^-(B_c)$, and the inclusions $D^-(B_c)\subset J^-(B_c)\subset b^{-1}((-\infty,c])$ are clear.

Let $k> \max\{1,c\}$ be an integer. Since $b_k\le b$, the set $J^-(B_c)= b^{-1}((-\infty,c])$ is contained in $b_k^{-1}((-\infty,c])$, i.e.\ in the set of those $p$ that satisfy $d(p,\Sigma_k)\ge k-c>0$. This set is contained in $J^-(\Sigma_k)$ thus $J^+(S)\cap J^-(B_c)\subset J^+(S)\cap J^-(\Sigma_k)$ where the last set is compact because $S$ and $\Sigma_k$ are Cauchy. Since  $J^+(S)$ is closed and $J^-(B_c)= b^{-1}((-\infty,c])$ is closed because $b$ is continuous, $J^+(S)\cap J^-(B_c)$  is compact.

Let $r\in B_c$ with $r\le x(t)$. Since $b$ is rushing and $x$ is $b$-calibrated and maximizing,
\[
d(r,x(t))\le b(x(t))-b(r)=b(x(t))-b(x(0))=d(x(0),x(t))=\ell(x\vert_{[0,t]})
\]
which states that $x$ is a $B_c$-ray. The last statement follows noting that $b$ goes to infinity along $x$, as along every future inextendible causal curve, which implies that $x$ intersects every $B_d$ for $d\ge c$. Chosen $d$, it is sufficient to redefine the domain so that $x$ starts from $B_d$ and repeat the previous argument.
\end{proof}

As previously mentioned, the set $b_k^{-1}(0)$ is what Galloway and Vega \cite{galloway14b} call a Cauchy prehorosphere; for $k>0$ it is the acausal set $(d^-_{\Sigma_k})^{-1}(k)$, which is a Cauchy hypersurface if $(M,g)$ is also past timelike complete (time dual of Theorem \ref{kvvb}). In the rest of this paragraph we assume that $(M,g)$ is timelike complete, as Galloway and Vega do when they express the Cauchy horosphere through the future sets $I^+(b_k^{-1}(0))$ \cite[Eq.\ (3.12)]{galloway14b}. Suppose $p\in b_{k+1}^{-1}(0)$, and consider a future inextendible timelike geodesic $x$ starting from $p$. The inequality $b_{k+1}\ge b_k$ implies $b_k(p)\le 0$ and as $b_k\to k$ along $x$, $x$ intersects $b_k^{-1}(0)$, which proves $b_{k+1}^{-1}(0)\subset J^-(b_k^{-1}(0))$. As a consequence, $I^+(b_k^{-1}(0))\subset I^+(b_{k+1}^{-1}(0))$, thus providing a non-decreasing sequence of future open sets.

\begin{proposition} \label{prophoro}
Assume that $(M,g)$ is timelike complete. The Cauchy horosphere $\p \cup_k I^+(b_k^{-1}(0))$ coincides with $b^{-1}(0)$.
\end{proposition}

The Cauchy horosphere is thus just one instance, $B_0$, of the level sets $B_c:=b^{-1}(c)$, $c\in \mathbb{R}$. These more general level sets might also be called Cauchy horospheres.

\begin{proof}
Let $F:=\cup_k I^+(b_k^{-1}(0))$.
Let us prove that $b^{-1}(0)$ does not intersect $F$. The function $b_k$ is a time function on $J^-(\Sigma_k)$ (Section \ref{secbuse}) which is $0$ on the Cauchy hypersurface $b_k^{-1}(0)$, thus $b_k>0$ on $I^+(b_k^{-1}(0))$, and so $b\ge b_k>0$ there. We proved that $b^{-1}(0)\subset M\backslash F$.

Let $p\in b^{-1}(0)$ and let $x: [0,a)\to M$, $x(0)=p$, be a proper-time parametrized future inextendible timelike geodesic. We want to prove that $x((0,a))$ is contained in $F$, which would prove that $b^{-1}(0)\subset \overline{F}$, hence $b^{-1}(0)\subset \p F$.
Let $\epsilon\in (0,a)$ be so small that $d(x(0),x(\epsilon))=\epsilon$. For sufficiently large $k$, $x(\epsilon)\in J^-(\Sigma_k)$ and $b_k(p)> b(p)-\epsilon/2=-\epsilon/2$. Thus by (\ref{rushtau}) $b_k(x(\epsilon))\ge b_k(x(0))+d(x(0),x(\epsilon))\ge \epsilon/2$. Moreover, $b_k(x(0))\le b(x(0))=0$, thus $x([0,\epsilon))$ intersects $b_k^{-1}(0)$ and $x(\epsilon)\in F$. As $F$ is a future set and $\epsilon$ can be chosen arbitrarily small, $x((0,a))\subset F$.

Let $p\in \p F$ and let $x: [0,a)\to M$, $x(0)=p$, be a proper-time parametrized future inextendible timelike geodesic. For every future set $F$, by the openness of $I^+$, $I^+(\overline{F})=I^+(F)\subset F$, thus  $x((0,a))\subset F$. For every $\epsilon\in (0,a)$ we have for sufficiently large $k$, $x(\epsilon)\in I^+(b_k^{-1}(0))$ and hence $b(x(\epsilon))\ge b_k(x(\epsilon))\ge 0$. Since $b$ is continuous, $b(p)\ge 0$. It cannot be $b(p)>0$, otherwise $b_k(p)>0$ for large $k$, so that $p\in I^+(b_k^{-1}(0))$ (as $b_k^{-1}(0)$ is a Cauchy hypersurface) and $p$ would be in the open set $F$, not in its boundary. Thus $b(p)=0$.
\end{proof}

\begin{corollary} \label{corhoro}
Assume that $(M,g)$ is timelike complete. For every $c\ge 0$, the level set $B_c$ is the Cauchy horosphere of the compact Cauchy hypersurface $\Sigma_c$ in the sense of Galloway and Vega, namely $B_c=\p \cup_k I^+\big((b^{\Sigma_c}_k)^{-1}(0)\big)$, where $b^{\Sigma_c}_k$ are the approximating functions built from $\Sigma_c$.
\end{corollary}

\begin{proof}
By Remark \ref{cort}, $b^{\Sigma_c}=b-c$, so $B_c=(b^{\Sigma_c})^{-1}(0)$, and Proposition \ref{prophoro} applied to $\Sigma_c$ gives the claim.
\end{proof}

For $c<0$ the identity of Remark \ref{cort} has no analogue, and we do not know whether $B_c$ is the Cauchy horosphere of some compact Cauchy hypersurface. In a product spacetime this is the case, since all level sets are $t$-slices (Example \ref{exstatic}), and a slice is its own Cauchy horosphere. Thus, for $c<0$, the name ``Cauchy horosphere'' for $B_c$ is a convention.

The following is a variant of \cite[Prop.\ 5.2]{galloway14b}.

\begin{proposition} \label{propline}
Assume that $(M,g)$ is future timelike complete. Any past causal $B_c$-ray extends to a timelike line.
If $(M,g)$ is also past timelike complete, this line is complete and thus, under TCC, the existence of a past causal $B_c$-ray implies that the spacetime splits.
\end{proposition}

\begin{proof}
Let $\gamma: (-\infty,0]\to M$, $\gamma(0)\in B_c$, be the past causal ray. From $\gamma(0)$ starts a  timelike $B_c$-ray $\eta: [0,+\infty)\to M$. Let $\alpha: (-\infty,+\infty)\to M$ be obtained joining the two curves. Let $a<0<b$ and let $\beta$ be a causal curve which joins $\alpha(a)$ and $\alpha(b)$. We have $b(\alpha(a))<c$, since $b(\alpha(a))\le c$ by the rushing property, and equality would place two causally related distinct points in the acausal set $B_c$; similarly $b(\alpha(b))>c$. By the continuity of $b$, $\beta$ intersects $B_c$ at some point $\beta(t)$, and so
\begin{align*}
\ell(\alpha\vert_{[a,b]})&=\ell(\alpha\vert_{[a,0]})+\ell(\alpha\vert_{[0,b]})=d(\alpha(a),B_c)+d(B_c,\alpha(b))\\
&\ge  d(\alpha(a),\beta(t))+d(\beta(t),\alpha(b))\ge \ell(\beta).
\end{align*}
By the arbitrariness of $\beta$, $\ell(\alpha\vert_{[a,b]})=d(\alpha(a),\alpha(b))$, which proves that $\alpha$ is maximizing, hence a geodesic line, which is timelike as $\eta$ is timelike.
\end{proof}

We say that $B_c$ is {\em past bounded} if there is a Cauchy hypersurface $\Sigma'$ such that $B_c\subset J^+(\Sigma')$.
The following result for $c=0$ is \cite[Cor.\ 4.5, Lemma 4.6]{galloway14b} (our result  is a generalization due to Remark \ref{cort}).

\begin{proposition} \label{propequiv}
Assume that $(M,g)$ is future timelike complete. For each non-empty level set $B_c$, $c\in \mathbb{R}$, the following are equivalent
\begin{itemize}
\item[(i)] $B_c$ is compact,
\item[(ii)] $B_c$ is past bounded,
\item[(iii)] $B_c$ is Cauchy,
\end{itemize}
and in this case the spacetime contains a timelike line; if $(M,g)$ is also past timelike complete then, under TCC, the spacetime splits.
\end{proposition}

\begin{proof}
By Prop.\ \ref{cllp} $B_c$ is an acausal edgeless set, non-empty by assumption.
(i) $\Rightarrow$ (iii). If $B_c$ is compact then it is a Cauchy hypersurface by \cite[Lemma 3.6]{galloway14b}, see also \cite{budic78}.
(iii) $\Rightarrow$ (ii). If $B_c$ is Cauchy we can take $\Sigma'=B_c$.
(ii) $\Rightarrow$ (i). We know from Prop.\ \ref{cllp} that $J^+(\Sigma')\cap J^-(B_c)$ is compact. From (ii) the closed set $B_c$ is contained in it, so it is compact.
Finally, if (iii) holds there is a past  $B_c$-ray (Theorem \ref{cpgx}, time dual), and by Prop.\ \ref{propline} it extends to a timelike line, which is complete under past timelike completeness, and so, under TCC, the spacetime splits by the Lorentzian splitting theorem.
\end{proof}

The next result clarifies the question of whether the level sets $B_c$ are Cauchy hypersurfaces of the spacetime $(\textrm{dom}\, b, g)$. Note that $\textrm{dom}\, b$ is globally hyperbolic being an open past set, hence causally convex, of a globally hyperbolic spacetime.

\begin{proposition} \label{propdomcauchy}
Assume that $(M,g)$ is future timelike complete and let $\textrm{dom}\, b$ be regarded as a spacetime.
\begin{itemize}
\item[(a)] Along every future inextendible causal curve in $\textrm{dom}\, b$, $b\to +\infty$.
\item[(b)] Every past inextendible causal curve in $\textrm{dom}\, b$ is past inextendible in $M$.
\item[(c)] The following conditions are equivalent:
\begin{enumerate}
\item[(1)] for every $c\in \mathbb{R}$, $B_c$ is a Cauchy hypersurface for $\textrm{dom}\, b$;
\item[(2)] for every past inextendible causal curve $x$ of $M$ starting from a point of $\textrm{dom}\, b$, $b(x(t))\to -\infty$.
\end{enumerate}
\item[(d)] If $x$ is a past inextendible causal curve starting from $\textrm{dom}\, b$ along which $b$ is bounded from below, then $\ell(x)<+\infty$.
\end{itemize}
\end{proposition}

\begin{proof}
(a). Let $x$ be future inextendible in $\textrm{dom}\, b$. If it is future inextendible in $M$ it intersects every $\Sigma_k$ and $b\ge d^+_\Sigma\to +\infty$ along it by (\ref{cmlr}). Otherwise it has a future endpoint $e\in \p\, \textrm{dom}\, b$ and $b(x(t))\to b(e)=+\infty$ by continuity.

(b). If $x$ had a past endpoint $e\in M$, then $e\in J^-(x(0))\subset \textrm{dom}\, b$ as $\textrm{dom}\, b$ is a past set, so $x$ would be extendible in $\textrm{dom}\, b$.

(c). By Proposition \ref{cllp} each $B_c$ is acausal. By (a) and (b) an inextendible causal curve of $\textrm{dom}\, b$ is future inextendible in $\textrm{dom}\, b$ and past inextendible in $M$, and $b$ is continuous, finite and non-decreasing along it, going to $+\infty$ in the future. It intersects every $B_c$ iff $b\to -\infty$ in the past direction.

(d). By the rushing property $b(x(0))-b(x(t))\ge d(x(t),x(0))\ge \ell(x\vert_{[t,0]})$.
\end{proof}

Thus the only obstruction to $B_c$ being Cauchy in $\textrm{dom}\, b$ is the existence of past inextendible causal curves of finite Lorentzian length along which $b$ stays bounded from below.

\subsection{The mean curvature of the horospheres} \label{secmean}

We now show that, under TCC, the Cauchy horospheres are mean convex in a suitable weak sense. This justifies the terminology {\em cosmologically expanding region} used in Section \ref{secstru}, and it provides the barrier used in the rigidity result of Section \ref{secmeanC}.

We recall the d'Alembertian comparison for the Lorentzian distance from a point. Here $n+1=\dim M$, so that $n$ is the dimension of the level sets.

The following estimate is classical: it is the Lorentzian Laplacian comparison theorem for the distance from a point \cite{eschenburg88,galloway89}, see also \cite[Eq.\ (14.29)]{beem96}. We include the short proof for completeness.

\begin{lemma} \label{lemdalem}
Assume TCC and let $\alpha:[0,s]\to M$ be a proper time parametrized maximizing timelike geodesic which extends as a maximizing geodesic beyond $\alpha(s)$. Let $r:=d(\cdot,\alpha(s))$, which is smooth in a neighborhood of $\alpha(0)$. Then
\[
\Box r\, (\alpha(0))\ \ge\ -\frac{n}{s} .
\]
Dually, if $\alpha:[-s,0]\to M$ is a proper time parametrized maximizing timelike geodesic which extends as a maximizing geodesic before $\alpha(-s)$, and $\rho:=d(\alpha(-s),\cdot)$, then $\Box \rho\,(\alpha(0))\ge -n/s$.
\end{lemma}

\begin{proof}
Since $\alpha$ extends as a maximizing geodesic beyond $\alpha(s)$, by Lemma \ref{lemeik} the function $r$ is smooth on a neighborhood of $\alpha([0,s))$, and there it satisfies the eikonal equation $g(\nabla r,\nabla r)=-1$ and $\nabla r$ is the future-directed unit tangent field of the maximizing geodesics ending at $\alpha(s)$; in particular $\nabla r(\alpha(t))=\dot \alpha(t)$ and $r(\alpha(t))=s-t$. The field $\nabla r$ is geodesic and hypersurface-orthogonal, hence irrotational, and its expansion is $\theta:=\textrm{div}\, \nabla r=\Box r$. The Raychaudhuri equation gives, along $\alpha$,
\[
\frac{\dd \theta}{\dd t}=-\frac{\theta^2}{n}-\vert \sigma\vert^2- Ric(\dot \alpha,\dot \alpha)\le -\frac{\theta^2}{n},
\]
where $\sigma$ is the shear of $\nabla r$ and we used the TCC. Here $\theta(t)$ denotes the value of $\theta$ at $\alpha(t)$, which is finite for every $t\in[0,s)$ since $r$ is smooth on a neighborhood of $\alpha([0,s))$. Thus $\theta$ is non-increasing on $[0,s)$.

If $\theta(0)\ge 0$ there is nothing to prove. Otherwise $\theta\le \theta(0)<0$ on $[0,s)$, so that $w:=1/\theta$ is a negative differentiable function on $[0,s)$ with
\[
\frac{\dd w}{\dd t}=-\frac{1}{\theta^2}\frac{\dd \theta}{\dd t}\ge \frac{1}{n} .
\]
Integrating, $0>w(t)\ge w(0)+t/n$ for every $t\in[0,s)$, hence $w(0)<-t/n$ for every $t<s$ and $w(0)\le -s/n$, namely $\Box r(\alpha(0))=\theta(0)\ge -n/s$. The dual statement follows by time reversal.
\end{proof}

We use a version of the notion of support solution introduced by Calabi, adapted to the hyperbolic operator $\Box$ by requiring the support functions to satisfy the eikonal equation, as do the support functions of the Busemann functions in the proof of the Lorentzian splitting theorem \cite{eschenburg88,galloway89,beem96}.

\begin{definition} \label{defsupport}
Let $u:O\to \mathbb{R}$ be continuous on an open set $O\subset M$. We say that $\Box u\le 0$ {\em in the support sense} if for every $p\in O$ and every $\epsilon>0$ there is a function $\varphi$, smooth in a neighborhood of $p$, such that $\varphi(p)=u(p)$, $\varphi\ge u$ near $p$, $g(\nabla\varphi,\nabla\varphi)=-1$ near $p$ with $\nabla\varphi$ past directed, and $\Box \varphi(p)\le \epsilon$. The condition $\Box u\ge 0$ in the support sense is defined dually, with lower support functions and $\Box \varphi(p)\ge -\epsilon$.
\end{definition}

The eikonal normalization of the support functions is essential. Without it the condition would be empty, as $\Box$ is not elliptic: for every smooth function $u$ and every smooth temporal function $t$ defined near $p$, the function $u+A(t-t(p))^2$, $A>0$, is an upper support function for $u$ at $p$, and its d'Alembertian at $p$ is $\Box u(p)+2A\, g(\nabla t,\nabla t)(p)$, which is arbitrarily negative for large $A$. For an eikonal support function $\varphi$, instead, $-\Box\varphi(p)$ is the future mean curvature at $p$ of the level set of $\varphi$ through $p$, so that Definition \ref{defsupport} expresses a condition on the mean curvature of the level sets of $u$.

\begin{theorem} \label{thmmean}
Assume TCC.
\begin{itemize}
\item[(a)] If $(M,g)$ is future timelike complete, then $\Box b^+\le 0$ in the support sense on $\textrm{dom}\, b^+$. The upper support functions can be taken to be the functions $b_{p,s}$ of (\ref{cmby}), namely the normalized Lorentzian distances to the points of a future $b^+$-calibrated ray starting from $p$.
\item[(b)] If $(M,g)$ is past timelike complete, then $\Box b^-\ge 0$ in the support sense on $\textrm{dom}\, b^-$. The lower support functions can be taken to be the time duals of the functions $b_{p,s}$, namely the normalized Lorentzian distances from the points of a past $b^-$-calibrated ray ending at $p$.
\end{itemize}
At the points of differentiability, $-\nabla b^+$ is the initial unit tangent of the unique future $b^+$-calibrated ray, and $-\nabla b^-$ is the final unit tangent of the unique past $b^-$-calibrated ray (Lemma \ref{lemgrad} and its time dual). Differentiability provides these unit tangents, but not their divergence. At the points where $b^\pm$ is twice differentiable, in particular where it is smooth, $-\Box b^\pm$ is the divergence, or expansion, of the field $-\nabla b^\pm$ of the unit tangents of the calibrated rays. In this sense (a) and (b) state, in the support sense, that the semi-flow of the future $b^+$-calibrated rays has non-negative expansion and that the semi-flow of the past $b^-$-calibrated rays has non-positive expansion. These are local statements; the behavior of the volume transported by the calibrated rays is studied in \cite{minguzzi26g}.
\end{theorem}

\begin{proof}
(a). Let $p\in \textrm{dom}\, b^+$ and $\epsilon>0$. Let $\alpha$ be a proper time parametrized $b^+$-calibrated ray starting from $p$; it is a timelike ray by Theorem \ref{xxhg}(i), so for every $s>0$ it is maximizing on $[0,s]$ and extends as a maximizing geodesic beyond $\alpha(s)$. By Theorem \ref{xxhg}(iii) the function
\[
b_{p,s}=b^+(p)+s-d(\cdot,\alpha(s))
\]
is smooth in a neighborhood of $p$, satisfies $b_{p,s}(p)=b^+(p)$ and $b_{p,s}\ge b^+$ there, namely it is a smooth upper support function at $p$, which satisfies the eikonal equation near $p$ with past directed gradient (Lemma \ref{lemeik}). By Lemma \ref{lemdalem}, with $r=d(\cdot,\alpha(s))$,
\[
\Box b_{p,s}(p)=-\Box r(p)\le \frac{n}{s} .
\]
Choosing $s>n/\epsilon$ gives $\Box b_{p,s}(p)\le \epsilon$.

(b). This is the time dual of (a). Let $p\in \textrm{dom}\, b^-$, $\epsilon>0$, and let $\eta$ be a proper time parametrized past $b^-$-calibrated ray ending at $p$, so that $b^-(\eta(-s))=b^-(p)-s$. By the rushing property of $b^-$, the function $\varphi:=b^-(p)-s+d(\eta(-s),\cdot)$ satisfies $\varphi\le b^-$ near $p$, with equality at $p$, and, by the time dual of Lemma \ref{lemeik}, as $\eta$ extends as a maximizing geodesic before $\eta(-s)$, it is smooth near $p$ and satisfies there the eikonal equation with past directed gradient, so it is a smooth eikonal lower support function for $b^-$ at $p$. By the dual statement of Lemma \ref{lemdalem}, with $\rho=d(\eta(-s),\cdot)$, its d'Alembertian at $p$ satisfies $\Box\varphi(p)=\Box\rho(p)\ge -n/s$, which is at least $-\epsilon$ for $s>n/\epsilon$.
\end{proof}

The geometric meaning of Theorem \ref{thmmean} concerns the mean curvature of the horospheres. The future mean curvature of a smooth spacelike hypersurface is the divergence of its future unit normal $N$, extended as the unit tangent of the normal geodesics. If the hypersurface is a level set of a smooth function $u$ with past directed unit timelike gradient, then $N=-\nabla u$ and the future mean curvature is $-\Box u$. For non-smooth hypersurfaces we use the following notion.

\begin{definition} \label{defmeansupp}
Following \cite{galloway97,galloway14b}, a {\em $C^0$ spacelike hypersurface} is a subset of $M$ which is locally acausal and edgeless. Given two such sets $S_1,S_2$ meeting at $q$, we say that $S_2$ lies {\em locally to the past} (resp.\ {\em future}) {\em of $S_1$ near $q$} if there is a neighborhood $U$ of $q$ in which $S_1$ is acausal and edgeless and such that $S_2\cap U\subset J^-(S_1,U)$ (resp.\ $S_2\cap U\subset J^+(S_1,U)$). A $C^0$ spacelike hypersurface $B$ has {\em future mean curvature $H_B\ge a$} (resp.\ {\em $H_B\le a$}) {\em in the support sense} \cite{andersson98b,galloway14b} if for every $q\in B$ and $\epsilon>0$ there is a smooth spacelike hypersurface $S$ through $q$ such that $B$ lies locally to the future (resp.\ past) of $S$ near $q$ and the future mean curvature of $S$ at $q$ is at least $a-\epsilon$ (resp.\ at most $a+\epsilon$).
\end{definition}

At the differentiability points of $b^+$ the horospheres have future unit normal $T=-\nabla b^+$ (Lemma \ref{lemgrad}){, and these points have full $\mathcal{H}^n$ measure on each horosphere, which is achronal (Proposition \ref{cllp}), by Theorem \ref{thmfine}(c);} and, at the points where $b^+$ is twice differentiable, their future mean curvature (expansion) is $\textrm{div}\, T=-\Box b^+$. In general, thanks to the eikonal property of the support functions, Theorem \ref{thmmean} translates into the following result. Its part (a) should be compared with \cite[Thm.\ 4.2]{galloway14b}, where Galloway and Vega use the same support hypersurfaces, the past Lorentzian distance spheres centered on the points of a ray, to prove, under future timelike completeness and TCC, that their past horospheres have mean curvature $\ge 0$ in the support sense, provided the future rays issued from them are timelike. Their Cauchy horospheres are defined under past timelike completeness, and correspond to our level sets $B^+_c$ with $c\ge 0$ (Proposition \ref{prophoro} and Corollary \ref{corhoro}). Corollary \ref{cormeanconv}(a) applies instead to all the level sets $B^+_c$, $c\in\mathbb{R}$, and does not require past timelike completeness. We include the short proof in the language of functions.

\begin{corollary} \label{cormeanconv}
Assume TCC.
\begin{itemize}
\item[(a)] Let $(M,g)$ be future timelike complete and $c\in \mathbb{R}$. Then the horosphere $B^+_c:=(b^+)^{-1}(c)$ has future mean curvature $H_{B^+_c}\ge 0$ in the support sense, the support hypersurfaces being the level sets of the upper support functions $b_{p,s}$ of Theorem \ref{thmmean}(a). Moreover, $\Box b^+\le 0$ at every point at which $b^+$ is twice differentiable.
\item[(b)] Let $(M,g)$ be past timelike complete and $d\in \mathbb{R}$. Then the horosphere $B^-_d:=(b^-)^{-1}(d)$ has future mean curvature $H_{B^-_d}\le 0$ in the support sense, the support hypersurfaces being the level sets of the lower support functions of Theorem \ref{thmmean}(b). Moreover, $\Box b^-\ge 0$ at every point at which $b^-$ is twice differentiable.
\end{itemize}
\end{corollary}

\begin{proof}
(a). By Proposition \ref{cllp}, $B^+_c$ is a $C^0$ spacelike hypersurface. Let $p\in B^+_c$, $\epsilon>0$, let $\alpha$ be a proper time parametrized $b^+$-calibrated ray starting from $p$, let $s>n/\epsilon$ and let $r:=d(\cdot,\alpha(s))$, so that $b_{p,s}=b^+(p)+s-r=c+s-r$. By Lemma \ref{lemeik}, $r$ is smooth near $p$, and $-\nabla b_{p,s}=\nabla r$ is the future directed unit tangent field of the maximizing geodesics ending at $\alpha(s)$. Take $S$ to be the level set of $b_{p,s}$ through $p$, namely the past Lorentzian distance sphere $\{r=s\}$; it is smooth and spacelike near $p$ since $\nabla b_{p,s}$ is timelike, and its future mean curvature at $p$ is $-\Box b_{p,s}(p)\ge -n/s>-\epsilon$ by Lemma \ref{lemdalem}.

Let $\chi$ be the flow of the vector field $\nabla r=-\nabla b_{p,s}$. Its integral curves are the maximizing geodesics ending at $\alpha(s)$, parametrized by proper time, and along each of them $r$ decreases by the elapsed proper time, as they are maximizing. Thus $b_{p,s}(\chi_t(x))=b_{p,s}(x)+t$; equivalently, $\frac{\dd}{\dd t}b_{p,s}(\chi_t(x))=-g(\nabla b_{p,s},\nabla b_{p,s})=1$ by the eikonal equation. Let $U=\chi((-\delta,\delta)\times S')$, where $S'\subset S$ is a small neighborhood of $p$ and $\delta>0$ is so small that $b_{p,s}$ is smooth and $b^+\le b_{p,s}$ on $U$ (Theorem \ref{xxhg}(iii)). In $U$ the function $b_{p,s}$ is strictly increasing along causal curves, so $S\cap U=\chi(\{0\}\times S')$ is acausal and edgeless in $U$. Let $x\in B^+_c\cap U$, and write $x=\chi_t(z)$ with $z\in S'$ and $t\in(-\delta,\delta)$. Since $z\in S$, $b_{p,s}(z)=b_{p,s}(p)=c$, hence $b_{p,s}(x)=c+t$, namely $t=b_{p,s}(x)-c$. As $b_{p,s}(x)\ge b^+(x)=c$, we get $t\ge 0$, and the flow line from $z$ to $x$ shows $x\in J^+(S\cap U,U)$. Thus $B^+_c$ lies locally to the future of $S$ near $p$, and $B^+_c$ has future mean curvature $H_{B^+_c}\ge 0$ in the support sense.

For the second statement, let $q$ be a point at which $b^+$ is twice differentiable, let $\alpha$ be a proper time parametrized $b^+$-calibrated ray starting from $q$ and let $T=\dot \alpha(0)$. As $b_{q,s}-b^+\ge 0$ vanishes at $q$ and both functions are differentiable there, $\dd b^+_q=\dd (b_{q,s})_q$ and $\textrm{Hess}(b_{q,s}-b^+)(q)$ is positive semi-definite. Moreover $b^+(\alpha(t))=b^+(q)+t$ by calibration and $\alpha$ is a geodesic, so $\textrm{Hess}\, b^+(T,T)=\frac{\dd^2}{\dd t^2}b^+(\alpha(t))\vert_{t=0}=0$, and similarly $\textrm{Hess}\, b_{q,s}(T,T)=0$, as $b_{q,s}$ is affine along $\alpha$. Therefore the Lorentzian trace of $\textrm{Hess}(b_{q,s}-b^+)(q)$ coincides with its trace on the orthogonal complement of $T$, which is non-negative, and $\Box b^+(q)\le \Box b_{q,s}(q)\le n/s$ for every $s>0$.

(b) is the time dual of (a): the level set $S$ through $p\in B^-_d$ of the lower support function $\varphi$ of the proof of Theorem \ref{thmmean}(b) is a future Lorentzian distance sphere with future mean curvature $-\Box\varphi(p)\le n/s$ at $p$, and $B^-_d$ lies locally to the past of $S$ near $p$, since $\varphi\le b^-$.
\end{proof}

Observe that in Minkowski spacetime, or more generally in a product $-\dd t^2+h$, the horospheres are the slices $t=\textrm{const}$ and the mean curvature vanishes, so the inequality is sharp.

Mean convexity becomes a rigidity tool through the geometric maximum principle. For smooth hypersurfaces it is classical; Eschenburg gave an intrinsic proof which also holds for spacelike hypersurfaces in Lorentzian manifolds, and a Riemannian version in which one of the two hypersurfaces satisfies the mean curvature inequality only in the support sense \cite[Thms.\ 1, 1b]{eschenburg89}. We shall need the following Lorentzian rough-smooth version, due to Galloway \cite[Thm.\ 3.4]{galloway97}, of which Andersson, Galloway and Howard proved a more general version, in which both hypersurfaces may be rough \cite{andersson98b}. We use the terminology of Definition \ref{defmeansupp} and state its time dual, which will be used in Step 2 of the proof of Theorem \ref{thmhawk} in Section \ref{secmeanC}. It is a local result, which holds in every spacetime: it requires neither an energy condition nor a completeness assumption.

\begin{theorem}[Galloway] \label{thmagh}
Let $S_1$ be a smooth spacelike hypersurface and let $S_2$ be a $C^0$ spacelike hypersurface, meeting at a point $q$, such that $S_2$ lies locally to the past of $S_1$ near $q$. Suppose that, for some constant $a$, the future mean curvatures satisfy
\[
H_{S_1}\le a\le H_{S_2},
\]
where the second inequality is meant in the support sense. Then $S_1$ and $S_2$ agree in a neighborhood of $q$.
\end{theorem}

For smooth hypersurfaces this is the classical consequence of the Hopf maximum principle for the mean curvature operator, which is elliptic on spacelike graphs; the rough version combines it with Calabi's notion of support functions \cite{galloway97,andersson98b}.

\section{Future and past Cauchy Busemann functions and splitting} \label{secsplit}

If $(M,g)$ is future timelike complete, let $b^+: M\to (-\infty,+\infty]$ be the function $b$ constructed so far from the Cauchy hypersurface $\Sigma$. Similarly, if $(M,g)$ is past timelike complete, let $b^-: M\to [-\infty,+\infty)$ be its time dual, obtained from $\tilde \Sigma^-=\{(p,r): p\in \Sigma, r\ge 0\}$, taking the chronological past in $\tilde M$ and hence the  horizon of its future domain of dependence, expressed as the graph of $b^-$. Explicitly,
\begin{align*}
b^+_k(p)&=k-d(p,\Sigma_k), & b^+&=\sup_k b^+_k=\lim_k b^+_k, \\
b^-_k(p)&=-k+d(\Sigma_{-k},p), & b^-&=\inf_k b^-_k=\lim_k b^-_k.
\end{align*}
All the previous results hold, after time duality, for $b^-$, with future timelike completeness replaced by past timelike completeness. In particular $b^-$ is rushing, $b^-\le d_\Sigma$, $\textrm{dom}\, b^-$ is an open future set, and for $p\in \textrm{dom}\, b^-$ there are past inextendible timelike $b^-$-calibrated rays ending at $p$. We write $B^\pm_c:=(b^\pm)^{-1}(c)$.

Observe that $(b^+_k)^{-1}(0)\subset J^-(\Sigma)$ for every $k$ and so $(b^+)^{-1}(0)\subset J^-(\Sigma)$. Similarly, $(b^-)^{-1}(0)\subset J^+(\Sigma)$. Thus $(b^+)^{-1}(0)$ might not intersect $(b^-)^{-1}(0)$; however, if $\textrm{dom}\, b^+\cap \textrm{dom}\, b^- \ne \emptyset$ we can find $B^+_c$ and $B^-_d$ that intersect.

Recall that $d_\Sigma$ denotes the signed distance from $\Sigma$ (Section \ref{secdist}), so that $\Sigma_s=d_\Sigma^{-1}(s)$; if $(M,g)$ is timelike complete, it is a Cauchy time function by Theorem \ref{kvvb}.

\begin{lemma} \label{lemsandwich}
The function $d_\Sigma$ is rushing on $M$. If $(M,g)$ is future timelike complete then $d_\Sigma\le b^+$, and if it is past timelike complete then $b^-\le d_\Sigma$.
\end{lemma}

\begin{proof}
The first statement is Proposition \ref{propdrush}, which uses no completeness assumption.
By (\ref{cmlr}), $b^+\ge d_\Sigma$. The inequality $b^-\le d_\Sigma$ is the time dual.
\end{proof}

The following theorem is the main splitting statement of this work: the splitting criteria of this section and of Section \ref{secstru}, starting with Theorem \ref{thmscal}, are derived from it, with the exception of the one-sided criterion of Theorem \ref{thmonesided}, which adapts its mechanism to spacetimes which are only future timelike complete.

\begin{theorem} \label{notr}
Assume that $(M,g)$ is timelike complete. We have $b^+\ge b^-$. If $\textrm{dom}\, b^+\cap \textrm{dom}\, b^-\ne \emptyset$ then the function $G:=b^+-b^-:M\to [0,+\infty]$ attains a finite minimum at some point of $\Sigma$, the minimum over $\Sigma$ is a global minimum, and through every global minimum point $p$ passes a complete timelike line which is both $b^+$-calibrated and $b^-$-calibrated. Thus, under TCC, if $\textrm{dom}\, b^+\cap \textrm{dom}\, b^-\ne \emptyset$ the spacetime splits.
\end{theorem}

\begin{proof}
The inequality $b^+\ge b^-$ follows from Lemma \ref{lemsandwich}. The function $G$ is well defined, since $b^+>-\infty$ and $b^-<+\infty$, and it is continuous as a map to $[0,+\infty]$ since $b^\pm$ are continuous (Theorem \ref{xxhg}).

Let $q\in \textrm{dom}\, b^+\cap \textrm{dom}\, b^-$. Let $\gamma:[0,+\infty)\to M$ be a proper time parametrized future $b^+$-calibrated ray starting from $q$ and let $\eta:(-\infty,0]\to M$ be a proper time parametrized past $b^-$-calibrated ray ending at $q$; they are timelike by Theorem \ref{xxhg}(i) and its dual. As $\textrm{dom}\, b^-$ is a future set and $\textrm{dom}\, b^+$ is a past set, both curves are contained in $\textrm{dom}\, b^+\cap \textrm{dom}\, b^-$. For $t\ge 0$, by calibration and by the rushing property of $b^-$,
\[
b^+(\gamma(t))=b^+(q)+t, \qquad b^-(\gamma(t))\ge b^-(q)+t, \qquad \textrm{so} \quad G(\gamma(t))\le G(q),
\]
and similarly $G(\eta(t))\le G(q)$ for $t\le 0$. The inextendible curve obtained joining $\eta$ and $\gamma$ intersects $\Sigma$, thus $\min_\Sigma G\le G(q)$, where the minimum over the compact set $\Sigma$ exists by continuity. Since $G=+\infty$ outside $\textrm{dom}\, b^+\cap \textrm{dom}\, b^-$, the minimum of $G$ over $\Sigma$ is a global minimum, and it is finite.

Let $p$ be a global minimum point, and let $\gamma$, $\eta$ be as above with $q=p$. By minimality the inequalities in the previous display are equalities, thus $b^-(\gamma(t))=b^-(p)+t$ for $t\ge 0$ and, analogously, $b^+(\eta(t))=b^+(p)+t$ for $t\le 0$. Let $\alpha:\mathbb{R}\to M$ be the join of $\eta$ and $\gamma$. For $a<b$ we have $b^+(\alpha(b))-b^+(\alpha(a))=b-a=\ell(\alpha\vert_{[a,b]})$, so by Proposition \ref{jjtr} $\alpha$ is maximizing and $b^+$-calibrated. It is then a geodesic without corners, namely a timelike line, and by the same argument it is $b^-$-calibrated. Under TCC the spacetime splits by the Lorentzian splitting theorem.

\end{proof}

Note that if $(b^+-b^-)(p)$ is finite for some $p$ then $p\in \textrm{dom}\, b^+\cap \textrm{dom}\, b^-$. As $b^+_k-b^-_k$ is non-decreasing in $k$ and converges to $b^+-b^-$, this finiteness condition is equivalent to asking that $(b^+_k-b^-_k)(p)=2k-d(\Sigma_{-k},p)-d(p,\Sigma_k)$ is  bounded from above. Remarkably, the point $p$ can be removed from this condition.

\begin{theorem} \label{thmscal}
Assume that $(M,g)$ is timelike complete. The sequence $\Delta_k:=2k-d(\Sigma_{-k},\Sigma_k)$, $k\in \mathbb{N}$, is non-negative and non-decreasing, and
\begin{equation} \label{deltalim}
\lim_{k\to \infty} \Delta_k=\min_{\Sigma} (b^+-b^-)=\inf_M (b^+-b^-) \in [0,+\infty].
\end{equation}
In particular, the following conditions are equivalent:
\begin{itemize}
\item[(i)] $\sup_k \Delta_k<+\infty$;
\item[(ii)] $\textrm{dom}\, b^+\cap \textrm{dom}\, b^-\ne \emptyset$;
\item[(iii)] $b^+-b^-$ has a finite minimum.
\end{itemize}
Each of them implies the existence of a timelike line and, under TCC, they are equivalent to the splitting of the spacetime.
\end{theorem}

By Proposition \ref{nner}, $d(\Sigma_{-k},\Sigma_k)\le 2k$, and $\Delta_k$ is the defect in this inequality. Under TCC, Theorem \ref{thmscal} states that the spacetime splits if and only if the defect stays bounded. Thus Bartnik's conjecture can be restated as follows: in a timelike complete spacetime satisfying TCC the defect $\Delta_k$ is bounded. Equivalently, a counterexample is a timelike complete spacetime satisfying TCC in which $\Delta_k\to+\infty$. Note that a sublinear divergence $\Delta_k=o(k)$ would be compatible with a counterexample.

\begin{proof}
Every causal curve from $\Sigma_{-k}$ to $\Sigma_k$ crosses $\Sigma$, and maximizing curves from $\Sigma_{-k}$ to $p\in \Sigma$ and from $p$ to $\Sigma_k$ can be concatenated. Thus, by compactness and continuity,
\[
d(\Sigma_{-k},\Sigma_k)=\max_{p\in \Sigma}\,[d(\Sigma_{-k},p)+d(p,\Sigma_k)],
\]
namely
\[
\Delta_k=\min_\Sigma f_k, \qquad f_k:=b^+_k-b^-_k .
\]
On $\Sigma$ the functions $f_k$ are continuous, non-negative (as $d(\Sigma_{-k},p)\le k$ and $d(p,\Sigma_k)\le k$) and non-decreasing in $k$, with limit $G=b^+-b^-$. Thus $\Delta_k$ is non-negative and non-decreasing. Let $L=\lim_k \Delta_k$. Clearly $L\le \min_\Sigma G$. If $L<+\infty$, let $p_k\in \Sigma$ be minimum points of $f_k$ and, passing to a subsequence, $p_k\to p\in \Sigma$. For $j\le k$, $f_j(p_k)\le f_k(p_k)\le L$, thus $f_j(p)\le L$ and, letting $j\to \infty$, $G(p)\le L$. This proves the first equality in (\ref{deltalim}), also in the case $L=+\infty$. The second equality follows from Theorem \ref{notr} if $\textrm{dom}\, b^+\cap \textrm{dom}\, b^-\ne \emptyset$, and it is trivial otherwise, since then $b^+-b^-=+\infty$ on $M$. The equivalences follow, and, by Theorem \ref{notr} and the Lorentzian splitting theorem, under TCC they imply the splitting. Conversely, suppose that $(M,g)$ splits as $(\mathbb{R}\times S,-\dd t^2+h)$. As $t$ is rushing, namely $d(p,q)\le t(q)-t(p)$ for $p\le q$, we have $t\ge k+\min_\Sigma t$ on $\Sigma_k$. For $p\in M$ and $k$ large, the vertical line through $p$ meets $\Sigma_k$ at a point $q$ with $d(p,q)=t(q)-t(p)\ge k+\min_\Sigma t-t(p)$, so $d(p,\Sigma_k)\ge k+\min_\Sigma t-t(p)$ and $b^+(p)\le t(p)-\min_\Sigma t<+\infty$. Dually $b^-\ge t-\max_\Sigma t>-\infty$, so $\textrm{dom}\, b^+=\textrm{dom}\, b^-=M$ and (ii) holds.
\end{proof}

The proof of Theorem \ref{notr} uses past timelike completeness only to construct $b^-$ and its past calibrated rays. In that proof $b^-$ plays two roles: it is a rushing function, and it provides past calibrated rays from every point of its domain. The first role can be taken by the signed distance $d_\Sigma$, which is rushing without any completeness assumption (Proposition \ref{propdrush}) and calibrated along past $\Sigma$-rays, while the second can be replaced by a limit argument based on the compactness of $\Sigma$. This leads to a one-sided criterion, in which $G=b^+-b^-$ is replaced by
\[
E:=b^+-d_\Sigma:M\to [0,+\infty],
\]
which is non-negative by Lemma \ref{lemsandwich} under future timelike completeness alone. Under timelike completeness $b^-\le d_\Sigma$, thus $E\le G$.

\begin{theorem}[One-sided criterion] \label{thmonesided}
Assume that $(M,g)$ is future timelike complete. Suppose that there is a sequence $p_n\in J^-(\Sigma)$ without accumulation points such that $C:=\sup_n E(p_n)<+\infty$; equivalently, suppose that $E$ is bounded along some past inextendible causal curve. Then there is a $b^+$-calibrated timelike line $\alpha$, namely an inextendible timelike geodesic maximizing between any two of its points, which is future complete, passes through a point $x\in \Sigma$ with $b^+(x)\le C$, and satisfies $E\le C$ on $\alpha\cap J^-(\Sigma)$. If $(M,g)$ is also past timelike complete, then $\alpha$ is complete, $I^+(\alpha)\subset \textrm{dom}\, b^-$, and the existence of such a sequence is equivalent to $\textrm{dom}\, b^+\cap \textrm{dom}\, b^-\ne \emptyset$.
\end{theorem}

\begin{proof}
The points $p_n$ belong to $\textrm{dom}\, b^+$, as $E(p_n)<+\infty$, and, discarding finitely many of them, to $I^-(\Sigma)$, as $\Sigma$ is compact. Let $\gamma_n$ be a future inextendible $b^+$-calibrated causal curve starting from $p_n$ (Proposition \ref{bber}); it is a timelike ray by Theorem \ref{xxhg}(i), and it meets the Cauchy hypersurface $\Sigma$ at a point $x_n$. For $y\in \gamma_n$ between $p_n$ and $x_n$, calibration gives $b^+(y)=b^+(p_n)+d(p_n,y)$, and the reverse triangle inequality gives $d(p_n,\Sigma)\ge d(p_n,y)+d(y,\Sigma)$. Since $d_\Sigma=-d(\cdot,\Sigma)$ on $J^-(\Sigma)$,
\[
E(y)=b^+(y)+d(y,\Sigma)\le b^+(p_n)+d(p_n,\Sigma)=E(p_n)\le C.
\]
In particular $b^+(x_n)=E(x_n)\le C$. Passing to a subsequence, $x_n\to x\in \Sigma$, and since $b^+$ is continuous as a map to $(-\infty,+\infty]$ (Theorem \ref{xxhg}), $b^+(x)\le C$, so $x\in \textrm{dom}\, b^+$.

Parametrize $\gamma_n$ with respect to the arc-length of a complete Riemannian metric $h$, with $\gamma_n(0)=x_n$, so that $\gamma_n$ is defined on $[-a_n,+\infty)$, with $\gamma_n(-a_n)=p_n$. Since the points $p_n$ escape every compact set, $a_n\to +\infty$, and by the limit curve theorem \cite[Thm.\ 2.51, Rem.\ 2.52]{minguzzi18b} a subsequence converges $h$-uniformly on compact subsets to an inextendible causal curve $\alpha:\mathbb{R}\to M$ with $\alpha(0)=x$. Let $a<b$. For large $n$, calibration gives $b^+(\gamma_n(b))-b^+(\gamma_n(a))=d(\gamma_n(a),\gamma_n(b))=\ell(\gamma_n\vert_{[a,b]})$, and $b^+(\gamma_n(b))\le C+d(x_n,\gamma_n(b))$ for $b\ge 0$ is bounded, since $d$ is continuous on the globally hyperbolic spacetime $(M,g)$. By the continuity of $b^+$ and $d$, $\alpha(a),\alpha(b)\in \textrm{dom}\, b^+$ and $b^+(\alpha(b))-b^+(\alpha(a))=d(\alpha(a),\alpha(b))$, while the upper semicontinuity of the Lorentzian length \cite{minguzzi18b} gives $\ell(\alpha\vert_{[a,b]})\ge \limsup_n \ell(\gamma_n\vert_{[a,b]})=d(\alpha(a),\alpha(b))$. Thus $\alpha$ is maximizing between any two of its points and $b^+$-calibrated (Proposition \ref{jjtr}), hence a causal geodesic without corners. Its future half is a future inextendible $b^+$-calibrated causal curve starting from $x\in \textrm{dom}\, b^+$, hence a future complete timelike ray by Theorem \ref{xxhg}(i). Therefore $\alpha$ is a timelike geodesic, and a timelike line. Since $\alpha$ crosses $\Sigma$ only at $x$, $\alpha\cap J^-(\Sigma)=\alpha((-\infty,0])$, and every $\alpha(s)$, $s\le 0$, is the limit of the points $\gamma_n(s)$, which lie between $p_n$ and $x_n$ for large $n$. As $E$ is continuous on $\textrm{dom}\, b^+$ (Theorem \ref{thmcont}), $E\le C$ on $\alpha\cap J^-(\Sigma)$. The two formulations of the hypothesis are equivalent: a past inextendible causal curve has no accumulation points in its past direction, as $(M,g)$ is strongly causal, and conversely the past half of $\alpha$ is a past inextendible causal curve along which $E$ is bounded.

Assume now that $(M,g)$ is also past timelike complete, and parametrize $\alpha$ with proper time, $\alpha(0)=x$; then $\alpha$ is defined on $\mathbb{R}$. For $t\le 0$, $b^+(\alpha(t))=b^+(x)+t$ and $b^+(x)\ge d_\Sigma(x)=0$, thus $E(\alpha(t))\le C$ gives $d_\Sigma(\alpha(t))\ge t-C$. For $k\in \mathbb{N}$ the curve $\alpha$ meets the Cauchy hypersurface $\Sigma_{-k}$ (Theorem \ref{kvvb}) at some $\alpha(t_k)$, and $-k=d_\Sigma(\alpha(t_k))\ge t_k-C$. Let $q\in I^+(\alpha)$ and let $t_*$ be such that $\alpha(t_*)\ll q$. For $k$ large, $t_k<t_*$, and since $\alpha$ is maximizing,
\[
b^-_k(q)=-k+d(\Sigma_{-k},q)\ge -k+d(\alpha(t_k),\alpha(t_*))=-k+t_*-t_k\ge t_*-C,
\]
thus $b^-(q)=\lim_k b^-_k(q)\ge t_*-C>-\infty$ and $I^+(\alpha)\subset \textrm{dom}\, b^-$. In particular $\alpha(1)\in \textrm{dom}\, b^+\cap \textrm{dom}\, b^-$. Conversely, if $p\in \textrm{dom}\, b^+\cap \textrm{dom}\, b^-$, let $\eta:(-\infty,0]\to M$ be a past $b^-$-calibrated ray ending at $p$. By the proof of Theorem \ref{notr}, $\eta\subset \textrm{dom}\, b^+\cap \textrm{dom}\, b^-$ and $G(\eta(t))\le G(p)$, hence $E\le G\le G(p)$ along $\eta$, and the points $\eta(-n)$, $n$ large, provide the desired sequence.
\end{proof}

Thus, under timelike completeness, the condition of Theorem \ref{thmonesided} is one more equivalent form of the conditions of Theorem \ref{thmscal}, which only involves the future Cauchy Busemann function and the distance from $\Sigma$. Under future timelike completeness alone the line $\alpha$ need not be past complete, so that the Lorentzian splitting theorem cannot be applied: in the warped products of Example \ref{exwarp} defined over a time interval bounded from below, for instance in the flat spacetimes $\big((0,+\infty)\times S,\,-\dd t^2+t^2h\big)$, with $(S,h)$ a compact hyperbolic manifold, which are the quotients of the interior of the future light cone of Minkowski spacetime by cocompact discrete groups of Lorentz transformations (compactified Milne universes), the same argument which gives $b=t-t_0$ gives $d_\Sigma=t-t_0$, hence $E\equiv 0$, and the timelike lines provided by Theorem \ref{thmonesided} are the comoving observers, which are past incomplete. Theorem \ref{thmonesided} is thus an existence result for future complete maximizing observers which holds also in the presence of an initial singularity.

\begin{corollary}[Bounded past depth] \label{corbigbang}
Assume that $(M,g)$ is future timelike complete. If $d_\Sigma$ is bounded from below along some past inextendible causal curve starting from a point of $\textrm{dom}\, b^+$, then there is a $b^+$-calibrated timelike line. In particular, if $(M,g)$ has regular cosmological time, then for every compact Cauchy hypersurface $\Sigma$ there is a $b^+$-calibrated timelike line, and there is a $b^{\textnormal{\textsc{m}}}$-calibrated timelike line.
\end{corollary}

\begin{proof}
Let $\gamma$ be such a curve, starting from $p\in \textrm{dom}\, b^+$. As $\textrm{dom}\, b^+$ is a past set, $\gamma\subset \textrm{dom}\, b^+$, and by the rushing property $b^+\le b^+(p)$ along $\gamma$. Thus $E\le b^+(p)-\inf_\gamma d_\Sigma<+\infty$ along $\gamma$, and Theorem \ref{thmonesided} applies. If the cosmological time $\tau$ is regular, then for $y\in J^-(\Sigma)$ there is $r\in \Sigma$ with $d(y,\Sigma)=d(y,r)\le \tau(r)\le \max_\Sigma \tau$, the maximum being finite as $\tau$ is continuous \cite[Thm.\ 1.2]{andersson98}. Hence $d_\Sigma\ge -\max_\Sigma\tau$ on $M$, and the first statement applies to every past inextendible causal curve starting from a point of the non-empty set $\textrm{dom}\, b^+$. For $b^{\textnormal{\textsc{m}}}$ apply this to $\Sigma=T_a$, since $b^{\textnormal{\textsc{m}}}=b^{T_a}+a$ (Proposition \ref{propcosm}(b)).
\end{proof}

Thus every future timelike complete cosmological model with compact Cauchy hypersurfaces and regular cosmological time contains an eternal maximally aging observer: a timelike geodesic which emanates from the initial singularity, is maximizing between any two of its events, and is $b^{\textnormal{\textsc{m}}}$-calibrated, so that, by Proposition \ref{propcosm}(d), from each of its events it realizes the minimal asymptotic delay with respect to the cosmological time. Its past half has proper length at most the cosmological time of its points. In the warped products of Example \ref{exwarp} all the comoving observers are of this type, but in general these observers are expected to be exceptional, and to coexist with the singularities of $b^{\textnormal{\textsc{m}}}$. The contrapositive is also of interest: if $(M,g)$ is future timelike complete and contains no $b^+$-calibrated timelike line, then $d(\cdot,\Sigma)\to+\infty$ along every past inextendible causal curve contained in $\textrm{dom}\, b^+$, namely, the future cosmological development has an infinitely deep past.

\begin{example} \label{exstatic}
Let $g=-\dd t^2+h$ on $\mathbb{R}\times S$ with $(S,h)$ compact and connected, a timelike complete spacetime, and let $\Sigma$ be any compact Cauchy hypersurface. Since $\Sigma$ meets every timelike line $t\mapsto (t,x)$ exactly once, it is the graph $t=f(x)$ of a continuous function $f$ on $S$. Let $m=\min f$, $m'=\max f$, and $D$ the diameter of $(S,h)$. Recall that, whenever $t_2-t_1\ge \textrm{dist}_h(x,y)$,
\begin{equation} \label{eqstatd}
d((t_1,y),(t_2,x))=\sqrt{(t_2-t_1)^2-\textrm{dist}_h(x,y)^2} .
\end{equation}
For large $t$, choosing a point of $\Sigma$ where $f=m$,
\[
\sqrt{(t-m)^2-D^2}\le d_\Sigma(t,x)\le t-m,
\]
thus $\Sigma_k$ lies between the slices $t=k+m$ and $t=k+m+O(1/k)$, and $d((t,x),\Sigma_k)=k+m-t+o(1)$. Hence
\[
b^+=t-m, \qquad b^-=t-m', \qquad \lim_k \Delta_k=\max f-\min f .
\]
In particular $\textrm{dom}\, b^\pm=M$, the functions $b^\pm$ are smooth with vanishing Hessian (as $\p_t$ is parallel), $b^+-b^-\equiv m'-m$, and every vertical line $t\mapsto (t,x)$ is a timelike line calibrated for both $b^+$ and $b^-$. Since $\lim_k\Delta_k=\max f-\min f$ vanishes only if $\Sigma$ is a slice $\{t=\textrm{const}\}$, one should not expect that condition (i) of Theorem \ref{thmscal} could be replaced by the stronger condition $\Delta_k\to 0$, namely by vanishing asymptotic defect. That the Cauchy horospheres of a compact Cauchy hypersurface in such a product are compact Cauchy hypersurfaces, in fact $t$-slices, was observed by Galloway and Vega \cite[Lemma 5.3 and proof of Thm.\ 5.4]{galloway14b}. The estimate on $d_\Sigma$ used above is essentially the one in the proof of \cite[Prop.\ 4.3]{galloway26}.
\end{example}

\begin{corollary}[Clock defects] \label{cordefect}
Assume that $(M,g)$ is timelike complete. For a complete proper time parametrized future directed timelike geodesic $\gamma:\mathbb{R}\to M$ the limits
\begin{align*}
A^+(\gamma)&=\lim_{t\to +\infty}[d_\Sigma(\gamma(t))-t]\in (-\infty,+\infty], \\
A^-(\gamma)&=\lim_{t\to -\infty}[d_\Sigma(\gamma(t))-t]\in [-\infty,+\infty)
\end{align*}
exist, and $b^-(\gamma(0))\ge A^-(\gamma)$, $b^+(\gamma(0))\le A^+(\gamma)$, with equality if $\gamma\vert_{[0,+\infty)}$ is $b^+$-calibrated (resp.\ $\gamma\vert_{(-\infty,0]}$ is $b^-$-calibrated). If the spacetime satisfies TCC and does not split, then every complete timelike geodesic satisfies $A^+(\gamma)=+\infty$ or $A^-(\gamma)=-\infty$.
\end{corollary}

\begin{proof}
Since $d_\Sigma$ is rushing, $t\mapsto d_\Sigma(\gamma(t))-t$ is non-decreasing, so the limits exist. For large $t$, $\gamma(t)\in J^+(\Sigma)$ and $d_\Sigma(\gamma(t))-t$ is the function $F_x(t)$ of Theorem \ref{thmintr} with $x=\gamma\vert_{[0,+\infty)}$, thus $A^+(\gamma)=L(x)\ge b^+(\gamma(0))$, with equality for calibrated rays. The past statement is the time dual. If both $A^\pm(\gamma)$ are finite, $\gamma(0)\in \textrm{dom}\, b^+\cap \textrm{dom}\, b^-$ and the spacetime splits by Theorem \ref{notr}.
\end{proof}

Thus, in a counterexample to Bartnik's conjecture, every inertial observer measures an unbounded clock defect, with respect to the distance from $\Sigma$, either in the far future or in the far past. In particular, a future $b^+$-calibrated ray, once extended to the past, must have infinite past defect.

The past clock defect also enters the one-sided criterion. Assume that $(M,g)$ is only future timelike complete, and let $\gamma:(a,0]\to M$, $a\in[-\infty,0)$, be a past inextendible proper time parametrized timelike geodesic with $\gamma(0)\in \textrm{dom}\, b^+$. As above, $t\mapsto d_\Sigma(\gamma(t))-t$ is non-decreasing, so $A^-(\gamma):=\lim_{t\to a^+}[d_\Sigma(\gamma(t))-t]\in[-\infty,+\infty)$ exists, and by the rushing property $b^+(\gamma(t))\le b^+(\gamma(0))+t$. Thus
\[
E(\gamma(t))\le b^+(\gamma(0))-[d_\Sigma(\gamma(t))-t]\le b^+(\gamma(0))-A^-(\gamma),
\]
and if $A^-(\gamma)>-\infty$ Theorem \ref{thmonesided} provides a $b^+$-calibrated timelike line. Hence, under future timelike completeness alone, an observer starting in the future cosmological development with finite past clock defect forces the existence of a $b^+$-calibrated timelike line, which is complete, and hence produces the splitting under TCC, if the spacetime is also past timelike complete.

\begin{corollary}[Divergence of the past defect] \label{corpastdiv}
Assume that $(M,g)$ is future timelike complete and contains no $b^+$-calibrated timelike line, for instance no timelike line. Then for every $C\in\mathbb{R}$ the set $\{p\in J^-(\Sigma): b^+(p)-d_\Sigma(p)\le C\}$ is compact; namely, $b^+-d_\Sigma\to+\infty$ at past infinity. Dually, if $(M,g)$ is past timelike complete and contains no $b^-$-calibrated timelike line, then the sets $\{p\in J^+(\Sigma): d_\Sigma(p)-b^-(p)\le C\}$ are compact. In particular, both conclusions hold in a counterexample to Bartnik's conjecture.
\end{corollary}

\begin{proof}
The function $E=b^+-d_\Sigma:M\to[0,+\infty]$ is continuous, as $b^+$ is continuous as a map to $(-\infty,+\infty]$ (Theorem \ref{xxhg}) and $d_\Sigma$ is continuous (Theorem \ref{thmcont}). Thus the set in question is closed and, if it were not compact, it would contain a sequence without accumulation points, and Theorem \ref{thmonesided} would provide a $b^+$-calibrated timelike line. The second statement is the time dual. In a counterexample to Bartnik's conjecture every timelike line is complete, so there are no timelike lines by the Lorentzian splitting theorem.
\end{proof}

This is a uniform version of the last statement of Corollary \ref{cordefect}: in a counterexample, not only every inertial observer, but every sequence of events escaping to the past measures an unbounded defect of $b^+$ with respect to the distance from $\Sigma$.

\begin{corollary} \label{corint}
Assume TCC and that $(M,g)$ is timelike complete. If $\textrm{dom}\, b^+\cap \textrm{dom}\, b^-\ne \emptyset$ then $(M,g)$ is isometric to a product $(\mathbb{R}\times S,-\dd t^2+h)$ with $(S,h)$ compact, $\Sigma$ is a graph $t=f(x)$, and
\[
b^+=t-\min f, \qquad b^-=t-\max f .
\]
In particular $\textrm{dom}\, b^\pm=M$, the functions $b^\pm$ are smooth with vanishing Hessian, $b^+-b^-$ is constant, and through every point there passes a timelike line calibrated for both $b^+$ and $b^-$.
\end{corollary}

\begin{proof}
By Theorem \ref{notr} the spacetime contains a timelike line, so it splits by the Lorentzian splitting theorem. The rest is Example \ref{exstatic}.
\end{proof}

Thus, under TCC, the intersection of the two cosmological developments is either empty or the whole spacetime. Equivalently: the cosmological wall is empty if and only if the two developments meet.

\begin{corollary} \label{bbwe}
Assume that $(M,g)$ is timelike complete. If $\textrm{dom}\, b^+$ (or $\textrm{dom}\, b^-$) contains a Cauchy hypersurface the spacetime contains a timelike line, hence splits under TCC.
\end{corollary}

\begin{proof}
If $S\subset \textrm{dom}\, b^+$ is a Cauchy hypersurface, by Prop.\ \ref{cmol} (time dual) $\textrm{dom}\, b^-$ intersects $S$ and hence $\textrm{dom}\, b^+$, so Thm.\ \ref{notr} applies.
\end{proof}

\begin{remark} \label{rembbwe}
The proof of Corollary \ref{bbwe} uses $b^-$, hence past timelike completeness. Assume instead only future timelike completeness, and let $S\subset \textrm{dom}\, b^+$ be a Cauchy hypersurface. The spacetime still contains a $b^+$-calibrated timelike line, possibly past incomplete. Indeed, take $S$ as the base hypersurface and let $\eta:(a,0]\to M$, $a\in[-\infty,0)$, be a past $S$-ray (Theorem \ref{cpgx} and its time dual), parametrized with proper time, so that $\eta(0)\in S$, $\eta(s)\le\eta(s')$ for $s\le s'$, and $d_S(\eta(s))=s$; then $\eta\subset J^-(S)\subset \textrm{dom}\, b^+$, and $b^+(\eta(s))-d_S(\eta(s))=b^+(\eta(s))-s$ is non-decreasing in $s$ by the rushing property, hence bounded by its value at $s=0$ along the past inextendible curve $\eta$, so that Theorem \ref{thmonesided} applies. If $\textrm{dom}\, b^+\ne M$, the future cosmological horizon $N:=\p\, \textrm{dom}\, b^+$ is non-empty, since $M$ is connected and $\textrm{dom}\, b^+$ is open and non-empty, and it has the following properties.
\begin{itemize}
\item[(a)] $N\subset I^+(S)$. Indeed, $J^-(S)\subset \textrm{dom}\, b^+$ as $\textrm{dom}\, b^+$ is a past set, while $N\cap \textrm{dom}\, b^+=\emptyset$ as $\textrm{dom}\, b^+$ is open, and $M\backslash J^-(S)=I^+(S)$.
\item[(b)] $N$ is not compact: from every point of $N$ starts a future lightlike ray contained in $N$ (Proposition \ref{propbdry}), which would be imprisoned in $N$, in contradiction with non-total imprisonment.
\item[(c)] $N$ is not contained in $J^-(S')$ for any Cauchy hypersurface $S'$: otherwise the closed set $N$ would be contained in the compact set $J^+(S)\cap J^-(S')$, in contradiction with (b).
\end{itemize}
Thus the future cosmological horizon would form only after $S$ and would reach the future infinity of the spacetime, much as the event horizon of a black hole which forms by gravitational collapse. We do not know whether, under TCC and future timelike completeness, the existence of a Cauchy hypersurface contained in $\textrm{dom}\, b^+$ implies $\textrm{dom}\, b^+=M$.
\end{remark}

The following result concerns the local equi-Lipschitz property of the Lorentzian distances to the leaves of a Cauchy foliation, which fails in general \cite{galloway26}. It shows that, for the foliation by the level sets $\Sigma_\tau$ of the distance from $\Sigma$, this property near $\Sigma$ is equivalent to the inclusion $\Sigma\subset\textrm{dom}\, b^+$, and hence, under TCC, to the splitting of the spacetime.

\begin{corollary}[Equi-Lipschitz criterion] \label{coreqlip}
Assume that $(M,g)$ is future timelike complete. The family $\{d(\cdot,\Sigma_\tau)\}_{\tau>0}$ is locally eventually equi-Lipschitz on $\textrm{dom}\, b^+$, while for no neighborhood $V$ of a point of $\p\, \textrm{dom}\, b^+$ and no $\tau_0>0$ is the family $\{d(\cdot,\Sigma_\tau)\}_{\tau\ge \tau_0}$ equi-Lipschitz on $V$. Consequently the following conditions are equivalent:
\begin{itemize}
\item[(i)] the family $\{d(\cdot,\Sigma_\tau)\}_{\tau>0}$ is locally eventually equi-Lipschitz on a neighborhood of $\Sigma$;
\item[(ii)] $\Sigma\subset \textrm{dom}\, b^+$.
\end{itemize}
If $(M,g)$ is timelike complete, they imply the existence of a timelike line and, under TCC, they are equivalent to the splitting of the spacetime. The time dual statements hold, under past timelike completeness, for $\{d(\Sigma_{-\tau},\cdot)\}_{\tau>0}$ and $\textrm{dom}\, b^-$.
\end{corollary}

\begin{proof}
The first statement is Theorem \ref{thmequi}. Let $q\in \p\, \textrm{dom}\, b^+$ and suppose that, for some $\tau_0>0$, the family $\{d(\cdot,\Sigma_\tau)\}_{\tau\ge \tau_0}$ is $L$-Lipschitz, with respect to a coordinate metric, on a convex coordinate ball $V$ centered at $q$. Since $\textrm{dom}\, b^+$ is open, $q\notin \textrm{dom}\, b^+$, while $V$ contains some $y\in \textrm{dom}\, b^+$. Then for $\tau\ge \tau_0$
\[
b_\tau(q)\le b_\tau(y)+L\vert q-y\vert\le b^+(y)+L\vert q-y\vert,
\]
a contradiction since $b_\tau(q)\to b^+(q)=+\infty$.

(i) $\Rightarrow$ (ii). By the previous paragraph $\Sigma\cap \p\, \textrm{dom}\, b^+=\emptyset$. Then $\Sigma\cap \textrm{dom}\, b^+$ is open in $\Sigma$, closed in $\Sigma$ (as $\Sigma\cap \overline{\textrm{dom}\, b^+}=\Sigma\cap \textrm{dom}\, b^+$) and non-empty (Proposition \ref{cmol}). As $\Sigma$ is connected, $\Sigma\subset \textrm{dom}\, b^+$.

(ii) $\Rightarrow$ (i) follows from Theorem \ref{thmequi}, since $\textrm{dom}\, b^+$ is an open neighborhood of $\Sigma$. Corollary \ref{bbwe} shows that (ii) implies the existence of a timelike line. Conversely, if the spacetime splits then $\textrm{dom}\, b^+=M$ by Example \ref{exstatic}.
\end{proof}

Corollary \ref{coreqlip} shows that, for the distance foliation, local equi-Lipschitz bounds near a single compact Cauchy hypersurface are equivalent to Bartnik's conjecture. This is in line with \cite{galloway26}, where conjectures equivalent to Bartnik's are formulated in terms of equi-Lipschitz bounds for the distances to the level sets of smooth Cauchy temporal functions \cite[Conjecture 4.1]{galloway26}. In a counterexample the bounds hold on the future development and fail at every point of the future cosmological horizon $\p\,\textrm{dom}\, b^+$, which meets every Cauchy hypersurface (Theorem \ref{thmwall}).

The {\em no future observer horizon} condition (NFOH) states that for every future inextendible causal curve $\gamma$, $I^-(\gamma)=M$.

\begin{proposition}
The NFOH condition is equivalent to the absence of future lightlike rays.
\end{proposition}

\begin{proof}
Suppose NFOH fails, so that $I^-(\gamma)\ne M$ for some future inextendible $\gamma$. Since $M$ is connected, $\p I^-(\gamma)$ is non-empty and from each of its points starts a future lightlike ray contained in it. Conversely, given a future lightlike ray $\eta$ starting from $p$, $p\notin I^-(\eta)$ since otherwise $d(p,\eta(t))>0$ for some $t$, contradicting maximality, thus $I^-(\eta)\ne M$.
\end{proof}

The following appeared in \cite{galloway19} (improving \cite{galloway84b} where the result first appeared under a stronger no observer horizon assumption).

\begin{corollary} \label{cornfoh}
Assume that $(M,g)$ is timelike complete. If NFOH (or its time dual) holds, then the spacetime contains a timelike line, hence it splits under TCC.
\end{corollary}

\begin{proof}
From every point of $\p\, \textrm{dom}\, b^+$ starts a future lightlike ray (Prop.\ \ref{propbdry}), thus this set is empty and hence $\textrm{dom}\, b^+=M$, $M$ being connected. As $\textrm{dom}\, b^-$ is non-empty the two domains intersect and Thm.\ \ref{notr} applies.
\end{proof}

Lesourd \cite{lesourd18b} proved that a timelike complete spacetime with compact Cauchy hypersurfaces which satisfies the TCC splits iff it does not admit a null ray. The proof above shows that the absence of future lightlike rays starting from $\p\, \textrm{dom}\, b^+$ suffices.

We obtain as another corollary \cite[Thm.\ 4.4]{galloway97}, and hence, under TCC, \cite[Thm.\ 4.1]{galloway97}.

\begin{corollary} \label{corray}
Assume that $(M,g)$ is timelike complete and let $\Sigma$ be a compact Cauchy hypersurface. Let $\gamma$ be a future $\Sigma$-ray and let $\eta$ be a past $\Sigma$-ray. If $I^+(\eta)\cap I^-(\gamma)\ne \emptyset$ the spacetime contains a timelike line, hence it splits under TCC.
\end{corollary}

Of course the condition $I^+(\eta)\cap I^-(\gamma)\ne \emptyset$ is weaker than the condition: there is a future $\Sigma$-ray $\gamma$ such that $I^-(\gamma)\supset \Sigma$, as $\Sigma$ admits a past $\Sigma$-ray $\eta$  and so $I^+(\eta)\cap I^-(\gamma)\ne \emptyset$.

\begin{proof}
From Prop.\ \ref{cmol} we know that $I^-(\gamma)\subset \textrm{dom}\, b^+$ and $I^+(\eta)\subset \textrm{dom}\, b^-$, thus $\textrm{dom}\, b^+\cap \textrm{dom}\, b^-\ne \emptyset$ and Thm.\ \ref{notr} applies.
\end{proof}

Galloway and Vega observed that in Galloway's ray-to-ray condition the two rays may be taken from different Cauchy hypersurfaces \cite[Sect.\ 5]{galloway17}. In our framework this generalization costs nothing: if $\gamma$ is a future $S$-ray and $\eta$ a past $S'$-ray, with $S,S'$ compact Cauchy hypersurfaces, the proof above gives $I^-(\gamma)\subset \textrm{dom}\, b^+_S$ and $I^+(\eta)\subset \textrm{dom}\, b^-_{S'}$, and these domains coincide with $\textrm{dom}\, b^\pm$ by Theorem \ref{xxxt} and its time dual. They also noted that Galloway's proof of \cite[Thm.\ 4.4]{galloway97} uses timelike completeness in one time direction only \cite[Sect.\ 5]{galloway17}. This sharper version follows from Theorem \ref{thmonesided}. Assume that $(M,g)$ is only future timelike complete, and let $\eta:(a,0]\to M$, $a\in[-\infty,0)$, be a past $S'$-ray, parametrized with proper time, so that $\eta(0)\in S'$, $\eta(s)\le\eta(s')$ for $s\le s'$ and $d_{S'}(\eta(s))=s$, with $I^+(\eta)\cap I^-(\gamma)\ne\emptyset$. Take $S'$ as the base hypersurface. As $I^-(\gamma)\subset \textrm{dom}\, b^+$ is a past set and $\eta(s_0)\ll q$ for some $q\in I^-(\gamma)$ and $s_0$, we have $\eta(s)\in \textrm{dom}\, b^+$ for $s\le s_0$. Since $d_{S'}(\eta(s))=s$ and $b^+$ is rushing, $E(\eta(s))=b^+(\eta(s))-s$ is non-decreasing in $s$, hence bounded by $E(\eta(s_0))$ along the past inextendible curve $\eta\vert_{(a,s_0]}$, and Theorem \ref{thmonesided} provides a $b^+$-calibrated timelike line. The past timelike complete case is the time dual.

\begin{corollary}
Assume that $(M,g)$ is timelike complete. Let $\Sigma$ be a compact Cauchy hypersurface and let $\sigma$ be a future $\Sigma$-ray. If a level set of $b^\sigma$ is a Cauchy hypersurface then the spacetime contains a timelike line, hence it splits under TCC.
\end{corollary}

\begin{proof}
Every Cauchy hypersurface intersects $I^-(\sigma)$, thus $b^\sigma$ is finite at some point of every Cauchy hypersurface. Thus $S:=(b^\sigma)^{-1}(c)$ can be Cauchy only for $c$ finite, and then $S\subset I^-(\sigma)$ since $b^\sigma=+\infty$ outside $I^-(\sigma)$. By Prop.\ \ref{cmol}, $S\subset I^-(\sigma)\subset \textrm{dom}\, b^+$, and Corollary \ref{bbwe} applies.
\end{proof}

\begin{remark} \label{remminmax}
Galloway and Vega \cite{galloway14b} introduce the max--min condition \cite[Def.\ 4.7]{galloway14b}: there is some $B>0$ such that for every $k$
\[
\max_{p\in \Sigma} d(p,\Sigma_k)-\min_{p\in \Sigma} d(p,\Sigma_k)\le B.
\]
This condition has a simple interpretation in terms of Cauchy Busemann functions (they used  Cauchy (pre)horospheres). Since $\max_{p\in \Sigma} d(p,\Sigma_k)=k$, the condition is equivalent to $\max_\Sigma b^+_k\le B$ for every $k$, or $b^+\vert_\Sigma\le B$, or, which is the same, $\Sigma\subset \textrm{dom}\, b^+$, since in this case $b^+$ is continuous over the compact set $\Sigma$ and hence bounded from above. By Corollary \ref{bbwe}, if $(M,g)$ is timelike complete and satisfies TCC, the max--min condition implies the splitting, which is \cite[Thm.\ 4.9]{galloway14b}. As observed by C.\ Vega, under timelike completeness the max--min condition implies Galloway's ray-to-ray condition of Corollary \ref{corray} \cite[Sect.\ 5]{galloway17}. Indeed, by \cite[Lemma 4.8]{galloway14b} the past Cauchy horosphere $S':=S^-_\infty(\Sigma)$ is a compact Cauchy hypersurface, hence it admits a past $S'$-ray $\eta$, ending at some point $q\in S'$, while a future timelike $S'$-ray $\gamma$ starts from each point of $S'$, in particular from $q$ \cite[Sect.\ 4.3]{galloway14b}, so that $I^+(\eta)\cap I^-(\gamma)\ne \emptyset$.
\end{remark}

In \cite{galloway17} Galloway and Vega introduced a further condition, involving the relative position of a future and a past Cauchy horosphere, which they showed to be implied by both the generalized ray-to-ray condition and the max--min condition \cite[Lemma 5.7]{galloway17}, and to imply the splitting under TCC \cite[Thm.\ 5.8]{galloway17}. Let us denote with $b^\pm_S$ the Cauchy Busemann functions constructed from a compact Cauchy hypersurface $S$, and with $S^-_\infty(S):=\p \cup_k I^+\big((b^{+}_{S,k})^{-1}(0)\big)$ the past Cauchy horosphere of $S$, and with $S^+_\infty(S)$ its time dual, the future Cauchy horosphere. By Proposition \ref{prophoro} and its time dual, $S^-_\infty(S)=(b^+_S)^{-1}(0)$ and $S^+_\infty(S)=(b^-_S)^{-1}(0)$. The Cauchy Busemann functions reveal that the condition of Galloway and Vega is not merely sufficient for the existence of a timelike line, but is in fact equivalent to the intersection of the two cosmological developments, hence to all the conditions of Theorem \ref{thmscal}.

\begin{corollary} \label{corhorohoro}
Assume that $(M,g)$ is timelike complete and let $\Sigma$ be a compact Cauchy hypersurface. The following conditions are equivalent:
\begin{itemize}
\item[(i)] there are compact Cauchy hypersurfaces $S$ and $S'$ such that
\[
J^+(S^+_\infty(S'))\cap J^-(S^-_\infty(S))\ne \emptyset;
\]
\item[(ii)] $\textrm{dom}\, b^+\cap \textrm{dom}\, b^-\ne \emptyset$.
\end{itemize}
Moreover, if (ii) holds then (i) holds with $S=\Sigma_s$ and $S'=\Sigma_{-s}$ for every sufficiently large $s>0$. In particular, (i) implies the existence of a timelike line and, under TCC, it is equivalent to the splitting of the spacetime.
\end{corollary}

\begin{proof}
By Propositions \ref{prophoro} and \ref{cllp},
\[
J^-(S^-_\infty(S))=J^-\big((b^+_S)^{-1}(0)\big)=(b^+_S)^{-1}\big((-\infty,0]\big),
\]
and dually $J^+(S^+_\infty(S'))=(b^-_{S'})^{-1}\big([0,+\infty)\big)$. Thus (i) holds iff there is $p\in M$ with $b^+_S(p)\le 0\le b^-_{S'}(p)$. Such a point belongs to $\textrm{dom}\, b^+_S\cap \textrm{dom}\, b^-_{S'}$, which coincides with $\textrm{dom}\, b^+\cap \textrm{dom}\, b^-$ by Theorem \ref{xxxt} and its time dual, so (i) implies (ii). Conversely, let $p\in \textrm{dom}\, b^+\cap \textrm{dom}\, b^-$ and let $s>0$ be such that $s\ge \max\{b^+(p),-b^-(p)\}$. The sets $\Sigma_{\pm s}$ are compact Cauchy hypersurfaces (Theorem \ref{kvvb}), and by Remark \ref{cort} and its time dual, $b^+_{\Sigma_s}=b^+-s$ and $b^-_{\Sigma_{-s}}=b^-+s$. Thus $b^+_{\Sigma_s}(p)\le 0\le b^-_{\Sigma_{-s}}(p)$, namely (i) holds with $S=\Sigma_s$, $S'=\Sigma_{-s}$. The last statement follows from Theorems \ref{notr} and \ref{thmscal}.
\end{proof}

Thus a condition formulated through a pair of auxiliary Cauchy hypersurfaces and the limits of their Lorentzian spheres is equivalent to a condition, the intersection of $\textrm{dom}\, b^+$ and $\textrm{dom}\, b^-$, which does not depend on any choice, and, by Theorem \ref{thmscal}, to the boundedness of the scalar defect $\Delta_k$. This equivalence becomes visible only once the horospheres are recognized as level sets of continuous functions whose finiteness domains are independent of the base hypersurface.

Garc\'ia-Heveling \cite{heveling25} calls a TIP {\em minimal} if it contains no other TIP as a proper subset, and similarly for TIFs. The spacetime satisfies the {\em min-min condition} if some minimal TIP intersects some minimal TIF. He proved that, in a spacetime with compact Cauchy hypersurfaces, every TIP contains a minimal TIP and every TIF contains a minimal TIF \cite[Thm.\ 1.2]{heveling25}. He also proved that every minimal TIP is the chronological past of a future $\Sigma$-ray \cite[Prop.\ 4.2]{heveling25}, and that the min-min condition implies the existence of a timelike line \cite[Thm.\ 1.4]{heveling25}. The next corollary recovers the last two results from Theorems \ref{cbdc} and \ref{notr}, and shows that minimal TIPs are contained in $\textrm{dom}\, b^+$ and minimal TIFs in $\textrm{dom}\, b^-$.

\begin{corollary} \label{cormintip}
Let $\Sigma$ be a compact Cauchy hypersurface.
\begin{itemize}
\item[(a)] Let $P$ be a minimal TIP. Then $P=I^-(\sigma)$ for some future $\Sigma$-ray $\sigma$, and $P\subset \textrm{dom}\, b^+$ if $(M,g)$ is future timelike complete. Dually, let $F$ be a minimal TIF. Then $F=I^+(\eta)$ for some past $\Sigma$-ray $\eta$, and $F\subset \textrm{dom}\, b^-$ if $(M,g)$ is past timelike complete.
\item[(b)] Assume that $(M,g)$ is timelike complete and that the min-min condition holds. Then $\textrm{dom}\, b^+\cap \textrm{dom}\, b^-\ne \emptyset$. Hence the spacetime contains a timelike line and, under TCC, it splits.
\end{itemize}
\end{corollary}

\begin{proof}
(a) By Theorem \ref{cbdc}, $P$ contains a TIP of the form $I^-(\sigma)$ with $\sigma$ a future $\Sigma$-ray. By minimality $P=I^-(\sigma)$, and $I^-(\sigma)\subset \textrm{dom}\, b^+$ by Prop.\ \ref{cmol}. The statement for $F$ is the time dual.

(b) Let $P$ be a minimal TIP and $F$ a minimal TIF with $P\cap F\ne \emptyset$. By (a), $P\cap F\subset \textrm{dom}\, b^+\cap \textrm{dom}\, b^-$, and Theorem \ref{notr} applies.
\end{proof}

By (a), the min-min condition implies the condition $I^+(\eta)\cap I^-(\sigma)\ne \emptyset$ of Corollary \ref{corray}, as already observed in \cite[Prop.\ 4.2]{heveling25}.

\subsection{Mean curvature of Cauchy hypersurfaces and rigidity of Hawking's theorem} \label{secmeanC}

Recall that future timelike completeness is needed to define $b^+$, and past timelike completeness to define $b^-$.

Let $\Sigma$ be a smooth spacelike compact Cauchy hypersurface, $N$ its future unit normal, $H_\Sigma$ its future mean curvature in the sense of Section \ref{secmean}, namely the expansion on $\Sigma$ of the congruence of future normal geodesics, and $E(z,s):=\exp_z(sN_z)$ the normal exponential map. Let $d_\Sigma$ be the signed distance from $\Sigma$; it is rushing (Lemma \ref{lemsandwich}). Let
\[
F^+:=\{y\in \Sigma: \ y \textrm{ is the initial point of a future } \Sigma\textrm{-ray}\},
\]
and dually $F^-$ the set of final points of past $\Sigma$-rays. The sets $F^\pm$ are non-empty and the $\Sigma$-rays are timelike (Theorem \ref{cpgx} and its time dual).

If $(M,g)$ is future timelike complete, let $b^+$ be the future Cauchy Busemann function of $\Sigma$, so that $d_\Sigma\le b^+$ by (\ref{cmlr}). Every future $\Sigma$-ray $\sigma$ is timelike, so $\sigma(0)\in I^-(\sigma)\subset \textrm{dom}\, b^+$ (Proposition \ref{cmol}): thus $F^+\subset \Sigma\cap \textrm{dom}\, b^+$. The next lemma characterizes the equality case in (\ref{cmlr}) on $J^+(\Sigma)$. It holds for every compact Cauchy hypersurface, not necessarily smooth.

\begin{lemma} \label{lemcontact}
Assume that $(M,g)$ is future timelike complete. Let $\Sigma$ be a compact Cauchy hypersurface and $b^+$ its future Cauchy Busemann function. A point $x\in J^+(\Sigma)$ satisfies $d_\Sigma(x)=b^+(x)$ if and only if it lies on a future $\Sigma$-ray. More precisely:
\begin{itemize}
\item[(a)] if $\sigma$ is a proper time parametrized future $\Sigma$-ray, then $d_\Sigma(\sigma(t))=b^+(\sigma(t))=t$ for every $t\ge 0$;
\item[(b)] if $x\in J^+(\Sigma)$ satisfies $d_\Sigma(x)=b^+(x)$, $\gamma$ is a $\Sigma$-maximizer to $x$ (a constant curve if $x\in \Sigma$) and $\alpha$ is a proper time parametrized $b^+$-calibrated ray starting from $x$, then the concatenation of $\gamma$ and $\alpha$ is a future $\Sigma$-ray.
\end{itemize}
In particular, $F^+=\Sigma\cap (b^+)^{-1}(0)$, and this set is closed.
\end{lemma}

\begin{proof}
(a) As $\sigma$ is $\Sigma$-maximizing, $d_\Sigma(\sigma(t))=t$, while $b^+(\sigma(t))=t$ by Proposition \ref{ncho}.

(b) As $b^+(x)=d_\Sigma(x)<+\infty$, $x\in \textrm{dom}\, b^+$, so $\alpha$ exists and is a future complete timelike ray (Proposition \ref{bber} and Theorem \ref{xxhg}(i)). For $t\ge 0$, by (\ref{cmlr}) and calibration, $d_\Sigma(\alpha(t))\le b^+(\alpha(t))=b^+(x)+t=d_\Sigma(x)+t$, while, $d_\Sigma$ being rushing (proof of Lemma \ref{lemsandwich}) and $\alpha$ maximizing, $d_\Sigma(\alpha(t))\ge d_\Sigma(x)+d(x,\alpha(t))=d_\Sigma(x)+t$. Thus the concatenation of $\gamma$ and $\alpha\vert_{[0,t]}$ has Lorentzian length $d_\Sigma(x)+t=d_\Sigma(\alpha(t))$, namely it is $\Sigma$-maximizing, for every $t\ge 0$. Since it is future inextendible, it is a future $\Sigma$-ray.

For $x\in \Sigma$ we have $d_\Sigma(x)=0$, so (a) and (b) give $F^+=\Sigma\cap (b^+)^{-1}(0)$, which is closed since $b^+$ is continuous (Theorem \ref{xxhg}(iv)).
\end{proof}

Dually, if $(M,g)$ is past timelike complete, then $F^-\subset \Sigma\cap \textrm{dom}\, b^-$ and $F^-=\Sigma\cap (b^-)^{-1}(0)$ is closed.

We shall use the following standard facts on the future cut locus of $\Sigma$, the analogs for hypersurfaces of those recalled before Lemma \ref{lemomega}. We refer to \cite[Sect.\ 2.3]{treude13}, whose proofs are given in \cite[Sect.\ 3.2.5]{treude11}; they apply since $\Sigma$, being compact, is future causally complete in the sense of \cite[Def.\ 2.3]{treude13}. They will be used only in results which assume future timelike completeness, so we assume it in this paragraph; then $E$ is defined on $\Sigma\times [0,+\infty)$. We identify the normal bundle of $\Sigma$ with $\Sigma\times \mathbb{R}$ through $(z,s)\mapsto sN_z$, so that $E$ is the normal exponential map. For $z\in \Sigma$ let
\[
s^+(z):=\sup\{s\ge 0: \ d_\Sigma(E(z,s))=s\}\in (0,+\infty],
\]
so that the normal geodesic $E(z,\cdot)$ is $\Sigma$-maximizing on $[0,s]$ for $0\le s<s^+(z)$, and not beyond $s^+(z)$. If $s^+(z)<+\infty$, the point $E(z,s^+(z))$ is the {\em future cut point} of $\Sigma$ along $E(z,\cdot)$, and the {\em future cut locus} $\textrm{Cut}^+(\Sigma)$ is the set of these points. Let $\mathcal{U}$ be the interior of $\{(z,s)\in \Sigma\times \mathbb{R}: 0\le s<s^+(z)\}$. Then:
\begin{itemize}
\item[(i)] every $\Sigma$-maximizer is a timelike geodesic orthogonal to $\Sigma$ without focal points of $\Sigma$ before its final point \cite[Thm.\ 2.5(1)]{treude13}, see also \cite[Ch.\ 10]{oneill83}; in particular, a future $\Sigma$-ray starting at $y\in \Sigma$ is of the form $s\mapsto E(y,s)$, $s\ge 0$;
\item[(ii)] $\textrm{Cut}^+(\Sigma)$ is closed in $M$, and $E$ maps $\mathcal{U}$ diffeomorphically onto the open set $I^+(\Sigma)\setminus \textrm{Cut}^+(\Sigma)$ \cite[Thm.\ 2.6(1)--(3)]{treude13};
\item[(iii)] every $x\in I^+(\Sigma)\setminus \textrm{Cut}^+(\Sigma)$ is reached by a unique $\Sigma$-maximizer, $d_\Sigma$ is smooth on $I^+(\Sigma)\setminus \textrm{Cut}^+(\Sigma)$, and $-\nabla d_\Sigma(x)$ is the final velocity of the proper time parametrized $\Sigma$-maximizer to $x$ \cite[Thm.\ 2.6(4), Prop.\ 2.7]{treude13}.
\end{itemize}

Fact (i) and the Raychaudhuri equation give the following sign property of $H_\Sigma$ on $F^\pm$, which is the mechanism behind Hawking's singularity theorem \cite{hawking67}.

\begin{proposition} \label{prop39}
Assume TCC and let $\Sigma$ be a smooth spacelike compact Cauchy hypersurface with future mean curvature $H_\Sigma$. If $(M,g)$ is future timelike complete, then $H_\Sigma\ge 0$ at the initial point of every future $\Sigma$-ray, namely on $F^+$; in particular, $H_\Sigma$ cannot be negative all over the compact set $\Sigma\cap \overline{\textrm{dom}\, b^+}$. Dually, if $(M,g)$ is past timelike complete, then $H_\Sigma\le 0$ at the final point of every past $\Sigma$-ray, namely on $F^-$.
\end{proposition}

\begin{proof}
A future $\Sigma$-ray $\sigma$ is a timelike geodesic orthogonal to $\Sigma$ without focal points of $\Sigma$ before any of its points, by (i). If $H_\Sigma(\sigma(0))<0$, by the Raychaudhuri equation and TCC there would be a focal point along $\sigma$ within proper time $n/\vert H_\Sigma(\sigma(0))\vert$ \cite[Ch.\ 12]{beem96}, which is reached since $\sigma$ is future complete. This contradicts the ray property.
\end{proof}

The proof gives the following quantitative version, which requires TCC but no completeness assumption. If $H_\Sigma(y)\le -\delta<0$ at $y\in \Sigma$, then the future normal geodesic issued from $y$ ceases to be $\Sigma$-maximizing at proper time at most $n/\delta$, if it is defined that far. Thus, on a compact subset of $\Sigma$ where $H_\Sigma\le -\delta$, the normal $\Sigma$-maximizing segments have proper time duration bounded by $n/\delta$.

The following lemma, which localizes facts (i)--(iii) along a future $\Sigma$-ray, will be used in the proof of Theorem \ref{thmhawk}.

\begin{lemma} \label{lemsmoothT}
Assume that $(M,g)$ is future timelike complete. Let $\sigma=E(y,\cdot)$ be a future $\Sigma$-ray and $a>0$. There is a connected neighborhood $W$ of $y$ in $\Sigma$ such that $E$ maps $W\times (0,a)$ diffeomorphically onto an open set $V\subset I^+(\Sigma)$, and $d_\Sigma(E(z,s))=s$ for $(z,s)\in W\times [0,a)$. In particular, $\sigma((0,a))\subset V$, $d_\Sigma$ is smooth on $V$, $-\nabla d_\Sigma$ is the tangent field of the normal geodesics $E(z,\cdot)$, $z\in W$, and these geodesics have no focal points of $\Sigma$ on $(0,a)$.
\end{lemma}

\begin{proof}
We first show that $(y,t)\in \mathcal{U}$ for every $t>0$. Since $\sigma$ is a future $\Sigma$-ray, $s^+(y)=+\infty$. The segment $\sigma\vert_{[0,t]}$ is the unique $\Sigma$-maximizer to $\sigma(t)$. Indeed, if $\gamma$ is a $\Sigma$-maximizer to $\sigma(t)$, then $\gamma$ followed by $\sigma\vert_{[t,t+1]}$ has Lorentzian length $t+1=d_\Sigma(\sigma(t+1))$, so it is a $\Sigma$-maximizer, hence an unbroken geodesic by (i), and thus $\gamma=\sigma\vert_{[0,t]}$. Suppose that $\sigma(t)=E(z,s^+(z))$ for some $z\in \Sigma$ with $s^+(z)<+\infty$. By the continuity of $d_\Sigma$, $d_\Sigma(E(z,s^+(z)))=s^+(z)$, so $E(z,\cdot)\vert_{[0,s^+(z)]}$ is a $\Sigma$-maximizer to $\sigma(t)$. By uniqueness it coincides with $\sigma\vert_{[0,t]}$, so $z=y$, in contradiction with $s^+(y)=+\infty$. Thus $\sigma(t)\in I^+(\Sigma)\setminus \textrm{Cut}^+(\Sigma)$, and by (ii) $\sigma(t)=E(z,s)$ for some $(z,s)\in \mathcal{U}$. As $s<s^+(z)$, the curve $E(z,\cdot)\vert_{[0,s]}$ is a $\Sigma$-maximizer to $\sigma(t)$, so $(z,s)=(y,t)$ by uniqueness, and $(y,t)\in \mathcal{U}$.

Next, there is $\epsilon>0$ such that $\Sigma\times (0,\epsilon)\subset \mathcal{U}$. Indeed, the future cut points lie in $I^+(\Sigma)$, so by (ii) $\textrm{Cut}^+(\Sigma)$ is a closed set which does not meet the compact set $\Sigma=E(\Sigma\times\{0\})$. By the continuity of $E$ and the compactness of $\Sigma$ there is $\epsilon>0$ such that $E(\Sigma\times [0,\epsilon))$ does not meet $\textrm{Cut}^+(\Sigma)$. If $(z,s)\in \Sigma\times (0,\epsilon)$ and $s\ge s^+(z)$, then the future cut point $E(z,s^+(z))$ would belong to $E(\Sigma\times (0,\epsilon))$, which is impossible. Hence $s<s^+(z)$ on the open set $\Sigma\times (0,\epsilon)$, which is therefore contained in $\mathcal{U}$.

Since $\mathcal{U}$ is open and contains the compact set $\{y\}\times [\epsilon/2,a]$, there is a connected neighborhood $W$ of $y$ in $\Sigma$ such that $W\times [\epsilon/2,a]\subset \mathcal{U}$, hence $W\times (0,a]\subset \mathcal{U}$. By (ii), $V:=E(W\times (0,a))$ is open and $E\vert_{W\times (0,a)}$ is a diffeomorphism onto $V$; moreover $\sigma((0,a))\subset V\subset I^+(\Sigma)$. For $(z,s)\in W\times (0,a]$ we have $s<s^+(z)$, so $d_\Sigma(E(z,s))=s$, which holds trivially for $s=0$, and $E(z,\cdot)\vert_{[0,s]}$ is the $\Sigma$-maximizer to $E(z,s)$. By (iii), $d_\Sigma$ is smooth on $V$ and $-\nabla d_\Sigma(E(z,s))=\p_s E(z,s)$. Finally, for $z\in W$ the geodesic $E(z,\cdot)\vert_{[0,a]}$ is $\Sigma$-maximizing, so by (i) it has no focal points of $\Sigma$ on $(0,a)$.
\end{proof}

The name of the next theorem refers to the rigidity problem for singularity theorems recalled in the Introduction. Hawking's theorem \cite{hawking67} states that, under TCC, if $H_\Sigma<0$ everywhere on $\Sigma$ then the spacetime is future timelike geodesically incomplete. Indeed, by compactness $H_\Sigma\le -\delta$ for some $\delta>0$, and by the Raychaudhuri equation every future normal geodesic has a focal point of $\Sigma$ within proper time $n/\delta$, so that $d_\Sigma\le n/\delta$ on $J^+(\Sigma)$. In the borderline case $H_\Sigma\le 0$ the conclusion fails, as shown by the slices $\{t=\textrm{const}\}$ of the static products $(\mathbb{R}\times S,-\dd t^2+h)$, which are totally geodesic. Rigidity means that, in the borderline case, future timelike completeness forces a product structure: under TCC, if $(M,g)$ is future timelike complete and $H_\Sigma\le 0$ everywhere, then $\Sigma$ is maximal and $J^+(\Sigma)$ is isometric, via the normal exponential map, to $([0,+\infty)\times \Sigma,-\dd t^2+h)$, where $h$ is the induced metric on $\Sigma$ \cite[Prop.\ 4.3]{galloway14b}. If $(M,g)$ is also past timelike complete, Bartnik's theorem then gives the splitting of the whole spacetime (Corollary \ref{corhawk}(a) below).

The next theorem localizes these statements: the sign condition is required only on a neighborhood of $F^+$. Since $H_\Sigma\ge 0$ on $F^+$ by Proposition \ref{prop39}, the hypothesis says that $\Sigma$ is in the borderline case of Hawking's theorem near the initial points of the future $\Sigma$-rays.

\begin{theorem}[Localized rigidity of Hawking's theorem] \label{thmhawk}
Assume TCC and that $(M,g)$ is future timelike complete. Let $\Sigma$ be a smooth spacelike compact Cauchy hypersurface with future mean curvature $H_\Sigma$. If $H_\Sigma\le 0$ on a neighborhood in $\Sigma$ of the set $F^+$ of the initial points of the future $\Sigma$-rays, in particular if
\[
H_\Sigma\le 0 \quad \textrm{on } \Sigma\cap \textrm{dom}\, b^+,
\]
then every point of $\Sigma$ is the initial point of a future $\Sigma$-ray, $\Sigma\subset \textrm{dom}\, b^+$, $\Sigma$ is maximal, i.e., $H_\Sigma=0$, and $J^+(\Sigma)$ is isometric, via the normal exponential map, to $([0,+\infty)\times \Sigma,-\dd t^2+h)$, where $h$ is the induced metric on $\Sigma$. If, moreover, $(M,g)$ is past timelike complete, then the spacetime splits.

Dually, assume TCC and that $(M,g)$ is past timelike complete. If $H_\Sigma\ge 0$ on a neighborhood in $\Sigma$ of the set $F^-$ of the final points of the past $\Sigma$-rays, in particular if $H_\Sigma\ge 0$ on $\Sigma\cap \textrm{dom}\, b^-$, then $\Sigma$ is maximal, i.e., $H_\Sigma=0$, and $J^-(\Sigma)$ splits as above. If, moreover, $(M,g)$ is future timelike complete, then the spacetime splits.
\end{theorem}

\begin{proof}
Since $\textrm{dom}\, b^+$ is open and contains $F^+$, the set $\Sigma\cap \textrm{dom}\, b^+$ is an open neighborhood of $F^+$ in $\Sigma$, so the second sign condition in the statement implies the first one. Let $\mathcal{N}$ be an open neighborhood of $F^+$ in $\Sigma$ on which $H_\Sigma\le 0$.

{\em Step 1: the level sets of $d_\Sigma$ near the rays have non-positive future mean curvature.} Let $\sigma=E(y,\cdot)$ be a future $\Sigma$-ray, and let $W$ and $V$ be as in Lemma \ref{lemsmoothT} with $a=2$ and $W\subset \mathcal{N}$ (recall $y\in F^+\subset \mathcal{N}$). For $z\in W$ the normal geodesic $E(z,\cdot)$ has no focal points on $(0,2)$, so the expansion $\theta(z,s)$ of the congruence of normal geodesics is defined for $s\in [0,2)$. For $s\in (0,2)$, $\theta(z,s)$ is the future mean curvature at $E(z,s)$ of the level set $d_\Sigma^{-1}(s)\cap V=E(W\times \{s\})$, whose future unit normal $-\nabla d_\Sigma$ is the tangent of the normal geodesics. By the Raychaudhuri equation and TCC, $\p_s\theta\le -\theta^2/n\le 0$, so $\theta(z,s)\le \theta(z,0)=H_\Sigma(z)\le 0$. Thus the smooth spacelike hypersurface $S:=E(W\times\{1\})$ has future mean curvature $H_S\le 0$.

{\em Step 2: local coincidence with a horosphere.} Let $B:=B^+_1=(b^+)^{-1}(1)$. By Proposition \ref{ncho}, $b^+(\sigma(1))=1=d_\Sigma(\sigma(1))$, so $\sigma(1)\in S\cap B$. By Proposition \ref{cllp} $B$ is acausal and edgeless, hence a $C^0$ spacelike hypersurface, and by Corollary \ref{cormeanconv} it has future mean curvature $H_B\ge 0$ in the support sense. Let $W'\subset W$ be a neighborhood of $y$ and $I'\subset (0,2)$ an open interval containing $1$, and let $U:=E(W'\times I')$. In $U$ the function $d_\Sigma$ is smooth with timelike gradient, so $S\cap U=E(W'\times \{1\})$ is acausal and edgeless in $U$. If $x=E(z,s)\in B\cap U$, then $s=d_\Sigma(x)\le b^+(x)=1$ by (\ref{cmlr}), and the normal geodesic $E(z,\cdot)\vert_{[s,1]}$ is a future-directed causal curve in $U$ from $x$ to $S$. Thus $B\cap U\subset J^-(S\cap U,U)$, namely $B$ lies locally to the past of $S$ near $\sigma(1)$. Theorem \ref{thmagh} with $a=0$ shows that $S$ and $B$ agree near $\sigma(1)$. Hence, possibly shrinking $W'$, for every $z\in W'$
\[
b^+(E(z,1))=1=d_\Sigma(E(z,1)).
\]

{\em Step 3: $F^+=\Sigma$.} The set $F^+$ is closed by Lemma \ref{lemcontact}.

The set $F^+$ is open. Let $y\in F^+$ and let $W'$ be as in Step 2. For $z\in W'$ the point $x_z:=E(z,1)$ satisfies $d_\Sigma(x_z)=1=b^+(x_z)$. The curve $E(z,\cdot)\vert_{[0,1]}$ has Lorentzian length $1=d_\Sigma(x_z)$, so it is a $\Sigma$-maximizer to $x_z$, and by Lemma \ref{lemcontact}(b) it extends to a future $\Sigma$-ray; this ray is the normal geodesic $E(z,\cdot)$. Thus $W'\subset F^+$.

Since $M$ is connected, so is $\Sigma$, and therefore $F^+=\Sigma$.

{\em Step 4: conclusion.} By Step 3, $\Sigma=F^+\subset \textrm{dom}\, b^+$ and $\Sigma=F^+\subset \mathcal{N}$, so $H_\Sigma\le 0$ on $\Sigma$, while $H_\Sigma\ge 0$ on $F^+=\Sigma$ by Proposition \ref{prop39}. Thus $H_\Sigma=0$, namely $\Sigma$ is maximal, and the splitting of $J^+(\Sigma)$ follows from \cite[Prop.\ 4.3]{galloway14b}. If $(M,g)$ is also past timelike complete, the future cosmological development $\textrm{dom}\, b^+$ contains the Cauchy hypersurface $\Sigma$, so by Corollary \ref{bbwe} the spacetime contains a timelike line and splits. The dual statement follows by time reversal, under which $H_\Sigma$ changes sign and $\textrm{dom}\, b^+$ and $\textrm{dom}\, b^-$ are exchanged.
\end{proof}

By Lemma \ref{lemcontact}, which does not use the assumption on $H_\Sigma$, the contact set $\{d_\Sigma=b^+\}\cap J^+(\Sigma)$ is exactly the union of the future $\Sigma$-rays. Steps 1 and 2 use the sign of $H_\Sigma$, the mean convexity of the horospheres and the maximum principle of Theorem \ref{thmagh} to show that the set $F^+$ of the initial points of these rays is open. Note that no differentiability of $b^+$ is used.

\begin{corollary} \label{corhawk}
Assume TCC and that $(M,g)$ is timelike complete.
\begin{itemize}
\item[(a)] If $M$ admits a smooth spacelike compact Cauchy hypersurface $\Sigma$ whose future mean curvature $H_\Sigma$ is everywhere non-positive, or everywhere non-negative, then $M$ splits. In particular, $M$ splits if it admits a compact CMC Cauchy hypersurface.
\item[(b)] If $M$ does not split, then for every smooth spacelike compact Cauchy hypersurface $\Sigma$ the closure of $\{H_\Sigma>0\}$ meets $F^+$ and the closure of $\{H_\Sigma<0\}$ meets $F^-$. In particular, $H_\Sigma$ is positive at some point of $\Sigma\cap \textrm{dom}\, b^+$ and negative at some point of $\Sigma\cap \textrm{dom}\, b^-$, two disjoint open subsets of $\Sigma$.
\end{itemize}
\end{corollary}

The last statement of (a) is Bartnik's theorem \cite[Cor.\ 1]{bartnik88}, see also \cite[Thm.\ 2]{galloway19}. The case $H_\Sigma\le 0$ everywhere of (a) also follows from \cite[Prop.\ 4.3]{galloway14b}, which gives the splitting of $J^+(\Sigma)$ and the maximality of $\Sigma$, combined with Bartnik's theorem.

\begin{proof}
(a) If $H_\Sigma\le 0$ (resp.\ $H_\Sigma\ge 0$) everywhere, apply Theorem \ref{thmhawk} (resp.\ its dual). A CMC hypersurface satisfies one of the two conditions. (b) $F^+$ is compact, so if the closure of $\{H_\Sigma>0\}$ did not meet $F^+$, then $H_\Sigma\le 0$ on a neighborhood of $F^+$, and Theorem \ref{thmhawk} would give a splitting. Since $F^+$ is contained in the open subset $\Sigma\cap \textrm{dom}\, b^+$ of $\Sigma$, the set $\{H_\Sigma>0\}$ meets $\Sigma\cap \textrm{dom}\, b^+$, and dually $\{H_\Sigma<0\}$ meets $\Sigma\cap \textrm{dom}\, b^-$. The disjointness of $\textrm{dom}\, b^+$ and $\textrm{dom}\, b^-$ in the non-split case is Theorem \ref{notr}.
\end{proof}

Thus, in a counterexample, every smooth compact Cauchy hypersurface is expanding somewhere in the future cosmological development and contracting somewhere in the past one. This sharpens part (a). The future mean curvature $H_\Sigma$ must change sign, and moreover its positive values must reach into the future development and its negative values into the past one. Combined with Proposition \ref{prop39}, which gives $H_\Sigma\ge 0$ on $F^+$, it says that either $H_\Sigma>0$ at some initial point of a future $\Sigma$-ray, or $H_\Sigma$ vanishes at every such point, namely on $F^+$, and takes positive values arbitrarily close to $F^+$.

\begin{remark}[Scope of the method] \label{remscope}
The proof of Theorem \ref{thmhawk} uses the future Cauchy Busemann function only through three properties: $d_\Sigma\le b^+$ with contact exactly along the future $\Sigma$-rays (Lemma \ref{lemcontact}), the mean convexity of the horospheres (Corollary \ref{cormeanconv}), and the geometric maximum principle (Theorem \ref{thmagh}). The horospheres act as a mean convex barrier at future infinity. They replace the second hypersurface in Galloway's hypersurface rigidity theorem \cite[Thm.\ 2.3]{galloway97}. The closest antecedent is \cite[Prop.\ 4.3]{galloway14b}. There, a hypersurface $S$ with $H_S\le 0$ everywhere in the support sense, admitting a future $S$-ray, is compared with the ray horosphere of that ray by means of the maximum principle of \cite{andersson98b}, and is shown to be maximal, with $J^+(S)$ split. The novelty of Theorem \ref{thmhawk} is that the sign of $H_\Sigma$ is only required near the initial points of the future $\Sigma$-rays. Since the Cauchy Busemann function collects all such rays at once, the open-closed argument can run on the set $F^+$ instead of on $\Sigma$.

More generally, the same argument shows that if a smooth spacelike hypersurface $S$ with future mean curvature $H_S\le 0$ touches a horosphere $B^+_c$ at a point $q$ and lies in $\{b^+\ge c\}$ near $q$, then $S$ and $B^+_c$ coincide near $q$; in particular $B^+_c$ is smooth and maximal near $q$.

This observation suggests a strategy for Bartnik's conjecture based on Bartnik's own splitting theorem. Under TCC and timelike completeness, the spacetime splits if and only if it admits a smooth compact maximal Cauchy hypersurface: in a product $(\mathbb{R}\times S,-\dd t^2+h)$ the slices $\{t\}\times S$ are maximal, and they are precisely the horospheres $B^+_c$ (Example \ref{exstatic}), while, conversely, a compact maximal Cauchy hypersurface has constant mean curvature, so the spacetime splits by Bartnik's theorem \cite[Cor.\ 1]{bartnik88}, see Corollary \ref{corhawk}(a). One may thus try to prove the conjecture by showing that some horosphere, or some other compact Cauchy hypersurface, is smooth and maximal. The observation above provides the local part of this program: a {\em barrier from the future} inside the future development, namely a smooth spacelike hypersurface with non-positive future mean curvature touching a horosphere from the future, makes the horosphere smooth and maximal near the contact point. The global part requires a continuation mechanism which propagates the contact over a whole compact Cauchy hypersurface. In the proof of Theorem \ref{thmhawk} both ingredients are available: the barriers are the level sets of $d_\Sigma$ near the future $\Sigma$-rays (Steps 1 and 2), and the continuation runs on the compact set $F^+\subset\Sigma$ (Step 3), so that, in the end, $\Sigma$ itself is a compact maximal Cauchy hypersurface. Without a sign condition on the mean curvature of a Cauchy hypersurface, we do not know how to produce a barrier from the future inside the future development.
\end{remark}

\section{Structure of a hypothetical counterexample} \label{secstru}

By Theorem \ref{xxxt} and its time dual, the future and past cosmological developments $\textrm{dom}\, b^\pm$ and the future and past cosmological horizons $\p\, \textrm{dom}\, b^\pm$, which are $C^0$ null hypersurfaces (Proposition \ref{propbdry} and its time dual), do not depend on $\Sigma$, and they carry physical significance (see the discussion after Theorem \ref{xxxt}). If Bartnik's conjecture is not true, in a timelike complete counterexample these objects are non-trivial: $\textrm{dom}\, b^\pm\ne M$ and $\p\, \textrm{dom}\, b^\pm\ne \emptyset$. If $(M,g)$ is timelike complete, satisfies TCC and does not split then, by Theorem \ref{notr}, $\textrm{dom}\, b^+$ and $\textrm{dom}\, b^-$ are disjoint non-empty open sets, thus, since $M$ is connected, $M\backslash (\textrm{dom}\, b^+\cup \textrm{dom}\, b^-)$ is non-empty. We call it the {\em cosmological wall}. It might consist of just the union of the cosmological horizons.

\begin{remark}[A $\Sigma$-independent form of the regularity criterion] \label{remcone}
By Theorem \ref{xxhg}(i), every future inextendible $b^+$-calibrated curve is a timelike geodesic ray, so it has a well defined future directed unit tangent vector at its initial point. For $p\in \textrm{dom}\, b^+$ let $C_p\subset T_pM$ be the set of the unit tangent vectors at $p$ of the future inextendible $b^+$-calibrated curves starting from $p$, where $b^+$ ranges over the future Cauchy Busemann functions of all the compact Cauchy hypersurfaces $\Sigma$. By Theorem \ref{xxxt} the set $\textrm{dom}\, b^+$ does not depend on $\Sigma$, so $C_p$ depends only on $(M,g)$ and $p$. By Remark \ref{remasym}, the elements of $C_p$ can be interpreted as the velocities at $p$ of the observers whose proper time asymptotically agrees, up to a constant, with the cosmic time $d_\Sigma$ of some compact Cauchy hypersurface $\Sigma$. If $(M,g)$ is timelike complete and satisfies TCC, the following conditions are equivalent:
\begin{itemize}
\item[(a)] $(M,g)$ splits;
\item[(b)] for every $p\in \textrm{dom}\, b^+$, $C_p$ consists of a single vector;
\item[(c)] for some compact Cauchy hypersurface $\Sigma$, from every point of $\textrm{dom}\, b^+$ starts exactly one future inextendible $b^+$-calibrated curve.
\end{itemize}
Indeed, suppose that (a) holds, so that $M=\mathbb{R}\times S$, $g=-\dd t^2+h$, with $(S,h)$ a compact Riemannian manifold. A proper time parametrized future directed timelike geodesic reads $\alpha(s)=(t_0+s/\sqrt{1-v^2}, \beta(vs/\sqrt{1-v^2}))$, $v\in[0,1)$, with $\beta$ a unit speed $h$-geodesic. For $s>0$ let $T:=s/\sqrt{1-v^2}$ and $u:=vT$. As $\textrm{dist}_h(\beta(0),\beta(u))\le u<T$, by (\ref{eqstatd}) we have $d(\alpha(0),\alpha(s))=\sqrt{T^2-\textrm{dist}_h(\beta(0),\beta(u))^2}$, while $\ell(\alpha\vert_{[0,s]})=s=\sqrt{T^2-u^2}$. Thus $\alpha$ is a ray, namely $d(\alpha(0),\alpha(s))=\ell(\alpha\vert_{[0,s]})$ for every $s>0$, only if $\textrm{dist}_h(\beta(0),\beta(u))=u$ for every $u>0$, that is, only if $\beta$ is a minimizing $h$-geodesic on $[0,+\infty)$. As $S$ is compact, $\textrm{dist}_h(\beta(0),\beta(u))$ is bounded by the diameter of $(S,h)$, so this forces $v=0$. Thus every future timelike ray is vertical and, by Theorem \ref{xxhg}(i), $C_p=\{\p_t\}$ for every $p$. The implication (b)$\Rightarrow$(c) is clear, as calibrated curves exist by Proposition \ref{bber}. If (c) holds, then $b^+$ is $C^1$ on $\textrm{dom}\, b^+$ by Proposition \ref{propcrease}, and the spacetime splits by Theorem \ref{thmcone}.

Thus, in a counterexample to Bartnik's conjecture, for every compact Cauchy hypersurface $\Sigma$ there are points of $\textrm{dom}\, b^+$ from which more than one $b^+$-calibrated curve starts. Without TCC the set $C_p$ can be much larger. De Sitter spacetime is the warped product of Example \ref{exwarp} with $I=\mathbb{R}$, $a=\cosh t$ and $(S,h)$ the round unit $n$-sphere. By the criterion given there it is future timelike complete and, as $a$ is even, it is timelike complete, while it violates TCC as $Ric(\p_t,\p_t)=-n$. Let $p\in M$ and $\Sigma=S_{t(p)}$. By Example \ref{exwarp}, $b^+=t-t(p)$, $\textrm{dom}\, b^+=M$, and the only future inextendible $b^+$-calibrated curve starting from $p$ is the comoving observer, whose velocity at $p$ is $\p_t$. Let $v$ be a future directed unit timelike vector at $p$. In the hyperboloid model $M=\{X\in \mathbb{R}^{n+1,1}: \langle X,X\rangle=1\}$ \cite[Sec.\ 5.2]{hawking73}, where $T_pM=p^\perp$, the linear isometries of $\mathbb{R}^{n+1,1}$ fixing $p$ act on $p^\perp$ as its full Lorentz group, so there is an isometry $\psi$ of $(M,g)$ such that $\psi(p)=p$ and $\dd\psi(\p_t)=v$ at $p$, which preserves the time orientation as it maps a future directed vector to a future directed vector. The Cauchy Busemann construction is invariant under such isometries: $\psi(\Sigma)$ is a compact Cauchy hypersurface, $d_{\psi(\Sigma)}=d_\Sigma\circ \psi^{-1}$, the sets $\Sigma_k$ of $\psi(\Sigma)$ are the sets $\psi(\Sigma_k)$, thus the future Cauchy Busemann function of $\psi(\Sigma)$ is $b^+\circ \psi^{-1}$, and its calibrated curves are the images under $\psi$ of the $b^+$-calibrated curves. In particular, the image under $\psi$ of the comoving observer through $p$ is a future inextendible calibrated curve of $\psi(\Sigma)$ with velocity $v$ at $p$, so $v\in C_p$. Thus, in de Sitter spacetime, $C_p$ is the whole set of future directed unit timelike vectors at $p$.

Under future timelike completeness alone the splitting conclusion is lost, while a weaker equivalence survives. For a fixed compact Cauchy hypersurface $\Sigma$, condition (c) is equivalent to the $C^1$ regularity of $b^+$ on $\textrm{dom}\, b^+$ (Proposition \ref{propcrease}), and to the existence, through every point of $\textrm{dom}\, b^+$, of a $b^+$-calibrated timelike line, which is future complete. Indeed, (c) implies the existence of these lines by Theorem \ref{thmcone} and, conversely, every point of such a line is a point $x(t)$, $t>0$, of a future inextendible $b^+$-calibrated curve $x$, at which $b^+$ is differentiable by Proposition \ref{propcrease}. These lines need not be past complete, and (c) does not imply that any part of the spacetime is static, even under TCC. Consider the warped product of Example \ref{exwarp} with $n=3$, $I=(0,+\infty)$, $a(t)=t^{2/3}$ and $h$ flat. It is future timelike complete, and the Ricci tensor formulas given there show that $Ric(X,X)>0$ for every timelike vector $X$, so TCC holds and no open subset is isometric to a static product, in which $Ric(\p_t,\p_t)=0$. Yet, for $\Sigma=S_{t_0}$, $b^+=t-t_0$ is smooth and its calibrated curves are the comoving observers, which are past incomplete timelike lines.
\end{remark}

Proposition \ref{prop39} makes it natural to consider examples in which the expansion of a smooth compact Cauchy hypersurface $\Sigma$ is non-negative on $\textrm{dom}\, b^+$ and non-positive on $\textrm{dom}\, b^-$. For this reason, with some abuse, $\textrm{dom}\, b^+$ might be called the {\em cosmologically expanding region}. A much stronger statement holds: by Theorem \ref{thmmean} all the horospheres $B^+_c$ have non-negative future mean curvature in the support sense, which is what really justifies the terminology.

\begin{theorem} \label{thmwall}
Assume that $(M,g)$ is timelike complete, satisfies TCC and does not split. There is a connected component of $\p\, \textrm{dom}\, b^+$ which contains a lightlike line. A connected component of $\p\, \textrm{dom}\, b^+$ contains a lightlike line if and only if it intersects every Cauchy hypersurface. Every other connected component is contained in $I^+(S)$ for some Cauchy hypersurface $S$, every Cauchy time function attains a minimum on it, and its chronological future is contained in the cosmological wall $M\backslash (\textrm{dom}\, b^+\cup \textrm{dom}\, b^-)$. The analogous statements hold for $\p\, \textrm{dom}\, b^-$.
\end{theorem}

\begin{proof}
Let $t$ be a Cauchy temporal function \cite{bernal04}, $S:=t^{-1}(0)$, and let $\phi$ be the flow of the future directed timelike vector field $\nabla t/g(\nabla t,\nabla t)$, so that $t(\phi(s,x))=s$ and $\phi: \mathbb{R}\times S\to M$ is a homeomorphism. Thus the Cauchy hypersurfaces $S_s:=t^{-1}(s)=\phi(\{s\}\times S)$ are connected, as $M$ is, and the sets $t^{-1}([\alpha,\beta])=\phi([\alpha,\beta]\times S)$ are compact. Let $N:=\p\, \textrm{dom}\, b^+$. By Prop.\ \ref{cmol} and Corollary \ref{bbwe}, $\textrm{dom}\, b^+$ intersects every $S_s$ without containing it, so the non-empty open proper subset $S_s\cap \textrm{dom}\, b^+$ of the connected space $S_s$ is not closed in $S_s$, and a point of $S_s$ in its closure but not in it belongs to $N$. Thus $N\cap S_s\ne \emptyset$ for every $s\in \mathbb{R}$.

Let $C$ be either $N$ or a connected component of $N$, and suppose that $C\cap S_{-k}\ne \emptyset$ for every $k\in \mathbb{N}$. We show that $C$ contains a lightlike line. Let $p_k\in C\cap S_{-k}$ and let $\sigma_k$ be a future lightlike ray starting from $p_k$ and contained in $N$ (Prop.\ \ref{propbdry}). Being connected, $\sigma_k$ is contained in $C$, and it intersects $\Sigma$ at a point $r_k\in C\cap \Sigma$. Since $C$ is closed (a component of a closed set is closed) and $\Sigma$ is compact, we can assume $r_k\to r\in C\cap \Sigma$. Let $h$ be a complete Riemannian metric. Parametrizing $\sigma_k$ with respect to $h$-arc length with $\sigma_k(0)=r_k$, the domains are $[-a_k,+\infty)$ with $a_k\to +\infty$, as $p_k$ escapes every compact set. Extending each $\sigma_k$ beyond $p_k$ to an inextendible causal curve, by the limit curve theorem \cite[Thm.\ 2.51]{minguzzi18b} a subsequence converges uniformly on compact subsets to an inextendible causal curve $\sigma$ through $r$. On every compact interval, for large $k$, $\sigma$ is the uniform limit of segments of the curves $\sigma_k$, which are maximizing, being causal curves contained in the achronal set $N$, thus $\sigma$ is maximizing \cite[Cor.\ 2.59]{minguzzi18b}. It is contained in the closed set $C$ and, as $C$ is achronal, $\sigma$ is a lightlike line.

Applied to $C=N$, this argument gives a lightlike line $\sigma\subset N$, and the connected component of $N$ containing $\sigma$ contains a lightlike line. A connected component of $N$ which contains a lightlike line intersects every Cauchy hypersurface, since the line is inextendible. Conversely, a connected component of $N$ which intersects every Cauchy hypersurface intersects every $S_{-k}$, and thus contains a lightlike line by the previous paragraph. Let $D'$ be a connected component of $N$ which does not intersect a Cauchy hypersurface $S'$. If there were $p\in D'\cap I^-(S')$, the future lightlike ray starting from $p$ and contained in $N$ would be contained in $D'$, being connected, and would intersect $S'$. Thus $D'\subset I^+(S')$. Let $\tau$ be a Cauchy time function. As $\tau$ is strictly increasing along causal curves and $S'$ is compact, $\tau>\min_{S'}\tau$ on $I^+(S')$, so $m:=\inf_{D'}\tau$ is finite. The closed set $D'\cap \tau^{-1}([m,m+1])$ is non-empty, and it is compact, as it is contained in $J^+(\tau^{-1}(m))\cap J^-(\tau^{-1}(m+1))$, where $t$ is bounded because it is bounded on the compact Cauchy hypersurfaces $\tau^{-1}(m)$ and $\tau^{-1}(m+1)$. Thus $\tau$ attains its minimum on $D'$.

It remains to show that $I^+(D')$ does not intersect $U^\pm:=\textrm{dom}\, b^\pm$. As $U^+$ is an open past set (Prop.\ \ref{cmol} and Theorem \ref{xxhg}(ii)) and $D'\cap U^+=\emptyset$, we have $J^+(D')\cap U^+=\emptyset$. We need two properties of $D'$.

{\em (i) Every $q_0\in D'$ has a neighborhood $Q$ such that $Q\backslash \overline{U^+}\subset I^+(D')$.} As $N$ is a topological hypersurface, hence locally connected, $D'$ is open in $N$ and there is a neighborhood $O$ of $q_0$ such that $O\cap N\subset D'$. By strong causality there are $z_\pm$ such that $Q:=I^+(z_-)\cap I^-(z_+)$ satisfies $q_0\in Q\subset O$, and $z_-\in I^-(q_0)\subset U^+$, as $U^+$ is an open past set. If $w\in Q\backslash \overline{U^+}$, a timelike curve from $z_-$ to $w$ meets $N$ at a point $q\ne w$ of $I^+(z_-)\cap J^-(w)\subset Q$, so that $q\in D'$ and $w\in I^+(q)$.

{\em (ii) $\p I^+(D')\subset D'$.} Let $z\in \p I^+(D')$, $z=\lim y_n$ with $y_n\in I^+(x_n)$, $x_n\in D'$. As $t(x_n)<t(y_n)$ is bounded and $t$ attains a minimum on $D'$, the points $x_n$ lie in a compact subset of the closed set $D'$, so we can assume $x_n\to x\in D'$, and $z\in J^+(x)$ as the causal relation of a globally hyperbolic spacetime is closed. Let $c:[0,1]\to M$ be a causal curve from $x$ to $z$. As $z\notin I^+(D')$, no point of $c$ lies in the future set $I^+(D')$, and as $c\subset J^+(D')$, no point of $c$ lies in $U^+$. Let $s_1:=\sup\{s: c([0,s])\subset N\}$. The connected set $c([0,s_1])$ is contained in $N$, hence in $D'$. If $s_1<1$, there would be points of $c$ arbitrarily close to $c(s_1)$ which are not in $N$, hence not in $\overline{U^+}$, and which would then belong to $I^+(D')$ by (i), a contradiction. Thus $z=c(1)\in D'$.

{\em Conclusion.} Suppose that $y\in I^+(D')\cap U^-$. By the time dual of Proposition \ref{bber}, there is a past inextendible causal curve $\eta$ ending at $y$ and contained in $U^-$. As $y\in I^+(S')$, $\eta$ intersects $S'$, which does not intersect $I^+(D')\subset I^+(S')$, thus $\eta$ intersects $\p I^+(D')$, hence $D'$ by (ii). But $D'\subset N$ and $N\cap U^-=\emptyset$, since $U^-$ is open and, by Theorem \ref{notr}, disjoint from $U^+$, a contradiction.
\end{proof}

In particular, since the chronological future of a non-empty set is a non-empty open set, if the cosmological wall has empty interior then every connected component of $\p\, \textrm{dom}\, b^\pm$ contains a lightlike line and intersects every Cauchy hypersurface.

\subsection{The cosmological wall}

Let $(M,g)$ be timelike complete, and assume that it satisfies TCC and does not split, so that $U^+:=\textrm{dom}\, b^+$ and $U^-:=\textrm{dom}\, b^-$ are disjoint (Theorem \ref{notr}), and let $N^\pm:=\p U^\pm$ be the future and past cosmological horizons and $W=M\backslash(U^+\cup U^-)$ the cosmological wall.

\begin{proposition} \label{propcommon}
Assume that $(M,g)$ is timelike complete, satisfies TCC and does not split. The wall $W$ is closed and causally convex, and $N^+\cup N^-\subset W$. Every point of $N^+\cap N^-$ lies on a lightlike line contained in $N^+\cap N^-$. No null completeness assumption is needed.
\end{proposition}

\begin{proof}
If $x\le y\le z$ with $x,z\in W$, then $y\in U^+$ would give $x\in U^+$, and $y\in U^-$ would give $z\in U^-$, since $U^+$ is a past set and $U^-$ a future set. $N^+$ does not intersect the open set $U^+$, and it does not intersect $U^-$ since a point of $U^-\cap \p U^+$ would have a neighborhood in $U^-$ meeting $U^+$. Thus $N^+\subset W$, and similarly $N^-\subset W$.

Let $p\in N^+\cap N^-$ and let $\sigma$ be a future lightlike ray starting from $p$ and contained in $N^+$ (Proposition \ref{propbdry}). The set $\overline{U^-}$ is a future set, thus $\sigma\subset \overline{U^-}$, and since $\sigma\subset N^+$ does not intersect $U^-$, $\sigma\subset N^-$. Dually, a past lightlike ray ending at $p$ and contained in $N^-$ is contained in $N^+$. The join of the two rays is an inextendible causal curve contained in the achronal set $N^+$, hence an achronal lightlike geodesic, namely a lightlike line, contained in $N^+\cap N^-$.
\end{proof}

\begin{theorem}[Rigidity of meeting horizons] \label{thmrigid}
Assume that $(M,g)$ is timelike complete, satisfies TCC, does not split and is null geodesically complete. If a connected component $D^+$ of $N^+$ intersects a connected component $D^-$ of $N^-$ then $D^+=D^-$, and this set is a smooth totally geodesic null hypersurface generated by lightlike lines.
\end{theorem}

\begin{proof}
By null completeness the generators of $N^+$ are future complete, thus by the null convergence condition (a consequence of TCC) $N^+$ has non-negative null expansion in the support sense, with support null second fundamental forms locally bounded from below \cite[Lemma IV.2]{galloway00}. Dually $N^-$ has non-positive null expansion in the support sense. Let $x\in D^+\cap D^-$. In a small neighborhood $O$ of $x$ the sets $U^+\cap O$ and $U^-\cap O$ lie, respectively, to the past of $N^+$ and to the future of $N^-$, and as they are disjoint, $N^-\cap O$ lies to the future of $N^+\cap O$ (in the non-strict sense). By the null maximum principle \cite[Thm.\ III.4]{galloway00}, $N^+$ and $N^-$ coincide near $x$ and form there a smooth null hypersurface with vanishing expansion; by the Raychaudhuri equation and the null convergence condition its shear also vanishes, so it is totally geodesic. Thus $D^+\cap D^-$ is open in $D^+$ and in $D^-$. It is also closed, thus $D^+=D^+\cap D^-=D^-$. Its generators are lightlike lines by Proposition \ref{propcommon}.
\end{proof}

\begin{theorem}[Thin walls] \label{thmthin}
Assume that $(M,g)$ is timelike complete, satisfies TCC and does not split. If the cosmological wall has empty interior then
\begin{itemize}
\item[(a)] $W=N^+=N^-$;
\item[(b)] every point of $W$ lies on a lightlike line contained in $W$;
\item[(c)] for every Cauchy hypersurface $S$, $S\backslash W$ is the disjoint union of the non-empty open sets $S\cap U^+$ and $S\cap U^-$, so $S\cap W$ separates $S$;
\item[(d)] if the spacetime is null geodesically complete, every connected component of $W$ is a smooth totally geodesic null hypersurface generated by lightlike lines.
\end{itemize}
\end{theorem}

\begin{proof}
(a). Since $W$ has empty interior, $M=\overline{U^+\cup U^-}=\overline{U^+}\cup \overline{U^-}$, thus $W\subset N^+\cup N^-$. Let $p\in N^+$ and suppose $p\notin \overline{U^-}$. Let $O$ be a neighborhood of $p$ disjoint from $U^-$. The open set $I^+(p)$ does not intersect $\overline{U^+}$, since $I^-(\overline{U^+})=I^-(U^+)\subset U^+$ and $p\notin U^+$. Thus $I^+(p)\cap O$ is a non-empty open subset of $W$, a contradiction. Therefore $N^+\subset \overline{U^-}\backslash U^-=N^-$, and dually $N^-\subset N^+$. Together with Proposition \ref{propcommon} we get $W=N^+=N^-$.

(b) follows from (a) and Proposition \ref{propcommon}, and (c) from (a) and Proposition \ref{cmol}. (d) follows from (a) and Theorem \ref{thmrigid}.
\end{proof}

Theorem \ref{thmthin} justifies the thin wall structure depicted in Figures \ref{figstruct} and \ref{figeg}. We stress that it does not prove that the wall is thin; whether the wall must have empty interior is an open problem.

\begin{remark} \label{remwallcurv}
Let $W_0$ be a connected component of the wall which is a smooth totally geodesic null hypersurface, as in Theorem \ref{thmthin}(d), with future generator $k$. The Raychaudhuri equation gives $Ric(k,k)=0$ on $W_0$, and testing TCC on the timelike vectors $k+sX+cs^2\ell$, where $X$ is a screen vector, $\ell$ a null vector with $g(k,\ell)=-1$, $c>g(X,X)/2$ and $s\ne 0$ small of both signs, gives $Ric(k,X)=0$. Thus $Ric(k,\cdot)=\lambda\, g(k,\cdot)$ on $W_0$ and, testing on $k+\epsilon \ell$, $\lambda\le 0$. The geometry of such walls and of their compact sections deserves further investigation.
\end{remark}

\begin{remark} \label{rem41}
Concerning the causal boundary of the future cosmological development, we have established the following: every TIP contains a TIP of the form $I^-(\sigma)\subset \textrm{dom}\, b^+$ with $\sigma$ a future $\Sigma$-ray (Theorem \ref{cbdc}), every minimal TIP is contained in $\textrm{dom}\, b^+$ (Corollary \ref{cormintip}), and $\textrm{dom}\, b^+$ is a union of TIPs of $b^+$-calibrated rays (Proposition \ref{coruni}).
\end{remark}

The next result shows that the past ideal points of $\textrm{dom}\, b^+$, regarded as ideal points of $M$, are never minimal.

\begin{proposition} \label{proppastnull}
Assume that $(M,g)$ is timelike complete, satisfies TCC and does not split. Let $\gamma$ be a past inextendible timelike curve of the spacetime $\textrm{dom}\, b^+$. Then $\gamma$ is past inextendible in $M$, and the TIF $I^+(\gamma)$ of $M$ properly contains a TIF $I^+(\eta)$, where $\eta$ is a $b^-$-calibrated past $\Sigma$-ray, such that $I^+(\eta)\cap \textrm{dom}\, b^+=\emptyset$. In particular, $I^+(\gamma)$ is not a minimal TIF.
\end{proposition}

\begin{proof}
A past endpoint of $\gamma$ in $M$ would belong to $J^-(\textrm{dom}\, b^+)\subset \textrm{dom}\, b^+$ (Prop.\ \ref{cmol}), thus $\gamma$ is past inextendible in $M$. By the time duals of Theorem \ref{cbdc} and Proposition \ref{cmol}, the TIF $I^+(\gamma)$ contains a TIF $I^+(\eta)\subset \textrm{dom}\, b^-$, where $\eta$ is a past $\Sigma$-ray, which is $b^-$-calibrated since $b^-=b^\eta$ over $\eta$. As $\textrm{dom}\, b^+\cap \textrm{dom}\, b^-=\emptyset$ (Theorem \ref{notr}), $I^+(\eta)$ does not meet $\textrm{dom}\, b^+$, while $\gamma\subset I^+(\gamma)\cap \textrm{dom}\, b^+$. Thus $I^+(\eta)\subsetneq I^+(\gamma)$.
\end{proof}

\begin{remark} \label{remcausbdry}
The assumptions of TCC and non-splitting enter the proof only through the disjointness $\textrm{dom}\, b^+\cap \textrm{dom}\, b^-=\emptyset$, which by Theorem \ref{notr} follows from the mere absence of timelike lines. Thus the conclusion of Proposition \ref{proppastnull} holds in every timelike complete spacetime without timelike lines and, contrapositively, if $(M,g)$ is timelike complete and some TIF generated by a past inextendible timelike curve contained in $\textrm{dom}\, b^+$ is minimal, then $(M,g)$ contains a timelike line. Following Wald and Yip, Galloway and Vega call the past causal boundary spacelike if no TIF is properly contained in another \cite[Sect.\ 5.3]{galloway17}, namely if all TIFs are minimal. Thus we recover, for timelike complete spacetimes, their result that a spacelike past (or, dually, future) causal boundary implies the existence of a timelike line, and hence Bartnik's conjecture under this assumption \cite[Prop.\ 5.9, Thm.\ 5.10]{galloway17}. (This also follows from Corollary \ref{cormintip}(b): if all TIFs are minimal, the min-min condition holds, as every minimal TIP, being an open past set, contains past inextendible timelike curves, whose TIFs meet it.) Proposition \ref{proppastnull} places their result in a structural context: the obstruction to splitting is not the mere presence of some non-minimal TIF, but the systematic non-minimality of all the TIFs generated within the future cosmological development, each of which properly contains the TIF of a $b^-$-calibrated ray lying entirely outside $\textrm{dom}\, b^+$. This is consistent with the example of Ehrlich and Galloway \cite{ehrlich90}, which contains no timelike line and whose past causal boundary is not spacelike \cite[Sect.\ 5.3]{galloway17}.
\end{remark}

Since $\textrm{dom}\, b^+$ is a past set, hence causally convex, the TIFs of the spacetime $\textrm{dom}\, b^+$ are the sets $I^+(\gamma)\cap \textrm{dom}\, b^+$, with $\gamma$ as in Proposition \ref{proppastnull}. Proposition \ref{proppastnull} suggests that the past causal boundary of $\textrm{dom}\, b^+$ is null: every past ideal point of $\textrm{dom}\, b^+$, regarded as an ideal point of $M$, lies in the causal past of an ideal point of $M$ generated by a $b^-$-calibrated ray, whose TIF does not meet $\textrm{dom}\, b^+$. This is the behavior of past null infinity in Minkowski spacetime, whose TIFs are never minimal, whereas the TIFs of a spacetime with spacelike past conformal boundary, such as de Sitter spacetime, are all minimal. We stress that this is an interpretation: we do not construct a boundary of $\textrm{dom}\, b^+$ endowed with a causal structure, for which a causal character could be established.

Although the past set $\textrm{dom}\, b^+$ regarded as a spacetime of its own is globally hyperbolic, it is not clear  if the level sets $B_c^+$ need to be its Cauchy hypersurfaces. By Proposition \ref{propdomcauchy} this is the case iff $b^+\to -\infty$ along every past inextendible causal curve starting from $\textrm{dom}\, b^+$. Note that graphically the condition that $B_c^+$ is Cauchy for $\textrm{dom}\, b^+$ coincides with its level sets `dying'  on the `spacelike boundary' rather than on the null boundary below.

The simplest causal structure that respects all the constraints seems to be that of Figure \ref{figeg}, which happens to be exactly that considered in the example by Ehrlich and Galloway \cite{ehrlich90}. One might wonder if a different conformal factor can be chosen to disprove the conjecture. The answer is negative. As commented in \cite{ehrlich90}, in 2 spacetime dimensions the TCC is equivalent to the non-positivity of the timelike sectional curvature (with the sign convention of Section \ref{secconv}) and with this strengthened assumption the Bartnik statement is correct.

\begin{figure}[ht]
\begin{center}
 \includegraphics[width=12cm]{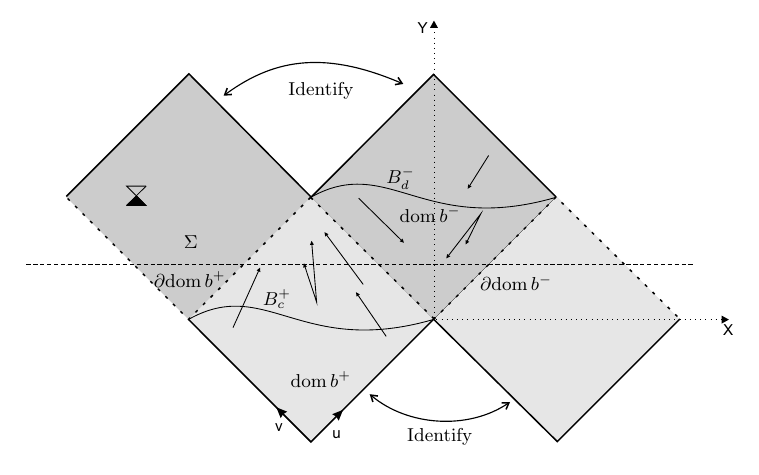}
\end{center} \caption{A simple structure where the cosmological developments are two squares and the  horospheres are Cauchy in the respective cosmological developments. Squares of the same shade are identified.} \label{figeg}
\end{figure}

In higher dimensions, the Bartnik statement is also known to hold under non-positive timelike sectional curvature \cite{ehrlich90}, and, in the synthetic setting of Lorentzian length spaces, under non-negative timelike curvature, the synthetic counterpart of this condition \cite{flores24}. Since the Ricci tensor is an average of sectional curvatures, a counterexample cannot satisfy such a sectional curvature bound: its sectional curvatures must have different signs depending on the timelike plane.

\section{Conclusions} \label{secconc}

We have studied the geometry of future timelike complete spacetimes admitting a compact Cauchy hypersurface, a class which allows for an initial singularity, excludes a final one, and disregards the incompleteness expected inside black holes. The Cauchy Busemann function $b^+$ of a compact Cauchy hypersurface $\Sigma$ determines the future cosmological development $\textrm{dom}\, b^+$ and the future cosmological horizon $\p\, \textrm{dom}\, b^+$, which do not depend on $\Sigma$ (Theorem \ref{xxxt}). The development is an open past set which meets every Cauchy hypersurface, and the horizon is a $C^0$ future null hypersurface which, much as the event horizon of a black hole, can be crossed by future directed causal curves in the outward direction only.

The main tool is the product trick, which exhibits $b^+$ as the graph function of a Cauchy horizon in $M\times \mathbb{R}$. It shows that $b^+$ is continuous, and locally Lipschitz and locally semi-concave on the development, and that through every point of the development passes a calibrated timelike ray, an observer who ages maximally into the far future. The function $b^+$ is differentiable exactly at the points from which a unique calibrated ray starts (Proposition \ref{propcrease}), and its non-differentiability points form an $\mathcal{H}^{n+1}$-null set (Theorem \ref{thmfine}). It admits a foliation-free characterization as a minimal asymptotic clock defect (Theorem \ref{thmintr}), and the distances to the level sets of the distance foliation are locally eventually equi-Lipschitz on the development (Theorem \ref{thmequi}), a property which fails for general Cauchy foliations \cite{galloway26}. For spacetimes with an initial singularity, namely with regular cosmological time, the construction becomes canonical: the normalized distances to the level sets of the cosmological time converge to a Cauchy Busemann function which does not depend on any choice, and its calibrated rays define a canonical calibrating semi-flow on a full measure subset of the development (Proposition \ref{propcosm}).

Perhaps the most important result of this work is the following: in cosmological spacetimes with an initial singularity there are really two natural frames. The first follows the gradient of the cosmological time function, and describes the motion of matter and radiation from the Big Bang, including dark matter; the second is given by the canonical calibrating semi-flow of the Cauchy Busemann function, and describes a family of privileged observers in undisturbed free fall to the infinite future. In warped products, such as the Friedmann--Lema\^itre--Robertson--Walker spacetimes (Example \ref{exwarp}), the two frames coincide, but the distinction between them should be kept in mind whenever the cosmological principle is broken. The possibility that this two-frame structure could lead to new mechanisms for explaining the Hubble tension will be considered in future work.

Returning to the description of the mathematical content, under TCC the level sets of $b^+$ are mean convex in the support sense (Theorem \ref{thmmean}). Combined with Galloway's geometric maximum principle (Theorem \ref{thmagh}), this gives a localized rigidity version of Hawking's singularity theorem: unless its future is static, every smooth compact Cauchy hypersurface has positive future mean curvature somewhere in the future development (Theorem \ref{thmhawk}). In the companion paper \cite{minguzzi26g}, the product trick also turns the area theorem for horizons into a volume theorem for the semi-flow of the calibrated rays, which under TCC gives the monotonicity of the volume measured by the calibrated observers, and, under TCC, the future cosmological horizon is shown to satisfy an area theorem.

These results apply to Bartnik's conjecture, which concerns timelike complete spacetimes satisfying TCC. If the future and past developments intersect, the spacetime contains a timelike line and hence splits (Theorem \ref{notr}), and this is the case if $b^+$ is $C^1$ on its finiteness domain (Theorem \ref{thmcone}). Bartnik's conjecture is thus equivalent to the statement that, for timelike complete spacetimes with compact Cauchy hypersurfaces satisfying TCC, the two developments always intersect, and also to the boundedness of the scalar defect $2k-d(\Sigma_{-k},\Sigma_k)$ (Theorem \ref{thmscal}), or of $b^+-d_\Sigma$ along some past inextendible curve (Theorem \ref{thmonesided}). The latter, one-sided, criterion requires only future timelike completeness to produce a $b^+$-calibrated timelike line, possibly past incomplete, and applies also to spacetimes with an initial singularity: every future timelike complete spacetime with compact Cauchy hypersurfaces and regular cosmological time contains $b^{\textnormal{\textsc{m}}}$-calibrated timelike lines, eternal maximally aging observers emanating from the initial singularity (Corollary \ref{corbigbang}). In a counterexample to Bartnik's conjecture, instead, $b^+-d_\Sigma$ diverges at past infinity (Corollary \ref{corpastdiv}). The localized rigidity theorem contains Bartnik's CMC splitting theorem (Corollary \ref{corhawk}), and the intersection criterion recovers several further partial results: the no observer horizon and no null ray conditions (Corollary \ref{cornfoh}), the ray-to-ray condition of Galloway (Corollary \ref{corray}), the max--min condition of Galloway and Vega (Remark \ref{remminmax}), their condition on past and future Cauchy horospheres, which turns out to be equivalent to the intersection of the developments (Corollary \ref{corhorohoro}), the min-min condition of Garc\'ia-Heveling (Corollary \ref{cormintip}), and the spacelike causal boundary condition of Galloway and Vega (Remark \ref{remcausbdry}).

In a counterexample, the two developments are separated by a non-empty, closed and causally convex cosmological wall, and every smooth compact Cauchy hypersurface is expanding somewhere in the future development and contracting somewhere in the past one. Some component of the future horizon contains a lightlike line, while the components which do not contain one lie to the future of a Cauchy hypersurface and have their chronological future inside the wall. The points common to the two horizons lie on lightlike lines. If the wall is thin, it coincides with both horizons, separates every Cauchy hypersurface and is ruled by lightlike lines; under null completeness it is a smooth totally geodesic null hypersurface. The simplest realization of this causal structure is the Ehrlich--Galloway example, which, however, violates TCC. Should a counterexample exist, the wall and its horizons would be its most distinctive physical feature: an observer can leave the future development and can enter the past one, but can never return to the former or exit the latter.

Several directions remain open. Beyond warped products, it would be interesting to determine the cosmological developments and horizons of explicit models. A new work will be devoted to exploring the physical role of the canonical calibrating semi-flow in spacetimes with an initial singularity. Concerning Bartnik's conjecture, one may look for a barrier from the future inside the future development, namely a smooth spacelike hypersurface with non-positive future mean curvature touching a horosphere from the future, together with a continuation mechanism leading to a compact maximal Cauchy hypersurface, as in Theorem \ref{thmhawk} (Remark \ref{remscope}). It would also be important to prove that the wall is thin and to analyze the geometry of its compact sections or, in the opposite direction, to construct explicit metrics which realize the causal structure of a thin wall and satisfy TCC.

\section*{Acknowledgments}


The main results of this work—such as the existence of the cosmological wall in potential counterexamples—and the main novel ideas—such as the use of the product trick to regard the Cauchy Busemann functions as Cauchy horizons in a product \(M\times \mathbb{R}\), thereby obtaining results on regularity—were all obtained and elaborated without the use of AI.
For the final writing of this paper we had the assistance of AI tools from Anthropic (Claude Opus 5.5) and OpenAI (GPT-6 Astra). This work was supported in part by a grant of access to OpenAI models through the ChatGPT for Academic Researchers program.


\begin{thebibliography}{10}

\bibitem{alberti92}
{G.~Alberti, L.~Ambrosio, and P.~Cannarsa.
\newblock On the singularities of convex functions.
\newblock {\em Manuscripta Math.}, 76:421--435, 1992.}

\bibitem{amini18}
R.~Amini, M.~Sharifzadeh, and Y.~Bahrampour.
\newblock Bartnik's splitting conjecture and {L}orentzian {B}usemann function.
\newblock {\em Class. Quantum Grav.}, 35:105009, 2018.

\bibitem{andersson98}
L.~Andersson, G.~J. Galloway, and R.~Howard.
\newblock The cosmological time function.
\newblock {\em Class. Quantum Grav.}, 15:309--322, 1998.

\bibitem{andersson98b}
L.~Andersson, G.~J. Galloway, and R.~Howard.
\newblock A strong maximum principle for weak solutions of quasi-linear
  elliptic equations with applications to {L}orentzian and {R}iemannian
  geometry.
\newblock {\em Comm. Pure Appl. Math.}, 51:581--624, 1998.

\bibitem{bartnik88}
R.~Bartnik.
\newblock Remarks on cosmological spacetimes and constant mean curvature
  surfaces.
\newblock {\em Commun. Math. Phys.}, 117:615--624, 1988.

\bibitem{beem96}
J.~K. Beem, P.~E. Ehrlich, and K.~L. Easley.
\newblock {\em Global Lorentzian Geometry}.
\newblock Marcel {D}ekker {I}nc., New York, 1996.

\bibitem{beem98}
J.~K. Beem and A.~Kr{\'o}lak.
\newblock Cauchy horizon end points and differentiability.
\newblock {\em J. Math. Phys.}, 39:6001--6010, 1998.

\bibitem{bernal04}
A.~N. Bernal and M.~{S\'a}nchez.
\newblock Smoothness of time functions and the metric splitting of globally
  hyperbolic spacetimes.
\newblock {\em Commun. Math. Phys.}, 257:43--50, 2005.

\bibitem{braun24}
M.~Braun, N.~Gigli, R.~J. McCann, A.~Ohanyan, and C.~S{\"a}mann.
\newblock An elliptic proof of the splitting theorems from {L}orentzian
  geometry.
\newblock {\em Comm. Amer. Math. Soc.}, 6:1087--1112, 2026.
\newblock {arXiv:}2410.12632.

\bibitem{budic78}
R.~Budic, J.~Isenberg, L.~Lindblom, and P.~B. Yasskin.
\newblock On the determination of {C}auchy surfaces from intrinsic properties.
\newblock {\em Commun. Math. Phys.}, 61:87--95, 1978.

\bibitem{cannarsa17}
{P.~Cannarsa and W.~Cheng.
\newblock Generalized characteristics and {L}ax--{O}leinik operators: global theory.
\newblock {\em Calc. Var. Partial Differential Equations}, 56:125, 2017.}

\bibitem{chrusciel98b}
P.~T. Chru{\'s}ciel.
\newblock A remark on differentiability of {C}auchy horizons.
\newblock {\em Class. Quantum Grav.}, 15:3845--3848, 1998.

\bibitem{chrusciel01}
P.~T. Chru{\'s}ciel, E.~Delay, G.~J. Galloway, and R.~Howard.
\newblock Regularity of horizons and the area theorem.
\newblock {\em Ann. {H}enri {P}oincar\'e}, 2:109--178, 2001.

\bibitem{costa20b}
I.~P.~Costa e~Silva, J.~L. Flores, and J.~Herrera.
\newblock Some remarks on conformal symmetries and {B}artnik's splitting
  conjecture.
\newblock {\em Mediterr. J. Math.}, 17:21, 2020.

\bibitem{ehrlich90}
P.~E. Ehrlich and G.~J. Galloway.
\newblock Timelike lines.
\newblock {\em Class. Quantum Grav.}, 7:297--307, 1990.

\bibitem{eschenburg88}
J.-H. Eschenburg.
\newblock The splitting theorem for space-times with strong energy condition.
\newblock {\em J. Differential Geom.}, 27:477--491, 1988.

\bibitem{eschenburg89}
J.-H. Eschenburg.
\newblock Maximum principle for hypersurfaces.
\newblock {\em Manuscripta Math.}, 64:55--75, 1989.

\bibitem{evans92}
L.~C. Evans and R.~F. Gariepy.
\newblock {\em Measure theory and fine properties of functions}.
\newblock CRC Press, Boca Raton, 1992.

\bibitem{federer69}
H.~Federer.
\newblock {\em Geometric measure theory}.
\newblock Springer-Verlag, Berlin, 1969.

\bibitem{flores24}
J.~L. Flores, J.~Herrera, and D.~A. Solis.
\newblock Low regularity approach to {B}artnik's conjecture, 2024.
\newblock {arXiv:}2412.08967.

\bibitem{galloway84b}
G.~J. Galloway.
\newblock Splitting theorems for spatially closed spacetimes.
\newblock {\em Commun. Math. Phys.}, 96:423--429, 1984.

\bibitem{galloway89}
G.~J. Galloway.
\newblock The {L}orentzian splitting theorem without the completeness
  assumption.
\newblock {\em J. Differential Geom.}, 29:373--387, 1989.

\bibitem{galloway97}
G.~J. Galloway.
\newblock {\em Some rigidity results for spatially closed space-times}, volume
  Mathematics of Gravitation {P}art {I}: {L}orentzian geometry and {E}instein
  equations of {\em Banach Center Publications, volume 41}, pages 21--34.
\newblock Warzawa: Polish Academy of Sciences, Institute of Mathematics, 1997.

\bibitem{galloway00}
G.~J. Galloway.
\newblock Maximum principles for null hypersurfaces and null splitting
  theorems.
\newblock {\em Ann. {H}enri {P}oincar\'e}, 1:543--567, 2000.

\bibitem{galloway19}
G.~J. Galloway.
\newblock Existence of {CMC} {C}auchy surfaces and spacetime splitting.
\newblock {\em Pure and Applied Mathematics Quarterly}, 15:667--682, 2019.

\bibitem{galloway96}
G.~J. Galloway and A.~Horta.
\newblock Regularity of {L}orentzian {B}usemann functions.
\newblock {\em Trans {A}mer. {M}ath. {S}oc.}, 348:2063--2084, 1996.

\bibitem{galloway18c}
G.~J. Galloway and E.~Ling.
\newblock Existence of {CMC} {C}auchy surfaces from a spacetime curvature
  condition.
\newblock {\em Gen. Relativ. Gravit.}, 50:108, 2018.

\bibitem{galloway26}
G.~J. Galloway, R.~J. McCann, and A.~Ohanyan.
\newblock Failure of local equi-{L}ipschitzness for families of {L}orentz
  distances to {C}auchy surface foliations.
\newblock {\em Gen. Relativ. Gravit.}, 58:105, 2026.

\bibitem{galloway25}
G.~J. Galloway and L.~Garc{\'i}a-Heveling.
\newblock Mean curvature, singularities, and time functions in cosmology.
\newblock In W.~Barrera et~al., editors, {\em Progress in {L}orentzian
  Geometry}, volume 512 of {\em Springer Proc. Math. Stat.}, pages 145--161.
  Springer, Cham, 2025.
\newblock {arXiv}:2407.04455.

\bibitem{galloway14b}
G.~J. Galloway and C.~Vega.
\newblock Achronal limits, {L}orentzian spheres, and splitting.
\newblock {\em Ann. Henri Poincar{\'e}}, 15:2241--2279, 2014.

\bibitem{galloway17}
G.~J. Galloway and C.~Vega.
\newblock Hausdorff closed limits and rigidity in {L}orentzian geometry.
\newblock {\em Ann. Henri Poincar{\'e}}, 18:3399--3426, 2017.

\bibitem{heveling25}
L.~Garc{\'i}a-Heveling.
\newblock Minimal elements of the causal boundary with applications to
  spacetime splitting.
\newblock {\em Mediterr. J. Math.}, 22:63, 2025.

\bibitem{geroch72}
R.~Geroch, E.~H. Kronheimer, and R.~Penrose.
\newblock Ideal points in spacetime.
\newblock {\em Proc. {R}oy. {S}oc. {L}ond. {A}}, 327:545--567, 1972.

\bibitem{gilbarg01}
D.~Gilbarg and N.~S. Trudinger.
\newblock {\em Elliptic partial differential equations of second order}.
\newblock Classics in Mathematics. Springer-Verlag, Berlin, 2001.

\bibitem{hawking67}
S.~W. Hawking.
\newblock The occurrence of singularities in cosmology. {III}. {C}ausality and
  singularities.
\newblock {\em Proc. R. Soc. Lond. A}, 300(1461):187--201, 1967.

\bibitem{hawking73}
S.~W. Hawking and G.~F.~R. Ellis.
\newblock {\em The Large Scale Structure of Space-Time}.
\newblock Cambridge University Press, Cambridge, 1973.

\bibitem{lesourd18b}
M.~Lesourd.
\newblock Bartnik splitting and cosmological singularities, 2018.
\newblock {arXiv:}1809.02071.

\bibitem{ling25}
E.~Ling and A.~Ohanyan.
\newblock Vacuum cosmological spacetimes without {CMC} {C}auchy surfaces.
\newblock In W.~Barrera et~al., editors, {\em Progress in {L}orentzian
  Geometry}, volume 512 of {\em Springer Proc. Math. Stat.}, page 231.
  Springer, Cham, 2025.
\newblock {arXiv:}2407.21524.

\bibitem{minguzzi14d}
E.~Minguzzi.
\newblock Area theorem and smoothness of compact {C}auchy horizons.
\newblock {\em Commun. Math. Phys.}, 339:57--98, 2015.
\newblock {arXiv}:1406.5919.

\bibitem{minguzzi17}
E.~Minguzzi.
\newblock Causality theory for closed cone structures with applications.
\newblock {\em Rev. Math. Phys.}, 31:1930001, 2019.
\newblock {arXiv}:1709.06494.

\bibitem{minguzzi18b}
E.~Minguzzi.
\newblock Lorentzian causality theory.
\newblock {\em Living Rev. Relativ.}, 22:3, 2019.

\bibitem{minguzzi19d}
E.~Minguzzi.
\newblock On the regularity of {C}auchy hypersurfaces and temporal functions in
  closed cone structures.
\newblock {\em Rev. Math. Phys.}, 32:2050033, 2020.
\newblock {arXiv:}1909.09352.

\bibitem{newman90}
R.~P. A.~C. Newman.
\newblock A proof of the splitting conjecture of {S}.-{T}. {Y}au.
\newblock {\em J. Diff. Geom.}, 31:163--184, 1990.

\bibitem{minguzzi26f}
{E.~Minguzzi.
\newblock Singularities of the cosmological time and {C}auchy {B}usemann functions: invariance, uniqueness and global timelike propagation, 2026.
\newblock Preprint.}

\bibitem{ohanyan26b}
A.~Ohanyan.
\newblock Cosmological rigidity via $p$-harmonic temporal functions, 2026.
\newblock {arXiv:}2611.01456.

\bibitem{oneill83}
B.~O'Neill.
\newblock {\em Semi-{R}iemannian Geometry with Applications to Relativity}.
\newblock Academic Press, San Diego, 1983.

\bibitem{rockafellar09}
R.~T. Rockafellar and R.~J-B Wets.
\newblock {\em Variational analysis}.
\newblock {Springer-Verlag}, Berlin, 2009.

\bibitem{sharifzadeh09}
M.~Sharifzadeh and Y.~Bahrampour.
\newblock Some results about the level sets of {L}orentzian {B}usemann function
  and {B}artnik's conjecture.
\newblock {\em Commun. Math. Phys.}, 286:389--398, 2009.

\bibitem{treude11}
J.-H. Treude.
\newblock Ricci curvature comparison in {R}iemannian and {L}orentzian geometry.
\newblock Master's thesis, Universit{\"a}t Freiburg, 2011.

\bibitem{treude13}
J.-H. Treude and J.~D.~E. Grant.
\newblock Volume comparison for hypersurfaces in {L}orentzian manifolds and
  singularity theorems.
\newblock {\em Ann. Global Anal. Geom.}, 43:233--251, 2013.

\bibitem{yau82}
S.-T. Yau.
\newblock Problem section.
\newblock In {\em Seminar on Differential Geometry}, volume 102 of {\em Ann. of
  Math. Stud.}, pages 669--706. Princeton University Press, Princeton, 1982.

\bibitem{minguzzi26g}
E.~Minguzzi.
\newblock Volume and area theorems for cosmological spacetimes via the {C}auchy {B}usemann function.
\newblock In preparation, 2026.

\end{thebibliography}

\end{document}